\documentclass[review,preprint,12pt,3p]{elsarticle}

\usepackage{graphicx}

\usepackage{amsmath}
\usepackage{amsthm}
\usepackage{amssymb}
\usepackage{subfig}
\usepackage{tikz}
\usetikzlibrary{arrows,shapes,positioning,shadows,trees, patterns}
\usepackage{epsfig}
\usepackage[ruled,linesnumbered]{algorithm2e}
\usepackage{setspace}
\usepackage{color, soul}
\newtheorem{theorem}{Theorem}
\usepackage{wrapfig}
\usepackage{tabularx}
\usepackage{stfloats}
\usepackage{pgfplots} 
\pgfplotsset{width=2cm,compat=1.5}

\theoremstyle{definition}
\newtheorem{definition}{Definition}

\newcommand{\revison}[1]{\definecolor{sky}{RGB}{255,255,255}\sethlcolor{sky}\hl{#1}}

\journal{Journal of Parallel and Distributed Computing}
\begin{document}

\pgfplotstableread[row sep=\\,col sep=&]{
	noofApp&  HeRAFC&  DRACO&  HFCCI\\ 
	5   &      35 &      35 &      33\\
	5.5 &      38 &      36 &      35\\
	6   &      42 &      36 &      36\\
	6.5 &      45 &      37 &      37\\
	7   &      52 &      42 &      41\\
	7.5 &      56 &      46 &      43\\
	8   &      60 &      55 &      45\\
	8.5 &      61 &      57 &      50\\
	9   &      64 &      58 &      53\\
	9.5 &      65.6 &      59 &      55\\
	10  &      66 &      60 &      56\\
}\oneacompCPUUtilF
\pgfplotstableread[row sep=\\,col sep=&]{
	noofApp &   HeRAFC &   DRACO &   HFCCI\\
        5   &	  26 &	  24 &	  24\\
	5.5 &	  28 &	  27 &	  28\\
	6   &	  36 &	  38 &	  34\\
	6.5 &	  44 &	  47 &	  40\\
	7   &	  50 &	  49 &	  42\\
	7.5 &	  52 &	  50 &	  44\\
	8   &	  58 &	  54 &	  48\\
	8.5 &	  62 &	  57 &	  53\\
	9   &	  64 &	  59 &	  54\\
	9.5 &	  66 &	  61 &	  57\\
	10  &	  67 &	  60 &	  58\\
}\onebcompMemUtilF
\pgfplotstableread[row sep=\\,col sep=&]{
	noofApp &  HeRAFC &  DRACO &  HFCCI\\ 
	5   &	  42 &	  34 &	  33\\
	5.5 &	  46 &	  35 &	  34\\
	6   &	  57 &	  39 &	  36\\
	6.5 &	  63 &	  45 &	  38\\
	7   &	  64 &	  50 &	  38\\
	7.5 &	  67 &	  50 &	  43\\
	8   &	  69 &	  52 &	  45\\
	8.5 &	  71 &	  57 &	  46\\
	9   &	  73 &	  59 &	  48\\
	9.5 &	  77 &	  59 &	  51\\
	10  &	  78 &	  62 &	  53\\ 
}\oneccompBWUtilF
\pgfplotstableread[row sep=\\,col sep=&]{
noofApp &  HeRAFC &  DRACO &  HFCCI \\ 
    5 &	  28 &	  28 &	  27\\
    5.5 & 30 &	  31 &	  34\\
    6 &	  31 &	  34 &	  39\\
    6.5 & 34 &	  40 &	  43\\
    7 &	  39 &	  43 &	  48\\
    7.5 & 41 &	  44 &	  50\\
    8 &	  47 &	  50 &	  59\\
    8.5 & 49 &	  58 &	  62\\
    9 &	  54 &	  59 &	  66\\
    9.5 & 56 &	  59 &	  67\\
    10 &  59 &	  61 &	  68\\
}\twoacompCPUUtilC
\pgfplotstableread[row sep=\\,col sep=&]{
 noofApp&  HeRAFC&  DRACO&  HFCCI\\ 
    5 &	      18 &	  20 &	  21\\
    5.5 &	  22 &	  21 &	  24\\
    6 &	      26 &	  25 &	  30\\
    6.5 &	  30 &	  29 &	  33\\
    7 &	      31 &	  32 &	  35\\
    7.5 &	  32 &	  33 &	  40\\
    8 &	      35 &	  38 &	  50\\
    8.5 &	  42 &	  45 &	  52\\
    9 &	      44 &	  54 &	  58\\
    9.5 &	  47 &	  57 &	  61\\
    10 &	  48 &	  58 &	  62\\ 
}\twobcompMemUtilC
\pgfplotstableread[row sep=\\,col sep=&]{
noofApp& HeRAFC& DRACO& HFCCI\\ 
    5 &	      23 &	  23 &	  26\\
    5.5 &	  23 &	  26 &	  29\\
    6 &	      23 &	  30 &	  31\\
    6.5 &	  25 &	  32 &	  37\\
    7 &	      27 &	  34 &	  40\\
    7.5 &	  27 &	  35 &	  40\\
    8 &	      32 &	  42 &	  45\\
    8.5 &	  32 &	  44 &	  46\\
    9 &	      36 &	  44 &	  47\\
    9.5 &	  37 &	  47 &	  50\\
    10 &	  40 &	  47 &	  52\\
}\twoccompBWUtilC

\pgfplotstableread[row sep=\\,col sep=&]{
noofTasks&  Priority5&  Priority4&  Priority3&  Priority2&  Priority1\\ 
100	&	25.0	&	23.5	&	20		&	17		&	15\\ 
200	&	24.6	&	24		&	19.5	&	16.5	&	15.4\\
300	&	25.2	&	24		&	19.2	&	16		&	15.7\\ 
400	&	24.3	&	23.15	&	18		&	17.8	&	17\\ 
500	&	22.2	&	21.8	&	18.8	&	19		&	18.2\\ 
600	&	21.7	&	21.5	&	19.1	&	19.1	&	18.9\\ 
700	&	20.9	&	20		&	19.3	&	19.3	&	20.6\\ 
800	&	20.1	&	20.1	&	19.5	&	19.8	&	20.6\\ 
900	&	20.9	&	20		&	18.9	&	19.8	&	20.4\\ 
1000&	21.2	&	20.4	&	18.8	&	19.9	&	19.7\\
}\fivePrtyDisttask

\begin{frontmatter}

\title{HeRAFC: Heuristic Resource Allocation and Optimization in MultiFog-Cloud Environment}

\author[label1]{Chinmaya~Kumar~Dehury}
    \ead{chinmaya.dehury@ut.ee}

\author[label2]{Bharadwaj~Veeravalli}
    \ead{elebv@nus.edu.sg}
    
\author[label3]{Satish~Narayana~Srirama \corref{cor}}
    \cortext[cor]{Corresponding author.}
    \ead{satish.srirama@uohyd.ac.in}

\address[label1]{Mobile \& Cloud Lab, Institute of Computer Science, University of Tartu, Tartu 50090, Estonia}

\address[label2]{Department of Electrical and Computer Engineering, National University of Singapore, Singapore.}

\address[label3]{School of Computer and Information Sciences, University of Hyderabad, Gachibowli, Telangana, India.}

\begin{abstract}
\revison{By bringing computing capacity from a remote cloud environment closer to the user, fog computing is introduced. As a result, users can access the services from more nearby computing environments, resulting in better quality of service and lower latency on the network. From the service providers' point of view, this addresses the network latency and congestion issues. This is achieved by deploying the services in cloud and fog computing environments. The responsibility of service providers is to manage the heterogeneous resources available in both computing environments. In recent years, resource management strategies have made it possible to efficiently allocate resources from nearby fog and clouds to users' applications. Unfortunately, these existing resource management strategies fail to give the desired result when the service providers have the opportunity to allocate the resources to the users' application from fog nodes that are at a multi-hop distance from the nearby fog node. The complexity of this resource management problem drastically increases in a MultiFog-Cloud environment. This problem motivates us to revisit and present a novel Heuristic Resource Allocation and Optimization algorithm in a MultiFog-Cloud (HeRAFC) environment. Taking users' application priority, execution time, and communication latency into account, HeRAFC optimizes resource utilization and minimizes cloud load. The proposed algorithm is evaluated and compared with related algorithms. The simulation results show the efficiency of the proposed HeRAFC over other algorithms.}

\end{abstract}



\end{frontmatter}

\section{Introduction}


In a bottom-up approach, the Fog Node (FN) receives the users' requests consisting of multiple tasks. FN may process some tasks of a request and offload others to the cloud environment based on the resource availability at the FN, the Service Level Agreement (SLA), the priority of the user or request, and other Quality of Service (QoS) requirements. However, in a top-down approach, the resource availability in the cloud environment is first checked to execute the tasks. Following this, the rest of the tasks of a user's request are assigned to the available resources at FN \cite{kumar_dehury_efficient_2020, 8014362}. Both approaches differ based on what environment is given higher preference while allocating the resource to the users' request. Resource in this context refers to CPU, memory, storage, processing capability, etc., of the servers and bandwidth and transmission latency of the network.

The current architectures consist of multiple users at the bottom layer, served by single FN at the fog layer (middle layer) and a single cloud at the top of the hierarchy \cite{MASS2019100051, comp1, comp2}. The FNs are consisting of less amount of computing and storage resource. On the other hand, from the user's perspective, the resource availability at the cloud is generally considered as infinite, thanks to the underlined virtualization technology and other resource management tools. Management of the entire lifecycle of such diverse computing environments with different sets of resource capacity, configurations and characteristics is refereed to as heterogeneous resource management \cite{delimitrou_quasar_2014}. The heterogeneous resource management for single FN and single cloud environments is crucial especially while hosting the user's request. Inefficient management may lead to resource unavailability at FN, prohibiting the task from being deployed on FN. Several approaches, such as Integer Linear Programming (ILP) \cite{comp1}, heuristic approach, genetic algorithm, game theory \cite{comp2}, etc., have been investigated recently to optimize resource utilization. This resource management problem becomes more complex when the multiple interconnected FNs are introduced at the middle layer of the hierarchy, as shown in Figure \ref{fig:mfcenv} \cite{IndieFog}. This paper addresses the resource allocation and optimization problem in the above-mentioned multiFog-Cloud environment. 

\begin{figure*}[htbp]
	\centering
	\subfloat[Abstract view of MultiFog-Cloud (MFC) environment.]
	{
		\includegraphics[width=0.35\linewidth]{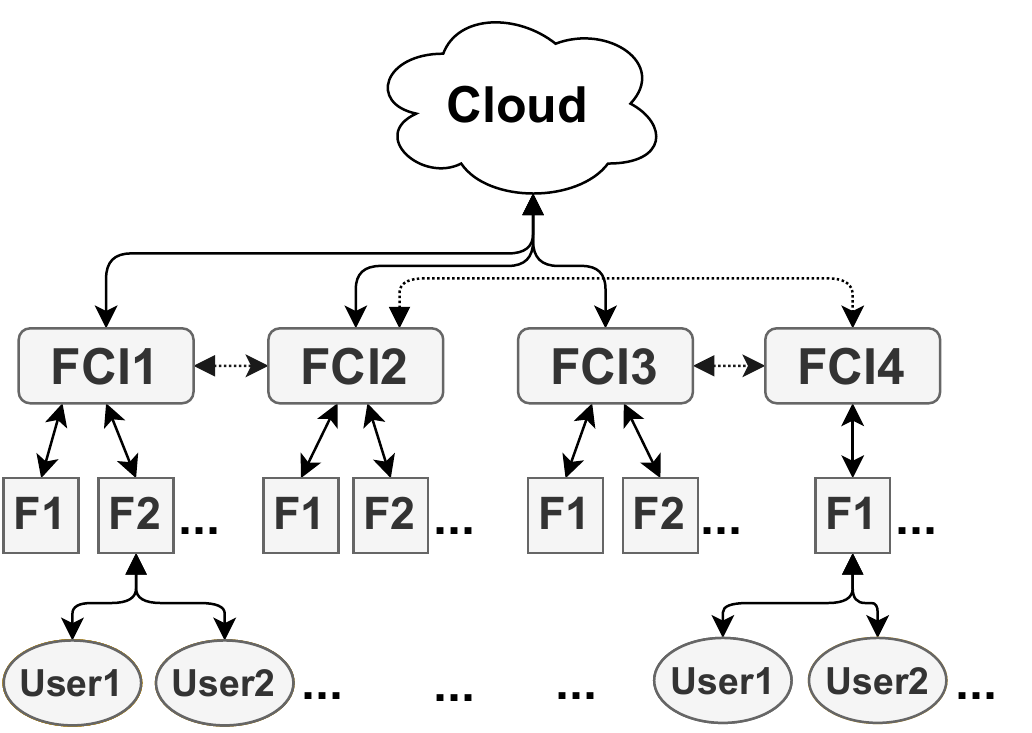}
		\label{fig:mfcenv}
	}\hspace{15mm} 
	\subfloat[An example of user's application (in DAG form).]
	{
		\includegraphics[width=0.15\linewidth]{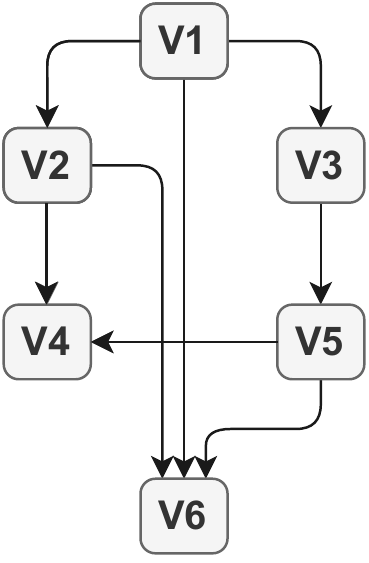}
		\label{fig:userApp}
	}
	\caption{Example of MFC and user's application.}
	\label{fig:mfcenv_userApp}
\end{figure*}

\subsection{Motivation and Contribution}
The existing mechanisms classify the application tasks into two groups: tasks for fog environment and tasks for the cloud environment, based on the resource demand, task type, priority, deadline of the application task, etc. As per the classification result, tasks are deployed and the resources are allocated to the tasks either from nearby fog or the cloud system. For example, in the case of real-time video streaming analysis in smart city, the raw video is preprocessed in the nearby fog computing node. 
Further, the preprocessed video is forwarded to cloud for in-detail processing, indexing and time-stamping of the video, which can be stored and accessed in future across the globe. In this fog-cloud architecture, single FN and one or more cloud providers are involved. 

The above fog-cloud architecture works efficiently as the multiple FNs across the street are deployed to handle the growing demand. 
However, due to uneven workload demand among the FNs in the city, some tasks may get rejected by the FNs at the city center and the request may get forwarded to the cloud. In such case, it is hard to maintain the QoS and the service provider may violate the Service Level Agreement (SLA). 
On the other hand, the resources of the FNs that receive less number of  demand, remain underutilized, resulting in lesser revenue for the service provider. Here, revenue refers to the monetary income generated from the service provider's business. In order to maintain the higher revenue, the service provider may increase the service cost affecting the users' QoE. This motivates us to revisit the resource allocation problem in a fog-cloud environment with the goal to maximize the resource utilization of the FNs. 

In order to address above research issues and be able to provide services in a real-time manner by the fog-cloud computing environment, in this paper, we attempt to modify the existing fog-cloud computing model. This is done via establishing the communication among the FNs through Fog-Cloud Interface (FCI), as shown in Figure \ref{fig:mfcenv}. 
The scope of this paper is to meet the following main goals which are summarized as follows:
\begin{itemize}
	\item Design a novel resource allocation strategy for MFC environment and multi-task user application.
	\item To optimize the resource utilization by allowing the tasks to be placed onto nearby FNs that are at multi-hop distance.
	\item To design and evaluate the performance of an algorithm to decide the order of processing and assigning the tasks onto MFC environment by considering the tasks' topology and their resource demand. 
\end{itemize}
Based on the motivation and the goals mentioned above, the main contributions in this paper can be summarized as follows:
\begin{itemize}
	\item An architecture for the multifog-cloud is designed that allows application tasks to be forwarded to other FNs at multi hop distance, thus reducing the workload on the cloud.
	\item The problem of resource allocation in multifog-cloud scenario is formulated as Integer Linear Programming (ILP) model. FNs located at multi-hop distances can also handle the users' requests, which complicates the problem when the different FCIs have different communication capabilities. 
	\item The physical infrastructure of multifog-cloud is modeled using graph theory. 
	\item A novel dedicated heuristic algorithm is designed for placement of the application tasks in multifog-cloud scenario.
	\item We exhibit the performance of proposed algorithms in improving fog and cloud computing resource utilization, network utilization, latency improvement, and other system parameters through extensive simulation.  
\end{itemize}

The rest of the paper is organized as follows. 
In Section \ref{sec:relWorks}, a brief summary of recent related articles is presented. 
In Section \ref{sec:symodel}, the system model on the users' application and physical infrastructure is presented. 
Section \ref{sec:probFormulation} presents the formulation of resource allocation problem in single-Fog single-cloud environment, where the user's tasks are deployed either in nearby FN or in cloud node.
This followed by Section \ref{sec:MFCmodel}, formulation of resource allocation problem in Multi-Fog Cloud environment, where user's tasks will have more options to be deployed. 
Considering the model, a heuristic resource allocation strategy is presented in Section \ref{sec:sol} followed by the performance evaluation of the proposed algorithm in Section \ref{sec:perfEval}. 
The concluding remarks and the scope of the future work are presented in Section \ref{sec:concls}. 
The list of acronyms that are used in this article is given in Table \ref{table:acrnm}.

\begin{table}
 	\caption{List of Acronyms}
 	\footnotesize
	\centering
	\begin{tabular}{|l|l||l|l|}
        \hline
            \textbf{Acronym } & \textbf{ Description } & \textbf{Acronym } & \textbf{ Description } \\ \hline
            CPS  &  Cyber Physical System   & QoS  &  Quality of Service  \\ \hline
            DAG  &  Directed Acyclic Graph   & SFC  &  Service Function Chain   \\ \hline
            FCI  &  Fog-Cloud Interface   & SLA  &  Service Level Agreement  \\ \hline
            FN  &  Fog Node  & VM  &  Virtual Machine  \\ \hline
            ILP  &  Integer Linear Programming   & WMD  &  Weighted Multi-Dimensional   \\ \hline
            IoT  &  Internet of Things   & LIFO  &  Last-In-First-Out    \\ \hline
            MFC  &  MultiFog-Cloud  & FCFS  &  First Come, Fist Served  \\ \hline
            QoE  &  Quality of Experience  & ~ & ~ \\ \hline
        \end{tabular}\label{table:acrnm}	
\end{table}

\section{Related Works}\label{sec:relWorks}
In this section, the recent related research articles are reviewed and presented taking different aspects into account. The resource management problem of fog and cloud computing is addressed in general, such as in \cite{tran-dang_frato_2021,kumar_dehury_efficient_2020,MISHRA2019217_2019_j7, 8456519_2018_j4, LI202224, 8014362}.

\subsection{Resource management in fog-cloud}
Authors in \cite{comp1, comp2} address the resource allocation problem with the objective to reduce the computation energy and the delay. The proposed computation offloading scheme in \cite{comp1} allows multiple end-users to transmit the data to the same nearby FN. 
However, the proposed algorithm may not be suitable when the cloud and fog environments are combined.
Similarly, in \cite{comp2} the game theory approach is followed to allocate the fog-cloud resources among IoT users with the goal to maximize the number of users, energy cost and the overall delay. The same resource allocation problem can also be seen and approached by using priced timed Petri nets method, as in \cite{comp3}, which predicts the task completion time and proactively allocate the resource in a dynamic manner considering both the monetary cost and task completion cost. The major pitfall of above resource allocation mechanisms is their inability to handle multi fog scenario. 
At a higher level of resource management problem, authors in \cite{8449762_2018_c1} proposed a resource identity management strategy considering fog and cloud computing environment. This would help the system administrator to manage not only the cloud resource but also the fog resources that are provided by the end-users. The limitation with this approach lies within single fog to cloud communication, which may not fit to the today's scenarios.

The proposed placement strategy in \cite{SOUZA20181_2018_j5} uses First-Fit approach to find the suitable FN for each parallel service modules. 
Extending the proposed scheme to fit the multi-fog scenario could further increase the concurrency of the service execution. Similarly, in \cite{MAHMUD2018_2018_j7}, authors proposed the application placement strategy in fog computing environment taking the Quality of Experience (QoE) into account. The proposed mechanism prioritizes the users' applications based on their requirements, intention,  execution platform and other users' expectation parameters, which may slightly deviate from the QoS parameters. Authors in \cite{8456519_2018_j4} investigated the resource allocation problem from the security point of view and proposed privacy-preserving resource allocation mechanisms for fog environments. The proposed mechanism mainly revolves around the fog computing environment and may not give efficient result in case of cloud environment. 

Introducing the fog computing environment between the industrial cloud and the terminal devices, authors in \cite{8399557_2018_j3} proposed a task scheduling algorithm, that ensures the task completion time and optimizes the execution concurrency. However, the proposed algorithm lacks the ability to utilize the computing resource available at the terminal of the nearby fog devices. 
Similarly, considering the heterogeneous FNs, a scalable and decentralized scheduling algorithm is presented by authors in \cite{8556474_2019_j3} with the goal to minimize the service delay by efficiently managing computation and communication resources of FNs.

\subsection{Energy optimization}
In order to offload the computation, while minimizing the energy consumption, authors in \cite{8673721_2019_j1, ADHIKARI2019100053_2019_j4} proposed different strategies, that mitigate and handle the problem of growing resource demand of IoT users. In \cite{8673721_2019_j1}, authors optimize the three major parameters: transmission power of the diverse application, computation resource distribution among those application and offloading decision. Similarly, authors in \cite{ADHIKARI2019100053_2019_j4} optimize two QoS parameters: energy consumption and computational time in sustainable fog-cloud infrastructure. To meet the requirement of multi-fog and cloud environment the above proposed mechanisms need to be extended significantly.

Based on the cost efficiency, authors in \cite{8437204_2018_j1} investigated the cloud and fog resource management problem. The problem is formulated and presented in a three-layer computing environment. A double two-sided matching optimization model is designed keeping the high cost-efficiency performance in mind. Similarly, with the min-max fairness guarantee, authors in \cite{8240666_2018_j2} proposed a sub-optimal resource allocation strategy that offloads the computational task to minimize the energy consumption and delay cost. The strategy is to divide the entire problem into offloading decision-making problem and resource allocation problem, which are solved by semidefinite relaxation and randomization method and Lagrangian dual decomposition method, respectively.

\subsection{Real-world application specific}
 The resource management problem is also addressed considering diverse practical real-life applications such as vehicular network \cite{07448886_2019_j2, comp1}, smart grid \cite{8616989_2018_c2, aranda_context-aware_2022}, smart Buildings \cite{8450331_2018_c3, 8450422_2018_c4}, smart manufacturing \cite{8399557_2018_j3, ROSENDO202271}, smart city \cite{WANG201911_2019_j5, ALAMGIRHOSSAIN2018226}. Authors in \cite{07448886_2019_j2} presented an adaptive resource management algorithm for vehicular networks with the goal to minimize the transmission rate, delay-jitter and the upper-bound of delay. A model for the integration of fog and cloud with smart grid is presented in \cite{8616989_2018_c2}, where the data flow and the request forwarding for electricity to micro-grid are handled by FNs. Fog and cloud computing environments are used for management of the smart building resources through different load balancing algorithms in \cite{8450331_2018_c3, 8450422_2018_c4, 8450410_2018_c5,8449762_2018_c1}. From the application point of view, authors in \cite{WANG201911_2019_j5} show how the smart city resources can be managed by taking advantage of fog computing. The fog computing model acts as the buffer and controller between cyber-physical world and the cloud computing environment with the goal to reduce the coupling in computing and maximize the utilization of resources.

\section{System Model}
\label{sec:symodel}

The users' request consists of multiple inter-connected tasks. An example of user's request consisting of six tasks is shown in Figure \ref{fig:userApp}. Each task is further associated with a certain amount of resource demand, fulfilled by the servers at fog and cloud environments. Figure \ref{fig:mfcenv} shows the  proposed hierarchy of users base at the bottom, FNs and cloud environment. 
Users access the services from the nearby FNs at the fog layer. For example, \emph{User1} and \emph{User2} access the requested service from the FN \emph{F1}. Between the fog layer and the cloud environment, a Fog-Cloud Interface (FCI) is introduced. FCI is responsible for (1) establishing a communication bridge between FNs and cloud and (2) establishing communication interface among FNs. It is also assumed that the FCIs may be interconnected among themselves. The communication among the FCIs enables the communication among FNs. As presented in Figure \ref{fig:mfcenv}, FNs that are connected to \emph{FCI1} may communicate with the FNs that are connected to \emph{FCI4} through \emph{FCI2}. This allows the users, connected to FN \emph{F1} under \emph{FCI4}, to access the cloud resource through \emph{FCI2}.

\begin{table}[htb]
\caption{List of key notations and description}
\resizebox{\textwidth}{!}{
\footnotesize
\begin{tabular}{l l}
	\hline 
	\textbf{Notation} & \textbf{Description} \\ 
	\hline
	$G$& $=\left( V, E\right) $ Resource graph  \\ 
	
	$V$ & $=\{C, F_1, F_2, \dots \}$ in MFC environment \\ 
	
	$C$ & Set of servers in cloud environment  $ = \{s_c^1, s_c^2, s_c^3, .., s_c^{n_c}\}$ \\
	
	$F_i$ & FN $i$ in MFC environment\\
	$s^i_{f_j}$ & Server $i$ in FN $F_j$ \\
	
	
	$\vec{P}_{ij}$ & Physical path between server $s^{m_i} \in F_k$ and $s^{m_j} \in F_l, k \ne l$ \\
	
	$\vec{H}_{ij}$ & Number of hops present in between the server $s^{m_i} \in F_k$ and $s^{m_j} \in F_l, k \ne l$ \\
	$R^x(s^{m_1})$&  Remaining resource of type $x$ available at server $s^{m_1} \in V$  \\ 
	
	$\alpha(s^{m_1},s^{m_2})$ & Available bandwidth available between server $s^{m_1} \in V$ and $s^{m_2} \in V$ \\
	
	$\beta(s^{m_1},s^{m_2})$ & Network latency on the physical link between server  $s^{m_1} \in V$ and $s^{m_2} \in V$ \\
	
	$G'$& $=\left( V', E'\right) $ User's application \\ 
	$V'$ & Set of tasks in user's application. $V'= \{v_1,v_2, ..., v_{n'}\}$ \\
	$E'$ & Set of edges between users' tasks \\
	$\hat{e}(v_i,v_j)$ & The boolean variable indicating if there is an edge existing from task $v_i$ to task $v_j$ \\
	$\bar{R}^x(v_i)$& Resource demand of type $x$ by task $v_i \in V'$  \\ 
	$\alpha'(v_i,v_j)$ & Bandwidth demand between  task $v_i$ and $v_j$ \\
	$\beta'(v_i,v_j)$ & Maximum network latency allowed from task $v_i$ to $v_j$ \\
	$\vec{\beta}(s^{m_i}, s^{m_j})$ & Network latency between the server $s^{m_i} \in F_k$ and $s^{m_j} \in F_l, k\ne l$ in MFC environment\\
	$\kappa$ & The mean network latency between the cloud environment and the fog environment \\ 
	 $P_i$ & Priority of a task $v_i$ \\
	 $M_i$ & Makespan of a task $v_i$\\
	 $\hat{M}_i$ & Normalized makespan of a task $v_i$ \\
	 $\hat{P}_i$ & Normalized Priority of a task $v_i$ \\
	 $\hat{R}_i$ & Normalized computing resource demand of a task $v_i$ \\
	\hline 
\end{tabular} \label{table:notation}
}
\end{table}

The placement and resource allocation strategy follows a bottom-up approach. As a result, the users' tasks are first forwarded to the fog environment. 
If the fog resources are not enough to fulfill the demand, the users' tasks are further forwarded to cloud. 
Due to the proximity the network latency between user and fog server is very less compared to the network latency between user and cloud. As a result, user can experience the real-time performance of different services, for instance services such as online gaming service, online video analysis service which depends on huge data exchange between user device and cloud or FN, etc.

The relationship among the set of cloud servers, fog servers, and the users can be represented as a graph structure. The list of notations that are used through-out the paper including the problem models in Section \ref{sec:probFormulation} and \ref{sec:MFCmodel} is given in Table \ref{table:notation}.

\subsection{Resource Graph}

\begin{figure}[t]
	\centering
	\subfloat[Resource graph consisting of one FN and one cloud environments.]
	{
		\includegraphics[width=0.30\linewidth]{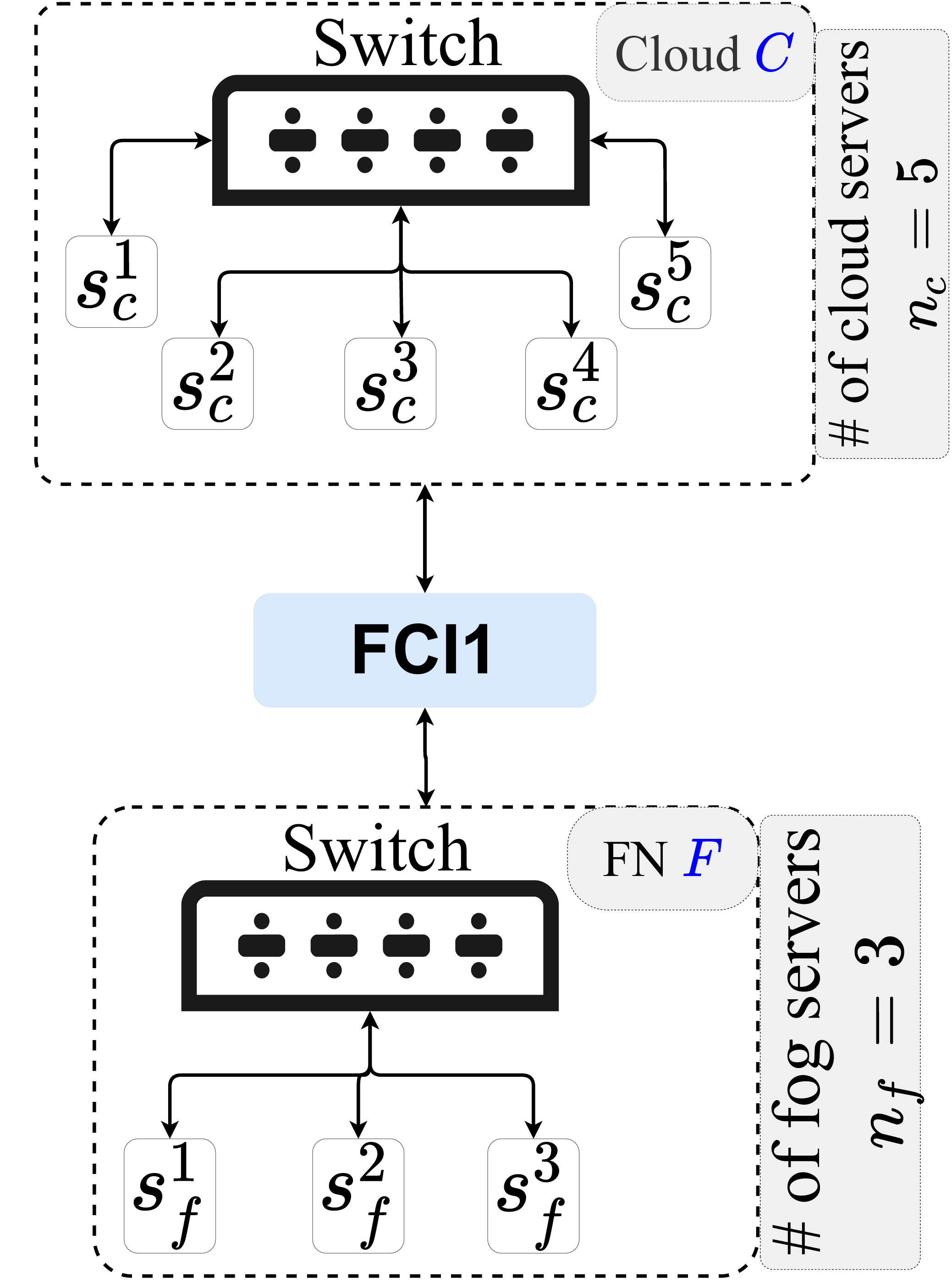}
		\label{fig:resource_graph:1fog-1cloud}
	}\hspace{10mm} 
	\subfloat[Resource graph consisting of two FNs and one cloud environments]
	{
		\includegraphics[width=0.30\linewidth]{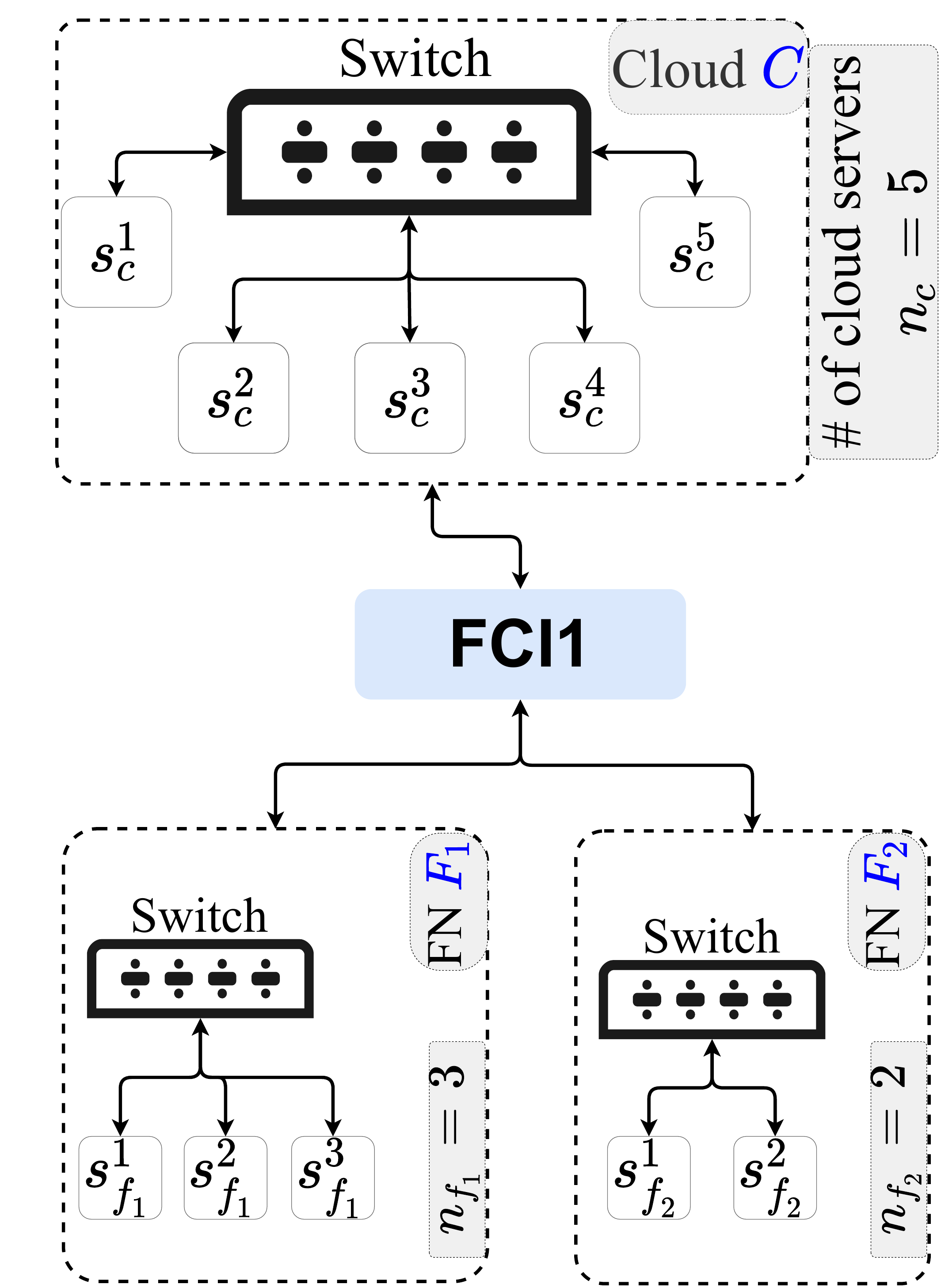}
		\label{fig:resource_graph:2fogs-1cloud}
	}
	\caption{An example of resource graph.} \vspace{-4mm}
	\label{fig:resource_graph}
\end{figure}

The resources of both cloud and fog servers are represented in a graph structure, known as resource graph. Let $G=\left( V, E\right) $ be the resource graph, where $V = \{C, F\}$ be the set of servers available in both Cloud and Fog environment and $E$ be the set of edges among those servers. An example of such resource graph can be seen in Figure \ref{fig:resource_graph:1fog-1cloud}. Further, $C=\{ s_c^1, s_c^2, s_c^3, .., s_c^{n_c} \}$ be the set of servers available in cloud environment and 
$F = \{s_f^1, s_f^2, s_f^3, .., s_f^{n_f}\}$
be the servers available in nearby fog environment. In Figure \ref{fig:resource_graph:1fog-1cloud}, the cloud environment consisting of five servers, $C=\{s_c^1, s_c^2, s_c^3, s_c^4, s_c^5\}$, and the fog environment consisting of three servers, $F = \{s_f^1, s_f^2, s_f^3\}$. For each server $s^{m_1} \in V$, $R^x(s^{m_1})$ represents the remaining resource available of type $x \in \{CPU, Memory\}$.  
Available bandwidth and network latency are the two parameters that are associated with each physical edge $e(s^{m_1},s^{m_2})$. Let $\alpha(s^{m_1},s^{m_2})$ and $\beta(s^{m_1},s^{m_2})$ be the available bandwidth and network latency, respectively, on physical link between server $s^{m_1}$ and $s^{m_2}$, $(s^{m_1},s^{m_2}) \in C$ or $(s^{m_1},s^{m_2}) \in F$. The mean network latency between the cloud and the fog environment is represented by $\kappa$. Since the network latency is assumed to be time-variant, the mean value is always considered during the placement of the tasks. 
Further, it is assumed that there exists a physical link between each pair of fog servers. For example the servers in Figure \ref{fig:resource_graph:1fog-1cloud} are connected directly without any FCI in-between. Mathematically, $e(s^{m_1},s^{m_2}) = 1, \forall (s^{m_1},s^{m_2}) \in F$. In other words, the fog servers are connected to each other in a complete graph. The other form of connection i.e. incomplete graph fog servers is discussed in Section \ref{sec:MFCmodel}, where it is considered as each pair of fog servers may not be connected directly, as the servers may belong to different fog environments in different geographical locations (e.g. in Figure \ref{fig:resource_graph:2fogs-1cloud}).

\subsection{Task Graph}

The users' requests are assumed to be comprised of multiple inter-connected tasks, represented as Directed Acyclic Graph (DAG), $G'=(V', E')$. An example of user's application is presented in Figure \ref{fig:userApp}, which consists of $6$ tasks, $V1 - V6$. Let $V'=\{v_1, v_2, ..., v_{n'}\}$ be the set of tasks and $E'$ be the set of edges among tasks. Each task $v_i$ is associated with resource demand of type CPU and memory, denoted by $\bar{R}^x(v_i), x \in \{CPU, Memory\}$. Let $\hat{e}(v_i,v_j)$ be the boolean variable indicating if there is an edge existing from task $v_i$ to task $v_j$. The direction of the edge also represents the dependencies of task $v_j$ upon the task $v_i$. This dependency indicates that the dependent task $v_j$ cannot start its execution before task $v_i$, as task $v_j$ needs input from task $v_i$. Since the tasks graph is modelled as DAG, the value of $\hat{e}(v_i,v_j)+\hat{e}(v_j,v_i) \le 1$. 

Taking the network resource demand into account, let $\alpha'(v_i,v_j)$ be the bandwidth demand between the task $v_i$ and $v_j$. In order to establish the smooth communication among the tasks, let $\beta'(v_i,v_j)$ be the maximum network latency allowed from task $v_i$ to $v_j$.

\section{Problem formulation } 
\label{sec:probFormulation}
Considering the resource graph and the task graph, the main goal is to map or place the interconnected tasks onto both cloud and fog servers such that the workload distribution among the physical servers can be balanced and the network latency can be minimized. 
The fog servers are given higher preference over the cloud servers while mapping process. The mapping process comprises of two stages: Task mapping and task edge mapping. In task mapping stage, all the tasks are mapped onto multiple servers. Here, we exploit the opportunity to map multiple tasks from same task graph onto the same physical servers. As a result; the number of physical servers involved is less than the number of tasks present in the task graph. In the task edge mapping stage, all the edges between the tasks are mapped onto the physical edges considering the network bandwidth and the network latency. 

%

\subsection{Task-edge mapping}
The boolean variable $\gamma_{c,j}^i = 1$, if the task $v_i$ is placed onto cloud server $s_c^j$, $0$ otherwise. Similarly, $\gamma_{f,k}^i = 1$ if the task $v_i$ is placed onto fog server $s_f^{k}$, $0$ otherwise. 

The boolean variable $D_{m_1m_2}^{i_1i_2}$ represents if the edge between task $v_{i_1}$ and $v_{i_2}$ is mapped onto the physical edge among cloud and fog servers $s^{m_1} \in V$ and $s^{m_2} \in V$. When a task edge needs to be mapped onto the physical edge between one cloud server and one fog server, an additional network latency $\kappa$ would be considered. In order to differentiate the set of fog servers and the cloud servers, the value of $\beta(s_c^{m_1},s_f^{m_2})= \kappa, \forall s_c^{m_1} \in C, \forall s_f^{m_2} \in F$ must be larger than the network latency between any two cloud or fog servers. Mathematically,

\begin{equation}
\kappa > \max_{0 < m_1, m_2 \le n_c}\left\lbrace \beta(s_c^{m_1},s_c^{m_2})\right\rbrace  
\end{equation}
\begin{equation}
\kappa > \max_{0 < m_1, m_2 \le n_f}\left\lbrace \beta(s_f^{m_1},s_f^{m_2})\right\rbrace 
\end{equation}

A precise summary of the assumptions considered in the context of single FN and cloud environment can be made as follows. A user's request is assumed to be comprised of multiple inter-connected tasks and hence no task is isolated from the other tasks in the same request. The fog servers within a FN are connected and there exists a physical link between each pair of fog servers. It is assumed that FCIs can communicate among themselves either over wired or wireless network connection. The network connection in the whole multiFog-cloud environment is time-variant.

\subsection{Objective function}\label{sec:obj:1F1C}
The problem of mapping the users' interconnected tasks to cloud and fog environment can be formulated as mixed integer linear programming problem. The main objectives are to balance the workload by uniform distribution of the tasks and minimize the network latency. 

Hence, the main objective can be divided into two sub-objectives: mapping the tasks to the servers and virtual edges to the physical edges. The following equation represents the assignment of a task $v_{i_1} \in V'$ onto a fog or cloud server.
\begin{equation}\label{eq:obj:task}
	z^{v_{i_1}} = \sum_{\forall s^{m_1} \in V}\left( \gamma_{c,m_1}^{i_1} + \Delta * \gamma_{f,m_1}^{i_1}\right)* \frac{1}{R^x(s^{m_1})}
\end{equation}
Similarly, the mapping of a virtual edge $\hat{e}(v_{i_1},v_{i_2}) \in E'$ onto a physical edge can be represented as follows:
\begin{equation}\label{eq:obj:edge}
	z^{(i_1i_2)} = \sum_{\forall(s^{m_1},s^{m_2}) \in E} D_{m_1m_2}^{i_1i_2} * \hat{e}(v_{i_1},v_{i_2}) * \left[\beta(s^{m_1},s^{m_2}) + \frac{1}{\alpha(s^{m_1},s^{m_2})}\right]
\end{equation}

Using the Equation \ref{eq:obj:task} for task mapping and Equation \ref{eq:obj:edge} for virtual edge mapping, the primary objective function for the user's request $G'$ can be formulated as below.

\textbf{Objective:}
\begin{equation} \label{eq:objfun}
 \min Z = \sum_{\forall v_{i_1} \in V'}  z^{v_{i_1}} + \sum_{\forall (v_{i_1},v_{i_2}) \in E'} z^{(i_1i_2)}
\end{equation}

\textbf{Constraint:}
\begin{equation}\label{const:cloudORfogServer}
\gamma_{c,j}^i + \gamma_{f,k}^i = 1, \quad \forall s_c^j \in C , \forall s_f^{k} \in F
\end{equation}
\begin{equation} \label{const:resrcCnstrnt}
	\bar{R}^x(v_i) < R^x(s_c^j) * \gamma_{c,j}^i + R^x(s_f^{k}) * \gamma_{f,k}^i
\end{equation}
\begin{equation}\label{const:BwCnstrnt}
\alpha(s^{m_1},s^{m_2}) > \alpha'(v_{i_1},v_{i_2})
\end{equation}
\begin{equation}\label{const:latencyCnstrnt}
\beta(s^{m_1},s^{m_2}) < \beta'(v_{i_1},v_{i_2})
\end{equation}
\begin{equation}\label{const:fogServerCnsnt} 
0 < \Delta < 1; \quad 0 < i_1, i_2 \le n_c; \quad 0< m_1, m_2 \le n_f
\end{equation}

The objective function is of two folds: 
\begin{itemize}
	\item The first part of the function allows all the tasks to be mapped onto the physical servers available at nearby FN and cloud. $\Delta$ is used as a constant to encourage the tasks to be mapped in fog servers instead of cloud servers. Since, the objective is to minimize, the task will be mapped to the servers having maximum remaining resource. Further, if two servers (one at FN and other at cloud) having minimum resources, the task will be mapped to the fog server due the multiplication of $\Delta$ constant. 
	\item The second part of the function (derived from Equation \ref{eq:obj:edge}) minimizes network latency while mapping the virtual edges to the physical edges. $\beta(s^{m_1},s^{m_2})$ in Equation \ref{eq:obj:edge} imposes additional network latency when two servers, $s^{m_1}\text{ and }s^{m_2}$, are in different computing environments. As a result, this discourages using any physical connection between cloud and fog as this would give a larger value of $\beta$. Minimizing the term $\frac{1}{\alpha(s^{m_1},s^{m_2})}$ in Equation \ref{eq:obj:edge} infers that the physical edges with maximum bandwidth available are highly preferred for any virtual edge. 
\end{itemize} 
However, the solution must meet the following constraints.

\begin{itemize}
	\item Constraint (\ref{const:cloudORfogServer}) ensures that all the tasks are mapped onto exactly one server. The physical server must either be cloud server or fog server. 
	\item Constraint (\ref{const:resrcCnstrnt}) and (\ref{const:BwCnstrnt}), ensure that the tasks and the edges between the tasks are mapped onto the physical servers and physical edges with enough remain resources available. 
	\item Similarly, Constraint (\ref{const:latencyCnstrnt}) ensures that the required network latency is fulfilled by the physical edge. On the other hand, Constraint (\ref{const:fogServerCnsnt}) ensures that the value of $\Delta$ must lies between $1$ and $0$ to encourage the tasks to be mapped onto the fog server. 
\end{itemize}

\section{MultiFog-Cloud (MFC) Environment} \label{sec:MFCmodel}

In Section \ref{sec:probFormulation}, the system consists of one cloud and one FN. For each user, the fog servers are connected to each other forming a complete graph. However, in order to fit into the real-scenario, a MultiFog-Cloud (MFC) environment is considered, as shown in Figure \ref{fig:mfcenv}, where more than one FNs are deployed in different geographical locations. All the FNs are directly or indirectly connected to each other. As a result, the fog servers in all FNs are forming an incomplete graph and no fog server is unreachable from any other fog server.  

\subsection{Multi-Fog resource graph}
As the system model, in Section \ref{sec:symodel}, is updated, the resource graph can be represented as $V = \{C, F_1, F_2, ....\}$. For example, in Figure \ref{fig:resource_graph:2fogs-1cloud}, two fog environments: $f_1$ consisting of three servers and $f_2$ consisting of two servers are connected to the cloud environment with five servers. Each fog consists of multiple servers represented as $F_i = \{s_{f_i}^1, s_{f_i}^2, s_{f_i}^3, ...\}$. Each server is represented as $s_{f_i}^j$, which indicates server $j$ in FN $F_i$. $|F_i|$ represents the number of servers present in FN $F_i$. It is assumed that the number of servers in all FNs are not uniform. In order to get the location of a server $s^{m_j}$, the notation $L(s^{m_j})$ is used. As the resource graph is not complete, the term physical edge is different from physical path. Physical edge can be defined as the direct link between two servers. 
On the other hand, a physical path is the combination of multiple physical links connected to each other to transfer data from one end to other ends. A physical link can be between FN and FCI, two FCIs, between FCI and cloud. 
The term physical path is used when establishing a connection between two servers from different FNs. The notation $\vec{P}_{ij}$ is used to indicate the physical path between server $s^{m_i} \in F_k$ and $s^{m_j} \in F_l, k\ne l$. If both the servers belongs to same node, the path $\vec{P}_{ij}$ is same to the notation $e(s^{m_i}, s^{m_j})$. $\vec{H}_{ij}$ indicates the number of hops present in between the server $s^{m_i} \in F_k$ and server $s^{m_j} \in F_l, k\ne l$.  $E$ represents the set of all physical edges and physical paths present in the multi-fog and cloud environment. $\vec{\beta}(s^{m_i}, s^{m_j})$ represents the network latency between the server $s^{m_i} \in F_k$ and $s^{m_j} \in F_l, k\ne l$. It is assumed that network latency of any physical link is smaller than the network latency of any path. Mathematically,
\begin{align}\label{eq:ltncyCnstrnt}
\nonumber \max\{\beta(s^{m_i}, s^{m_j}) | \forall F_k, s^{m_i} \text{and} s^{m_j} \in F_k\} < \\
\min\{\vec{\beta}(s^{m_i}, s^{m_j}) | \forall F_k,F_j, s^{m_i}\in F_k, s^{m_j} \in F_l \}
\end{align}
There is no relationship between network latency of a path between FN and cloud and the network latency of fog paths.

\subsection{Resource mapping in multi-fog environment}
Considering the new multi-fog cloud model, the complexity of resource mapping problem leverages when the mapping scope stretched from single-fog single-cloud scenario to multi-fog single-cloud scenario. The user's tasks can be mapped onto nearby fog, remote fog or remote cloud computing environment. During mapping, multiple parameters need to be taken into account, such as the number of hops, network latency, workload distribution, etc. 
The boolean variable $\gamma^i_{f_j,k} = 1$, if the task is placed onto the server $k$ in the FN (except the nearby FN) $F_j$, $0$ otherwise.   
However, the boolean variable $\hat{\gamma}_{f_j,k}^i = 1$, if task $v_i$ is placed onto a nearby fog server $s_{f_j}^{k}$, $0$ otherwise.

The task edges can be mapped onto the physical edge of physical path. In order to achieve the goal of minimizing network latency, the physical edges are given higher preferences over physical path while mapping the task edge. The boolean variable $D_{m_km_l}^{ij}$ represents if the task edge between task $v_i$ and $v_j$ is mapped onto the physical edge $e(s^{m_k}, s^{m_l})$ or the physical path $\vec{P}_{kl}$. 

A precise summary of the key assumptions considered in the context of multi-fog and cloud environment (in addition to the key assumptions in Section \ref{sec:symodel}) can be made as follows. The resource configurations/capacities in all FNs are heterogeneous. 
It is also assumed that the network latency of any physical link is smaller than the network latency of any path. Hence, the latency value between any FN and the corresponding FCI is smaller than the latency value among FCIs and between FCIs and the cloud. The proposed work assumed that the user assigns a priority to each task. When tasks are being scheduled, the effect of resource availability fluctuation is negligible.

\subsection{Objective function}\label{sec:obj:multiF1C}

Extending the objective function in Section \ref{sec:obj:1F1C}, the problem of mapping the user's task graph onto multi-fog and cloud environment can be formulated as mixed integer linear programming problem with the goal to minimize the total network latency and distribute the workload among fog and cloud servers in a balanced manner. The tasks are mapped in such a way that total network latency among tasks is minimum without violating the computing resource demand.

From Equation \ref{eq:objfun}, the mathematical formulation of single task $v_{i_1} \in V'$ mapping onto a physical server available at multi-fog single-cloud environment can be represented as 
\begin{equation}\label{eq:obj:MF1C:task}
	\hat{z}^{v_{i_1}} = \sum_{\forall s^{m_1} \in V} \frac{\left( \gamma_{c,m_1}^{i_1} +  \gamma_{f_j,m_1}^{i_1} + \Delta*\hat{\gamma}_{f_k,m_1}^{i_1} \right) 
	}{R^x(s^{m_1})}
\end{equation}

For further simplification, mapping virtual to physical edges is associated with two objectives: minimizing the network latency and balancing the network workload. While mapping a virtual edge, $(v_{i_1},v_{i_2})$, to a physical edge, the network latency can be calculated as below.
\begin{align}\label{eq:obj:MF1C:edge:latency}
	\nonumber \hat{z}^{(i_1i_2)}_{latency} = \sum_{\forall(s^{m_1},s^{m_2}) \in E} \left\lbrace  \hat{e}(v_{i_1},v_{i_2}) *  \left( D_{m_1m_2}^{i_1i_2} * \right.\right.\\
	\left. \left. \beta(s^{m_1},s^{m_2}) +   D_{m_1m_2}^{i_1i_2} * \vec{\beta}(s^{m_1},s^{m_2}) + \vec{H}_{s^{m_1},s^{m_2}} \right) \right\rbrace
\end{align}
Similarly, the mathematical formulation for mapping a virtual edge, $(v_{i_1},v_{i_2})$, onto a physical edge taking bandwidth availability into account can be represented as, 
\begin{equation}\label{eq:obj:MF1C:edge:bandwidth}
	\hat{z}^{(i_1i_2)}_{bandwidth} = \sum_{\forall(s^{m_1},s^{m_2}) \in E} \frac{D_{m_1m_2}^{i_1i_2} * \hat{e}(v_{i_1},v_{i_2})}{\alpha(s^{m_1},s^{m_2})}
\end{equation}

Considering the Equations \ref{eq:obj:MF1C:task}, \ref{eq:obj:MF1C:edge:latency}, and \ref{eq:obj:MF1C:edge:bandwidth}, the previous objective function in Equation \ref{eq:objfun} can further be modified as follows.\\
\textbf{Objective:}\\
\begin{equation}\label{eq:objfun2}
\min \hat{Z} = \sum_{\forall v_{i_1} \in V'} \hat{z}^{v_{i_1}} + \sum_{\forall (v_{i_1},v_{i_2}) \in E'}   \left[ \hat{z}^{(i_1i_2)}_{latency} + \hat{z}^{(i_1i_2)}_{bandwidth} \right]
\end{equation}

\textbf{Constraint:}\\
\begin{equation}\label{const:obj2:oneTaskOneFog}
\gamma_{f_j,m_1}^i + \gamma_{f_k,m_1}^i \le 1
\end{equation}
\begin{equation}\label{const:obj2:nwltncy}
	\beta'(v_1,v_2) \ge \vec{\beta}(s^{m_1}, s^{m_2})
\end{equation}
\begin{equation}\label{const:obj2:atlstOneTaskinOneFog}
\sum_{\forall v_i \in V'}\hat{\gamma}_{f_j,k}^i \ge 1
\end{equation}

In the first part of the function, servers available at nearby FN, other FNs, and the cloud environments are considered while assigning a task. The second part of the objective function focuses on mapping the virtual edges onto the physical edges. If no suitable physical edge is available, a physical path consisting of multiple physical edges will be chosen (details in Section \ref{sec:sol}). The number of hops and the network latency are collectively considered while selecting a physical path. In addition to the constraints mentioned in Equations \ref{const:cloudORfogServer}-\ref{const:fogServerCnsnt}, the objective function in Equation \ref{eq:objfun2} must satisfy following constraints.

\begin{itemize}
	\item Constraint \ref{const:obj2:oneTaskOneFog} ensures that no task is mapped onto multiple FNs. On the other hand, Constraint \ref{const:obj2:nwltncy} ensures that the maximum allowed network latency of any task edge $\hat{e}(v_{i_1},v_{i_2})$ is smaller than the corresponding physical path $\vec{P}_{m_1m_2}$.
	\item Constraint \ref{const:obj2:atlstOneTaskinOneFog} ensure that at least one task is placed in the nearby FN. This would allow the service provider to deploy the most critical task in the task graph onto the nearby FN. This would also prohibit the service provider to place the entire task graph in a very remote FN
\end{itemize}

\section{Proposed Algorithm} \label{sec:sol}

In this section, we discuss the proposed novel Heuristic Resource Allocation and optimization in MultiFog-Cloud (HeRAFC) algorithm, which works in a non-distributed manner and the users requests are handled on FCFS manner. Figure \ref{fig:mfcenv} shows the architecture of multiple fogs and cloud and Figure \ref{fig:userApp} shows an example of users' application. As discussed in earlier section, the job of this algorithm is to distribute the application tasks among nearby FN, the FNs at multi-hop distance and the cloud. Each FN consists of multiple interconnected fog servers and network devices such as switches, routers, etc.

\begin{definition}{\textbf{1-hop distance:}}
	A FN $F_i$ is said to be at 1-hop distance from other FNs $\forall F_j, F_i \ne F_j$, if both the FNs, $F_i$ and $F_j$, share same FCI. Taking Figure \ref{fig:mfcenv} into account, for all the servers located in cloud $C$, all the FNs in $FCI1, FCI2,$ and $FCI3$ are at 1-hop distance, but not the FNs that are connected to $FCI4$. Similarly, FN $F_1$ and $F_2$ under $FCI1$ are at 1-hop distance from each other.
\end{definition}
\begin{definition}{\textbf{2-hop distance:}}
	A FN $F_i$ is said to be at 2-hop distance from other FNs $\forall F_j, F_i \ne F_j$, if the corresponding FCIs can communicate directly with each other without cloud. For example, fog $F_1$ under $FCI1$ is at 2-hop distance from all the FNs that are connected to $FCI2$. Additionally, the servers in cloud $C$ and the FNs connected to $FCI4$ are at 2-hop distance from each other.
\end{definition}
\begin{definition}{\textbf{n-hop distance:}}
	In general, two FNs are said to at n-hop distance, if $n, n > 0$ number of FCIs are involved in the physical path between corresponding FNs.  For example in Figure \ref{fig:mfcenv}, the FN under $FCI4$ is at 3-hop (here $n=3$) distance from the FN $F_2$ under $FCI1$.
\end{definition}


Three major parameters associated with the users' applications are taken into account: makespan of the task, priority value of each task, and computing resource demand that includes CPU and memory demand.
\begin{definition}{\textbf{Task makespan:}}
	Makespan of a task $v_i$, $M_i$ refers to the time required to execute the task, also referred to as the task's execution time. This excluded the response time. The makespan of the tasks can be used to derive the makespan of the application. 
\end{definition}

\begin{definition}{\textbf{Priority value of task:}}
	Priority of a task $v_i$, $P_i$, refers to how important the task is. The value of $P_i$ is decided by the user. Higher the value of priority, more important the task is. Mathematically, a task $v_1$ is said to be higher priority than the task $v_2$, if $P_1 > P_2$.
\end{definition}

Based on above-mentioned three major parameters, the critical value of each task is calculated. This work considers parameters with a wide range of values. For example, makespan of a task is in time unit and the value may range from 1 to several thousands. Similarly, memory resource demand value is in megabytes and the value may range from hundreds to thousands. The value of CPU resource demand may reach up to 100 and the priority of tasks can be of any discrete integer value. For such a diverse range of values, it is necessary to normalize the values to a range of (0, 1], i.e. greater than 0 and less than or equal to 1. Without normalization, a slight increase in one parameter may have a substantial impact on the final result. On the other hand, a significant change in another parameter may result in a slight change in the final result. The normalized makespan of a task is calculated as:
\begin{equation}\label{eq:normMkspn}
\hat{M}_i = \frac{M_i}{\max\{M_i | \forall v_i \in V\}}
\end{equation}

Similarly the normalized priority value of a task can be calculated as 
\begin{equation}\label{eq:normPriority}
\hat{P}_i = \frac{P_i}{\max\{P_i | \forall v_i \in V\}}
\end{equation}
Since computing resource refers to the CPU and memory demand which are in different units, the normalized value of each type of resource demand is calculated followed by calculating the average of both normalized values. Mathematically, the normalized CPU demand can be calculated as:
\begin{equation}\label{eq:normCPU}
\hat{R}^{CPU}(v_i) = \frac{\bar{R}^{CPU}(v_i)}{ \max \{ R^{CPU}(s^{m_1}) | \forall m_1 \in \{ C, F_1, F_2, \dots \}  \}}
\end{equation}
The normalized value of memory demand can be calculated as:
\begin{equation}\label{eq:normMem}
\hat{R}^{mem}(v_i) = \frac{\bar{R}^{mem}(v_i)}{ \max \{ R^{mem}(s^{m_1}) | \forall m_1 \in \{ C, F_1, F_2, \dots \}  \}}
\end{equation}

Taking the value calculated using Equation \ref{eq:normCPU} and \ref{eq:normMem}, the normalized average weighted computing resource demand can be calculated as follows.
\begin{equation}\label{eq:normCompRsrc}
\hat{R}_i = \frac{\Omega^c*\hat{R}^{CPU}(v_i) + \Omega^m\hat{R}^{mem}(v_i)}{2}
\end{equation}
, where $\Omega^c\text{ and }\Omega^m$ are the constants, $\Omega^c + \Omega^m = 1, 0 < \Omega^c,\Omega^m < 1 $. These constants are used to assign preference values to \textit{CPU} and \textit{mem} resource.

Considering the normalized values calculated in Equation \ref{eq:normMkspn}-\ref{eq:normCompRsrc}, the critical value of a task is calculated. Weighted Multi-Dimensional (WMD) approach is followed, where each parameter represents one dimension in a multi-dimensional space. In this context, the makespan, priority value and computing resource demand represent three dimensions of a space. In this 3-Dimentional (3D) space, a task is represented as a 3D rectangular prism object. The volume of such 3D rectangular prism object represents the critical value of that corresponding task. However, in order to allow the service provider to give higher preference to one parameter over other, a weight factor is used in the critical value calculation. The critical value (or the volume of the 3D object) of a task $v_i$ can be calculated as 
\begin{equation}\label{eq:criticalVal}
	WV(v_i) = (w_1*\hat{M}_i)*(w_2*\hat{P}_i)*(w_3*\hat{R}_i),\quad w_1+w_2+w_3 = 1
\end{equation}

\subsection{HeRAFC algorithm}\label{sec:sol:algo}
The whole process of the proposed resource allocation strategy can be divided into two phases. In the first phase, the order of tasks, based on WMD approach, for deployment is decided. In the second phase, suitable locations for each task is chosen. 

\subsubsection{Order of task}\label{sec:sol:algo:taskOrder}
The WMD-based algorithm for deciding the order of tasks is presented in Algorithm \ref{algo:HeRAFC:tskOrdr}. The user's application is taken as the input to this algorithm. The job of this algorithm is to arrange the tasks in queue, which further would be followed to assign or map to the fog-cloud environment. The detail description of this algorithm is presented in Section \ref{sec:sol:desc}. 

\begin{algorithm}[htb] \small
	 
	\KwIn{User's application; MFC environment}
	
	$ST = $\{Sort the list of tasks in ascending order based on the number of out-edges\}\;\label{algo:HeRAFC:tskOrdr:sortTask} 
	
	$ProcessQ=$\{ \} \textit{/* Order of tasks need to be deployed */} \; 
	
	$TL=$\{Extract the list of tasks from $ST$ with number of out-edges $0$\}\;\label{algo:HeRAFC:tskOrdr:extrctTaskLst}
	
	\While{True}{\label{algo:HeRAFC:tskOrdr:LoopStart}
		\ForEach{task $t \in TL$}{\label{algo:HeRAFC:tskOrdr:MCVstart}
			Calculate critical value $WV(t)$ using Equation \ref{eq:criticalVal}\;
			Calculate number of out-edges $OD(t)$\;
			$MCV(t) = WV(t)/(OD(t)+\delta)$ \textit{/*Calculate mean critical value of the task*/} \label{algo:HeRAFC:tskOrdr:MCVend} \;
		}	
        
		Sort $TL$ based on their $MCV$ in ascending order\;\label{algo:HeRAFC:tskOrdr:sortExtTskLst}	
		Append $TL$ to $ProcessQ$ \label{algo:HeRAFC:tskOrdr:append2PrcsQ}\;
		
		Append $-1$ to $ProcessQ$ \label{algo:HeRAFC:tskOrdr:apnd-1} \textit{/* to mark the tasks in level */} \;
		
		$tmp = $ \{Extract the list of parent tasks of $TL$ from $ST$\}\;\label{algo:HeRAFC:tskOrdr:extractPrntTask}
		Remove all the tasks in $TL$\;
		Add the tasks in $tmp$ to $TL$\;\label{algo:HeRAFC:tskOrdr:repeatExtrctTskLst}
	}\label{algo:HeRAFC:tskOrdr:LoopEnd}
\caption{WMD-based Task order selection algorithm}   
\label{algo:HeRAFC:tskOrdr}
\end{algorithm}

\subsubsection{Location selection}\label{sec:sol:algo:locSel}
In the second phase of the HeRAFC algorithm (as persented in Algorithm \ref{algo:HeRAFC}), the ordered task queue is followed and for each task, a suitable location is searched, known as node mapping. Location here refers to one FN or the cloud. Following this, all the edges among tasks are mapped by finding the suitable path from the corresponding source to destination environment where both the end tasks are already mapped. The detail description of this algorithm is presented in Section \ref{sec:sol:desc}.
\begin{algorithm} \small
	\KwIn{User's application; MFC environment}
	
	$ProcessQ = $Queue of tasks decided by Algorithm \ref{algo:HeRAFC:tskOrdr}\; \label{algo:HeRAFC:callAlgo1}
	Calculated the resource availability matrix $RM$\; \label{algo:HeRAFC:calcAvaMat}
	$mapQ =$ \{\} \; 
	\While{True}{ \label{algo:HeRAFC:nodeMap:start}
		$ETL = $\{Extract the list of tasks from $ProcessQ$ until $-1$ in LIFO manner.\}\;\label{algo:HeRAFC:extrctTskSet}
		\ForEach{(task $t$ in $ETL$) AND ($t$ is not scheduled)}{\label{algo:HeRAFC:nodeMap:innerLoop:start}
			$CFL = $\{get the list of fogs or cloud where the child tasks are deployed\}\;\label{algo:HeRAFC:chldFogLst}
			$nf = CFL$\;\label{algo:HeRAFC:chldNrbyFogLst}
			\If{$CFL == NULL$}{
				$nf = $Nearby FN of user's application\; \label{algo:HeRAFC:tskOrdr:} \label{algo:HeRAFC:appNrbyFN}
			}
			Check if $t$ can be deployed on any one of the fog or cloud in $nf$\; \label{algo:HeRAFC:chckNFLst:start}
			\If{$t$ deployed on $nf$}{
				Append $t$ to $mapQ$ and update $RM$ \;
				Continue with next task from Step \ref{algo:HeRAFC:nodeMap:innerLoop:start} \;\label{algo:HeRAFC:chckNFLst:end}
			}
			$oh = $\{Find the FNs or cloud that are one hop distance to the all fogs in $nf$\}\;\label{algo:HeRAFC:findOH}
			Check if $t$ can be deployed on any one of the fog or cloud in $oh$\;\label{algo:HeRAFC:chckOH:start}
			\If{$t$ deployed on $oh$}{
				Append $t$ to $mapQ$ and update $RM$ \;
				Continue with next task from Step \ref{algo:HeRAFC:nodeMap:innerLoop:start} \;\label{algo:HeRAFC:chckOH:end}
			}
			$th = $\{Find the FNs or cloud that are two hop distance to the all fogs in $nf$\} \;\label{algo:HeRAFC:findTH}
			Check if $t$ can be deployed on any one of the fog or cloud in $th$ \;\label{algo:HeRAFC:chckTH:start}
			\If{$t$ deployed on $th$}{
				Append $t$ to $mapQ$ and update $RM$\;
				Continue with next task from Step \ref{algo:HeRAFC:nodeMap:innerLoop:start} \;\label{algo:HeRAFC:chckTH:end}
			}
		}\label{algo:HeRAFC:nodeMap:innerLoop:end}
		$E'=$\{set of task edges that are connected to the tasks in $ETL$\}\;\label{algo:HeRAFC:formEdgeSet}
		Sort $E'$ in descending order based on bandwidth demand\;\label{algo:HeRAFC:EdgeSorting}
		\ForEach{edge $e \in E'$}{\label{algo:HeRAFC:EdgeMapping:start}
			
			Find source and destination task, $s$ \& $t$\; \label{algo:HeRAFC:getSrcDstTsk}
			$sh = $ Find the FN or cloud where $s$ is deployed\;
			$th = $ Find the FN or cloud where $t$ is deployed\;
			\eIf{$sh \ne NULL$ OR $th \ne NULL$}{
				Find the shortest path from $sh$ to $th$ and map the task edge\; \label{algo:HeRAFC:shrtstPath}
			}{
				Ignore the current edge $e$\; \label{algo:HeRAFC:EdgeMapping:end}
			}
		}
		Reset $RM$ and continue with next set of tasks.\; \label{alog:HeRAFC:resetRsrcMtrx}
	}
	\caption{Resource Allocation in MultiFog-Cloud algorithm}
	\label{algo:HeRAFC}
\end{algorithm}

\subsection{HeRAFC description}\label{sec:sol:desc}
User's application and the information regarding the physical environment are provided to Algorithm \ref{algo:HeRAFC} as the input. The algorithm starts by invoking Algorithm \ref{algo:HeRAFC:tskOrdr} that would decide the order of the tasks that need to be followed in the location selection phase. Algorithm \ref{algo:HeRAFC:tskOrdr} starts by sorting the list of tasks based on the number of child tasks. To calculate the number of child tasks of a parent task, the number of out-edges is calculated as in Line \ref{algo:HeRAFC:tskOrdr:sortTask}. Out-edge refers to the number of outgoing edges from a task. From the sorted task list, the tasks that are having out-edge value 0 or the task having no child are first selected, as in Line \ref{algo:HeRAFC:tskOrdr:extrctTaskLst}. This means, the algorithm gives higher priority to leaf tasks while processing. This way, this algorithm takes the tasks' dependency into consideration.

For each extracted task,the mean critical value (which eventually follows WMD approach) is calculated, in Line \ref{algo:HeRAFC:tskOrdr:MCVstart}-\ref{algo:HeRAFC:tskOrdr:MCVend}, taking the ratio of critical value to the number of out-edges. 
For the leaf node the number of out edges is $0$. To avoid any \emph{Divide by Zero error}, a small constant $\delta$ is added to $OD(t)$. 
The extracted list is further sorted based on the mean critical value and appended to the process task queue, in Line \ref{algo:HeRAFC:tskOrdr:sortExtTskLst} and \ref{algo:HeRAFC:tskOrdr:append2PrcsQ}, respectively. The same process is applied to the parent task of the currently extracted tasks, as in Line \ref{algo:HeRAFC:tskOrdr:repeatExtrctTskLst}. 
Processing the tasks level-by-level is taken into account to ensure that the dependent tasks (or the child tasks) are executed not before the execution of their predecessor tasks (or the parent tasks). The time complexity of this algorithm is $ \mathcal{O}(|V'| \log{}|V'|)$ , where $V'$ be the set of tasks (proof in Appendix \ref{appen:TimeCompAlgo1}).

The output of Algorithm \ref{algo:HeRAFC:tskOrdr}, $ProcessQ$, is further followed to map the tasks and the edges. Each set of tasks are separated by $-1$ in $ProcessQ$, as in Algorithm \ref{algo:HeRAFC:tskOrdr}, Line \ref{algo:HeRAFC:tskOrdr:apnd-1}, can be logically represents as the tasks in one level of the task graph. 
The set of tasks, $ETL$, is first extracted from the $ProcessQ$, as in Line \ref{algo:HeRAFC:extrctTskSet} in a Last-In-First-Out (LIFO) manner. Extracting in LIFO would further ensure that the dependent tasks (or child tasks) are not executed before the execution of predecessor tasks. This indicates that the root task in the DAG will be the first one to get executed. For each task in $ETL$, the set of child tasks is calculated followed by forming a list of FNs and cloud, $CFL$, that are hosting those child tasks, in Line \ref{algo:HeRAFC:chldFogLst}. The set of hosting environment, $CFL$ is considered as nearby fog list, $nf$, in Line \ref{algo:HeRAFC:chldNrbyFogLst}. However, in case of absence of no child task, the FN that is nearer to the user is considered, as in Line \ref{algo:HeRAFC:appNrbyFN}. In Line \ref{algo:HeRAFC:chckNFLst:start} - \ref{algo:HeRAFC:chckNFLst:end}, the current task is checked if this can be assigned to any one of the nearby FNs list. If the amount of resource can be fulfilled by the anyone of the FN or by the cloud, the task will be assigned to the selected FN or cloud and the same process will be applied to next task. However, if no hosting environment is found in Line \ref{algo:HeRAFC:chckNFLst:start}, the list of FNs and cloud that are at one hop distance, $oh$ from the child task is calculated, as in Line \ref{algo:HeRAFC:findOH}. On calculation of the set $oh$, in Line \ref{algo:HeRAFC:chckOH:start} - \ref{algo:HeRAFC:chckOH:end}, all the hosting environments will be checked if the current task can be deployed. In case of availability of the resources on any FN that can fulfill the demand, the current task will be assigned to the selected FN or the cloud. Further, in case of failure of resource availability for the current task, the hosting environments that are at two hops distance will be calculated in Line \ref{algo:HeRAFC:findTH} and the same process will be applied for checking the resource and assigning the task, as mentioned in Line \ref{algo:HeRAFC:chckTH:start} - \ref{algo:HeRAFC:chckTH:end}. This concludes the procedure for mapping the tasks on to different FNs and cloud. 

Upon mapping the set of tasks, $ETL$, the set of adjacent edges, $E'$ in Line \ref{algo:HeRAFC:formEdgeSet}, will be assigned, in Line \ref{algo:HeRAFC:EdgeMapping:start} - \ref{algo:HeRAFC:EdgeMapping:end}. For each edge, $e$, the information regarding the source task and the destination task are extracted from the task graph, in Line \ref{algo:HeRAFC:getSrcDstTsk}. If both the tasks are already mapped onto the physical environment, the shortest path between the corresponding hosting environment will be calculated and assigned to host the edge $e$, in Line \ref{algo:HeRAFC:shrtstPath}. Upon mapping of both the set of tasks in one level of task graph and the edges among them, the resource matrix will be reset, in Line \ref{alog:HeRAFC:resetRsrcMtrx}. This is due to the fact that one level tasks can not run in parallel to the other level of tasks (parent or child tasks) and need to be scheduled to run in a sequential manner. Hence, upon completion of execution of one level of tasks, the leased resources will be released for the set of tasks in the next level of the task graph. The time complexity of the proposed  algorithm is $\mathcal{O}(|E'|\left[\log{}|E'| + |E|log(|V|)\right])$ , where $E$ and $E'$ are the set of physical edges and task edges, respectively and $V$ is the number of cloud and FNs ( proof in \ref{appen:TimeCompAlgo2}).

\section{Performance evaluation}\label{sec:perfEval}
In this section, the performance of the proposed algorithm is discussed. Liu \emph{et al.} (DRACO) \cite{comp1} and Shah-Mansour \emph{et al.} (HFCCI) \cite{comp2} algorithms are used to compare and present the efficiency in terms of computing and network resource utilization in different hosting environment, such as fog and cloud. FNs nearby can receive data from multiple users in case of DRACO algorithm. By using game theory, HFCCI minimizes the number of users and reduces energy costs while allocating cloud resources. Various performance matrices are employed for evaluating efficiency, including the effect of tasks' order on resource utilization, a comparison of network resource utilization between fog and cloud environments, and tasks' latency. The simulation uses a rule-based approach over a data-driven ML-based approach. Rule-based approaches can be used to generate simulated environments that mimic the real-world. This uses predefined configurations and rules, and does not require any historical data to fit the environment. On the other hand, a data-driven ML-based approach needs historical data that fits the environment and the problem statement. In the context of heterogeneous fog and cloud resource allocation and optimization, the data related to historical performance of the computing and storage cluster, network resource utilization, real requests from the users are difficult to acquire. To the best of our knowledge, there are no such data available that fit the problem addressed in this paper. This is the main reason why most of the research community and we preferred a rule-based approach over a data-driven ML-based approach. In the following subsection, the detail information on simulation environment and the results obtained from the simulation are presented.

\subsection{Simulation environment}\label{sec:perfEval:env}

A Python-based discrete event simulator is used to simulate the proposed resource management algorithm. This simulation does not explicitly configure tasks to run on different CPU cores. This is entirely dependent on the number of cores available and how the native operating system behaves. The resource configuration of the FCIs, fog, cloud and the users' applications are presented in JSON file. 
The applications are created dynamically/programmatically. Separate JSON files are created for each user's app, including the tasks list, edges list, resource demand of tasks and edges, and other related information. The detailed configuration of the simulation parameters are given in Table \ref{table:simPar}.
\ref{appen:simEnv}\revison{, provides a statistical analysis of the entire environment consisting of users' applications, tasks, resource demand, FNs, FCIs, resource availability and priority, connectivity among FNs, FCI and cloud, and network resource availability. } 

\revison{The MFC environment consists of $500$ FNs that are connected to $200$ FCIs. Each FCI is connected to at least one FN, whereas a single FN is connected to precisely one FCI. For each FN, the total amount of CPU capacity ranges from 50 through 100 and the values are assigned by following a random distribution. The total memory (RAM) for each FN is assigned randomly ranging from $200 GB$ to $400 GB$, where the unit $GB$ represents GigaByte. The approximate computation power of a FN is calculated as MIPS (Million instructions per  second), which ranges from $3000$ to $5000$.} 

\revison{As discussed before, each FN is connected to exactly one FCI. The maximum network bandwidth available between a FN and the corresponding FCI ranges from $300 Mbps$ through $400 Mbps$. In this paper, it is considered that a FCI may directly communicate with another FCI. In such a situation, there exist a network bandwidth capacity among FCIs, which ranges from $400 Mbps$ through $1000 Mbps$. Similarly, the bandwidth between FCIs and the cloud range from $400 Mbps$ through $1000 Mbps$. The network bandwidth unit is in Mbps (Mega bits Per Second). It is assumed that latency value between any FN and the corresponding FCI is small than the latency among FCIs and between FCIs and cloud. Latency values are assigned to the physical links by following a uniform distribution. The latency (in millisecond) between FN and FCIs ranges from $50 ms$ through $100 ms$. Similarly, the latency values for rest of the physical links range from $100 ms$ through $200 ms$. Technically, a task can be hosted on 1-hop or at a multi-hop distance. However, in the current implementation maximum number of hops is set to $2$. This means there are a maximum of two FCIs and no cloud are present in the physical path between any two FNs that are used to host user's tasks.} 

\revison{The maximum number of applications is set to $10000$. For each application, random number of tasks are generated ranging from $4$ through $12$. However, the total maximum number of tasks for the entire simulation is set to $100,000$. For each task, the CPU and memory demand range from $1-4 CPUs$ and $100 MB - 1000 MB$, respectively. The average makespan of each task is $350 ms$ with minimum and maximum makespan of $10 ms$ and $1000 ms$, respectively. The network bandwidth demand of edges is ranging from $100 Mbps$ and $200 Mbps$. Similarly, the average latency demand among tasks ranges between $10 ms$ and $50 ms$. The unit of latency is millisecond (ms). The minimum and maximum priority values of a task are $1$ and $5$, respectively. The bandwidth and latency demand of a task ranges between $100-200 Mbps$ and $10-300 ms$, respectively. The number of edges within an application is decided by a link probability of $0.6$. This indicates the connectivity probability among two task.}

\begin{table}[]
\centering
\scriptsize
\caption{Simulation parameters for Fog-Cloud environment and users application}
    \begin{tabular}{|p{0.33\textwidth} p{0.13\textwidth} || p{0.33\textwidth} p{0.08\textwidth}|}    
    \hline
    \multicolumn{4}{|c|}{Simulation Parameters}                                                                   \\ \hline
    \multicolumn{2}{|c||}{\textbf{Fog-Cloud Environment}}       & \multicolumn{2}{c|}{\textbf{Users Application}} \\ \hline    
    Environment configuration file format & JSON &    Application configuration file format & JSON \\ 
    Total number of FNs &  500 &   Number of applications & 5000-10000  \\ 
	Number of FCIs & 200  &   Number of tasks per application & 4-12  \\  
                    &       & Maximum number of tasks in the simulation  & 100,000 \\
	CPU capacity per FN & 50-100 &   Required vCPU per task & 1-4  \\  
	RAM (in GB) capacity per FN & 100-200 &   Required RAM (in MB) per task & 100-1000       \\  
	Average MIPS per FN & 250-400  &  Makespan of each task (in ms) & 10-1000   \\  
	Bandwidth between of FN and FCI & 300-400 Mbps &   Priority of each task & 1-5  \\  
	Bandwidth between of FCI and cloud & 400-1000 Mbps &   Edge bandwidth demand (in Mbps) & 100-200  \\  
	Bandwidth among FCIs & 400-1000 Mbps &  Edge latency demand (in ms) & 10-50   \\  
	Latency between FN and FCI & 50-100 ms &   &  \\  
	Latency between FCI and Cloud & 101-200 ms  &  &   \\  
	Latency among FCIs & 101-200 ms &   &  \\  
	Maximum number of hops & 2  &   & \\ \hline
    \end{tabular} \vspace{-4mm}
\label{table:simPar}
\end{table}

\subsection{Simulation results}\label{sec:perfEval:result}
Following the discussion of simulation environment in above subsection, this subsection presents the results of the simulation in more details. Figure \ref{fig:sim:compRsrcUtil_F} and Figure \ref{fig:sim:compRsrcUtil_C} provide the resource utilization of fog and cloud environment, respectively. In this simulation implementation, resource refers to CPU, Memory and network bandwidth. Figure \ref{fig:sim:compRsrcUtil_F} shows the comparison of resource utilization of proposed HeRAFC algorithm with DRACO \cite{comp1} and HFCCI \cite{comp2} algorithms only in fog environment. Resource utilization is calculated as the ratio of amount of resources allocated and the total resource capacity. 

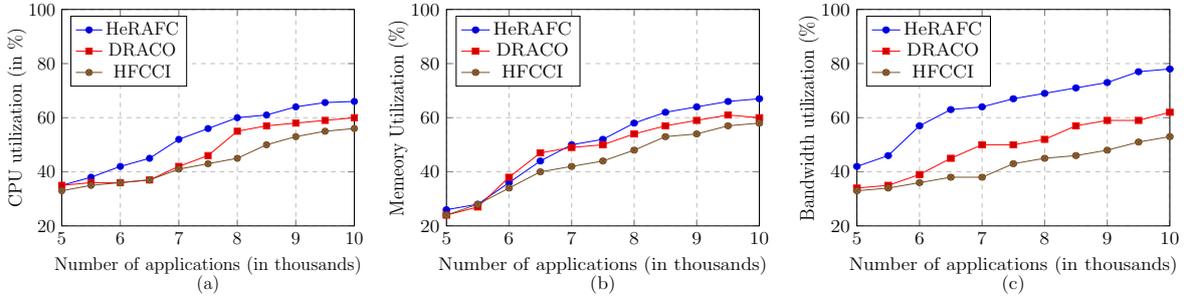
\begin{figure}[h!]
	\centering
\subfloat{\resizebox{50mm}{40mm}{
    \begin{tikzpicture}
        \begin{axis}[
                width=0.5\textwidth,
                height=.4\textwidth,
                legend pos=north west,
                symbolic x coords={5, 5.5, 6, 6.5, 7, 7.5, 8, 8.5, 9, 9.5,  10},
                nodes near coords align={vertical},
                ymin=20,ymax=100,xmin=5,xmax=10,
                ylabel={CPU utilization (in \%)},
                y label style={at={(-0.1,0.5)}},
                xlabel={Number of applications (in thousands)},ymajorgrids=true, 
                xmajorgrids=true, grid style=dashed,
                title style={at={(0.5,-0.24)},anchor=north},
                title = (a),
            ]
            \addplot table[x=noofApp,y=HeRAFC]{\oneacompCPUUtilF};
            \addplot table[x=noofApp,y=DRACO] {\oneacompCPUUtilF};
            \addplot table[x=noofApp,y=HFCCI] {\oneacompCPUUtilF};
            \legend{HeRAFC, DRACO, HFCCI}
        \end{axis}
    \end{tikzpicture} \label{fig:sim:compCPUUtil_F}}}
\subfloat{\resizebox{0.32\textwidth}{40mm}{
    \begin{tikzpicture}
        \begin{axis}[
                width=0.5\textwidth,
                height=.4\textwidth,
                legend pos=north west,
                symbolic x coords={5, 5.5, 6, 6.5, 7, 7.5, 8, 8.5, 9, 9.5,  10},
                nodes near coords align={vertical},
                ymin=20,ymax=100,xmin=5,xmax=10,
                ylabel={Memeory Utilization (\%)},
                y label style={at={(-0.1,0.5)}},
                xlabel={Number of applications (in thousands)},ymajorgrids=true, xmajorgrids=true, grid style=dashed,
                title style={at={(0.5,-0.24)},anchor=north},
                title = (b),
            ]
            \addplot table[x=noofApp,y=HeRAFC]{\onebcompMemUtilF};
            \addplot table[x=noofApp,y=DRACO] {\onebcompMemUtilF};
            \addplot table[x=noofApp,y=HFCCI] {\onebcompMemUtilF};
            \legend{HeRAFC, DRACO, HFCCI}
        \end{axis}
    \end{tikzpicture}\label{fig:sim:compMemUtil_F}} }
\subfloat{\resizebox{0.32\textwidth}{40mm}{
    \begin{tikzpicture}
        \begin{axis}[
                width=0.5\textwidth,
                height=.4\textwidth,
                legend pos=north west,
                symbolic x coords={5, 5.5, 6, 6.5, 7, 7.5, 8, 8.5, 9, 9.5,  10},
                nodes near coords align={vertical},
                ymin=20,ymax=100,xmin=5,xmax=10,
                ylabel={Bandwidth utilization (\%)},
                y label style={at={(-0.1,0.5)}},
                xlabel={Number of applications (in thousands)},ymajorgrids=true, xmajorgrids=true, grid style=dashed,
                title style={at={(0.5,-0.24)},anchor=north},
                title = (c),
            ]
            \addplot table[x=noofApp,y=HeRAFC]{\oneccompBWUtilF};
            \addplot table[x=noofApp,y=DRACO] {\oneccompBWUtilF};
            \addplot table[x=noofApp,y=HFCCI] {\oneccompBWUtilF};
            \legend{HeRAFC, DRACO, HFCCI}
        \end{axis}
    \end{tikzpicture}\label{fig:sim:compBWUtil_F}} }
	\caption{Comparison of only FOG resource utilization \protect \subref{fig:sim:compCPUUtil_F} CPU \protect \subref{fig:sim:compMemUtil_F} memory  \protect \subref{fig:sim:compBWUtil_F} network bandwidth.} \vspace{-4mm}
	\label{fig:sim:compRsrcUtil_F} 
\end{figure} 
\revison{Figures }\ref{fig:sim:compRsrcUtil_F}\subref{fig:sim:compCPUUtil_F}\revison{, }\ref{fig:sim:compRsrcUtil_F}\subref{fig:sim:compMemUtil_F}\revison{, and }\ref{fig:sim:compRsrcUtil_F}\subref{fig:sim:compBWUtil_F}\revison{ show the average CPU, memory, and network bandwidth utilization of FNs, respectively. X-axis represents the total number of applications ranging from $5000$ to $10000$ and Y-axis represents the average resource utilization in percentage. It can be observed that, Upon assigning $5000$ applications using the HeRAFC algorithm, the average CPU utilization of fog servers is approximately $35\%$, similar to the DRACO algorithm. However, it is observed, the CPU utilization is approximately $33\%$ in the case of the HFCCI algorithm, as shown in Figure }\ref{fig:sim:compRsrcUtil_F}\subref{fig:sim:compCPUUtil_F}\revison{. Resource utilization continuously increases when the number of applications increases to $8000$. However, when the number of applications increases beyond $8000$, the rate of increase slows down. It is observed that the average CPU utilization is $66\%, 60\%, $ and $56\%$ in case of HeRAFC, DRACO, and HFCCI algorithms, respectively, when there are a total of $10$ thousand applications. A similar pattern is also observed in memory resource utilization. In case of the proposed HeRAFC algorithm the memory utilization of the fog environment increases from approximately $26\%$ to $67\%$, when the number of applications increases to $5000$ to $10000$. Figure }\ref{fig:sim:compRsrcUtil_F}\subref{fig:sim:compMemUtil_F}\revison{ demonstrates that the proposed HeRAFC algorithm outperforms DRACO and HFCCI algorithms. The memory utilization of FNs is $60\%$ and $58\%$ when DRACO and HFCCI algorithms are applied, respectively, when the number of applications is $10$ thousand. Overall, the proposed HeRAFC algorithm performs better than the others in utilizing the CPU and memory resources. This is due to the underlined multi-fog environment where schedulers can take advantage of the fog resources and hence assign a maximum number of tasks to the fog environment instead of the cloud. However, DRACO and HFCCI do not take advantage of the nearby available fog resources.} 

\revison{Figure }\ref{fig:sim:compRsrcUtil_F}\subref{fig:sim:compBWUtil_F}\revison{, shows the average network bandwidth utilization of different algorithms with the number of applications ranging from $10$ to $100$ within fog environment. As the number of applications increases, the bandwidth utilization of all three methods also increases, but each method has a different rate of increase. It is evident from the chart that HeRAFC and DRACO have similar bandwidth utilization rates, while HFCCI has a lower bandwidth utilization rate. In case of HeRAFC, the average network utilization ranges from $42\%$ to $78\%$ when the number of applications ranges from $5$ to $10$ thousand, which is more than the network utilization obtained using other two related algorithms. This is due to the fact that HeRAFC explores the other FNs that are at multi-hop distance to fulfill the resource demand of application. As a result, the workload on the fog servers and the network among FNs are relatively higher. The average network utilization within FNs in case of DRACO and HFCCI are approximately $62\%$ and $53\%$, respectively, when the number of applications is $10$ thousand. It is to be noted that out of all the applications, some of the applications are hosted on FNs and rest are hosted on cloud environment. The resource utilization of the cloud environment is provided in Figure }\ref{fig:sim:compRsrcUtil_C}.


\begin{figure}[t]
	\centering
\subfloat{\resizebox{0.32\textwidth}{40mm}{
    \begin{tikzpicture}
        \begin{axis}[
                width=0.5\textwidth,
                height=.4\textwidth,
                legend pos=north west,
                symbolic x coords={5,  5.5,  6,  6.5,  7,  7.5,  8,  8.5,  9,  9.5,  10},
                nodes near coords align={vertical},
                ymin=20,ymax=100,xmin=5,xmax=10,
                ylabel={CPU utilization (in \%)},
                y label style={at={(-0.1,0.5)}},
                xlabel={Number of applications},ymajorgrids=true, xmajorgrids=true,  grid style=dashed,
                title style={at={(0.5,-0.22)},anchor=north},
                title = (a),
            ]
            \addplot table[x=noofApp,y=HeRAFC]{\twoacompCPUUtilC};
            \addplot table[x=noofApp,y=DRACO] {\twoacompCPUUtilC};
            \addplot table[x=noofApp,y=HFCCI] {\twoacompCPUUtilC};
            \legend{HeRAFC, DRACO, HFCCI}
        \end{axis}
    \end{tikzpicture} \label{fig:sim:compCPUUtil_C}}}
\subfloat{
    \resizebox{0.32\textwidth}{40mm}{
    \begin{tikzpicture}
        \begin{axis}[
                width=0.5\textwidth,
                height=.4\textwidth,
                legend pos=north west,
                symbolic x coords={5,  5.5,  6,  6.5,  7,  7.5,  8,  8.5,  9,  9.5,  10},
                nodes near coords align={vertical},
                ymin=20,ymax=100,xmin=5,xmax=10,
                ylabel={Memeory Utilization (\%)},
                y label style={at={(-0.1,0.5)}},
                xlabel={Number of applications},ymajorgrids=true, xmajorgrids=true, grid style=dashed,
                title style={at={(0.5,-0.22)},anchor=north},
                title = (b),
            ]
            \addplot table[x=noofApp,y=HeRAFC]{\twobcompMemUtilC};
            \addplot table[x=noofApp,y=DRACO] {\twobcompMemUtilC};
            \addplot table[x=noofApp,y=HFCCI] {\twobcompMemUtilC};
            \legend{HeRAFC, DRACO, HFCCI}
        \end{axis}
    \end{tikzpicture}\label{fig:sim:compMemUtil_C}} }
\subfloat{\resizebox{0.32\textwidth}{40mm}{
    \begin{tikzpicture}
        \begin{axis}[
                width=0.5\textwidth,
                height=.4\textwidth,
                legend pos=north west,
                symbolic x coords={5,  5.5,  6,  6.5,  7,  7.5,  8,  8.5,  9,  9.5,  10},
                nodes near coords align={vertical},
                ymin=20,ymax=100,xmin=5,xmax=10,
                ylabel={Bandwidth utilization (\%)},
                y label style={at={(-0.1,0.5)}},
                xlabel={Number of applications},ymajorgrids=true, xmajorgrids=true, grid style=dashed,
                title style={at={(0.5,-0.22)},anchor=north},
                title = (c),
            ]
            \addplot table[x=noofApp,y=HeRAFC]{\twoccompBWUtilC};
            \addplot table[x=noofApp,y=DRACO] {\twoccompBWUtilC};
            \addplot table[x=noofApp,y=HFCCI] {\twoccompBWUtilC};
            \legend{HeRAFC, DRACO, HFCCI}
        \end{axis}
    \end{tikzpicture}\label{fig:sim:compBWUtil_C}} }
    \caption{Comparison of only CLOUD resource utilization (a) CPU (b) memory  (c) network bandwidth.} \vspace{-4mm}
	\label{fig:sim:compRsrcUtil_C} 
\end{figure}
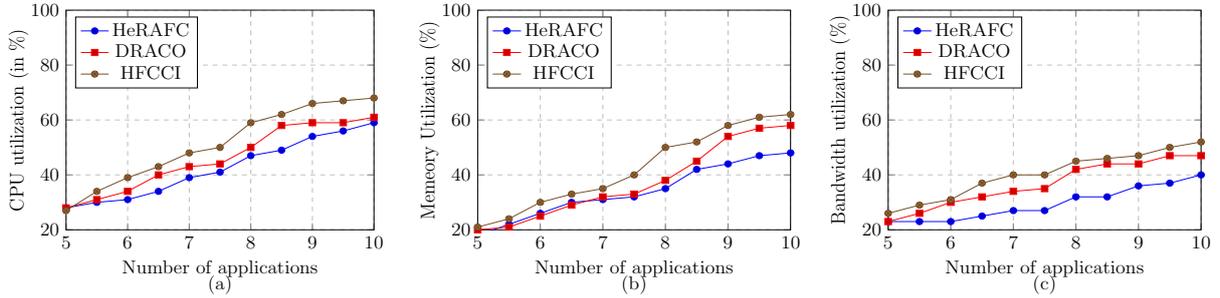
\revison{The cloud resource utilization in the case of all three algorithms is presented in Figure }\ref{fig:sim:compRsrcUtil_C}\revison{. Using HeRAFC, the average cloud CPU utilization increases from approximately $28\%$ to $59\%$ when the number of applications increases from $5000$ to $10000$, as in Figure }\ref{fig:sim:compRsrcUtil_C}\subref{fig:sim:compCPUUtil_C}\revison{. However, in the case of DRACO and HFCCI algorithms, the average CPU utilization increases from $28\%$ to $61\%$ and from approx. $27\%$ to $68\%$, respectively, with the same number of applications. Evidently, the cloud CPU utilization using HeRAFC is less than that of the other two algorithms, as the FNs (including those at a multi-hop distance) are given higher priority over the cloud while assigning the applications. A similar trend can be seen in Figure }\ref{fig:sim:compRsrcUtil_C}\subref{fig:sim:compMemUtil_C}\revison{. The average memory utilization in the cloud environment using the HeRAFC algorithm with $5000$ applications is approximately $18\%$. In contrast, the memory utilization values in the case of DRACO and HFCCI algorithms are approximately $20\%$ and $21\%$, respectively. Such a utilization pattern indicates a lesser dependency on the cloud. This is because HeRAFC gives higher priority to the nearby and other FNs, letting fewer tasks be offloaded to the cloud. The impact of such a strategy can also be seen in the significant reduction of cloud network resource utilization, as discussed below.}

\revison{The average network bandwidth utilization, as shown in Figure }\ref{fig:sim:compRsrcUtil_C}\subref{fig:sim:compBWUtil_C}\revison{, refers to the utilization of physical links that are attached to the cloud. When the number of applications increases from $5000$ to $10000$, the average network bandwidth utilization using the HeRAFC algorithm increases from $23\%$ to approximately $40\%$, which is less than that of the other two algorithms. For DRACO and HFCCI, the average network bandwidth utilization increased to $47\%$ and $52\%$ when the number of applications is increased to $10000$. The HeRAFC algorithm gives higher preferences to the nearby FNs. As a result, the maximum number of task edges are assigned to the physical path within FNs,  leading to relatively lesser bandwidth utilization in the cloud environment.}

\revison{The results in Figure }\ref{fig:sim:Ltncy_F}\revison{ and }\ref{fig:sim:Ltncy_C}\revison{ show the effect of weighted critical value ($WV$, as calculated in Equation }\ref{eq:criticalVal})\revison{ and tasks' priority on average network latency in Fog and in cloud environments, respectively. In other words, we observed the average latency of tasks with different priority and $WV$ values. The value of $WV$ depends upon three parameters: makespan of the tasks (Equation }\ref{eq:normMkspn}\revison{), priority (Equation }\ref{eq:normPriority}\revison{), and resource demand (Equation }\ref{eq:normCompRsrc}\revison{). By changing the value of $w_1$ (weight for makespan), $w_2$ (weight for priority), and $w_3$ (weight for resource demand), the value of $WV$ can be changed. Four configurations of $w_1, w_2, $and$ w_3$ are considered in the simulation. Under the first configuration, the same weightage values ($w_1=w_2=w_3$) are given to all the parameters (Figure }\ref{fig:sim:Ltncy_F}\revison{(a)). Under the second configuration, a higher weightage is given to the resource demand of the tasks ($w_1=0.2, w_2=0.3, w_3=0.5$), as shown in Figure }\ref{fig:sim:Ltncy_F}(b)\revison{. In the third and forth configurations, higher importance is given to task's priority (Figure }\ref{fig:sim:Ltncy_F}(c)\revison{) and makespan (}\ref{fig:sim:Ltncy_F}(d)\revison{), respectively. The latency value of a task is calculated by taking the average latency of out-going adjacent edges. So in case of an edge from a task with Priority 1 ($P_i=1$) to a task with priority 2 ($P_i=2$), the value of $P_i=1$ task is calculated by considering the value of latency of edge. In other words, the latency of an edge is assigned to the source task and not the destination task. In this simulation, the tasks with priority value 5 are given higher priority and the minimum priority value of a task is set to 1. The priority values are assigned to the task following a random distribution. Figure }\ref{fig:sim:Ltncy_F}\revison{ shows the latency of different priority tasks within the fog environment. The x-axis represents the total number of tasks, ranging from $50$ thousand to $100$ thousand and Y-axis represents average latency of different tasks in milliseconds (ms). It is observed that the average latency of the high priority tasks is always less than that of the tasks with low priority value, under different configurations.}

\pgfplotstableread[row sep=\\,col sep=&]{
noofTasks &	  fog_p5 &	  fog_p4 &	  fog_p3 &	  fog_p2 &	  fog_p1 &	  cloud_p5 &	  cloud_p4 &	  cloud_p3 &	  cloud_p2 &	  cloud_p1\\
50 &	  11 &	  30 &	  60 &	  93 &	  146 &	  55 &	  240 &	  480 &	  837 &	  1200\\
60 &	  19 &	  33 &	  62 &	  110 &	  150 &	  95 &	  297 &	  496 &	  880 &	  1460\\
70 &	  20 &	  43 &	  63 &	  118 &	  178 &	  140 &	  336 &	  504 &	  944 &	  1520\\
80 &	  22 &	  48 &	  76 &	  128 &	  181 &	  132 &	  384 &	  608 &	  1024 &	  1592\\
90 &	  25 &	  48 &	  80 &	  130 &	  190 &	  175 &	  387 &	  693 &	  1048 &	  1780\\
100 &	  34 &	  50 &	  99 &	  131 &	  199 &	  204 &	  450 &	  720 &	  1300 &	  1810\\
}\latencyWCVweightuniform

\pgfplotstableread[row sep=\\,col sep=&]{
noofTasks &	  fog_p5 &	  fog_p4 &	  fog_p3 &	  fog_p2 &	  fog_p1 &	  cloud_p5 &	  cloud_p4 &	  cloud_p3 &	  cloud_p2 &	  cloud_p1\\
50 &	  9 &	  25 &	  53 &	  82 &	  123 &	  61 &	  250 &	  500 &	  882 &	  1287\\
60 &	  16 &	  28 &	  53 &	  100 &	  129 &	  105 &	  307 &	  524 &	  923 &	  1557\\
70 &	  17 &	  38 &	  56 &	  107 &	  157 &	  149 &	  350 &	  530 &	  986 &	  1580\\
80 &	  18 &	  41 &	  68 &	  116 &	  161 &	  141 &	  394 &	  633 &	  1060 &	  1687\\
90 &	  21 &	  42 &	  73 &	  119 &	  167 &	  184 &	  399 &	  719 &	  1096 &	  1876\\
100 &	  29 &	  43 &	  90 &	  120 &	  179 &	  213 &	  465 &	  746 &	  1344 &	  1888\\
}\latencyWCVweightonresource

\pgfplotstableread[row sep=\\,col sep=&]{
noofTasks &	  fog_p5 &	  fog_p4 &	  fog_p3 &	  fog_p2 &	  fog_p1 &	  cloud_p5 &	  cloud_p4 &	  cloud_p3 &	  cloud_p2 &	  cloud_p1\\
50 &	  10 &	  30 &	  55 &	  93 &	  126 &	  63.1 &	  256.8 &	  515.1 &	  900.6 &	  1353.2\\
60 &	  10 &	  30 &	  60 &	  101 &	  128 &	  102.7 &	  312.7 &	  533.7 &	  940.8 &	  1599.3\\
70 &	  12 &	  34 &	  62 &	  105 &	  129 &	  143.5 &	  356.1 &	  540.1 &	  1022.7 &	  1648.4\\
80 &	  22 &	  37 &	  62 &	  109 &	  130 &	  148.5 &	  403.8 &	  645.7 &	  1088.7 &	  1708.6\\
90 &	  23 &	  37 &	  74 &	  114 &	  138 &	  182.7 &	  405.1 &	  731.8 &	  1111.6 &	  1894.4\\
100 &	  30 &	  45 &	  89 &	  116 &	  152 &	  217 &	  471.9 &	  765.7 &	  1367.5 &	  1936.4\\
}\latencyWCVweightonpriority

\pgfplotstableread[row sep=\\,col sep=&]{
noofTasks &	  fog_p5 &	  fog_p4 &	  fog_p3 &	  fog_p2 &	  fog_p1 &	  cloud_p5 &	  cloud_p4 &	  cloud_p3 &	  cloud_p2 &	  cloud_p1\\
50  &	  8  &	  28  &	  52  &	  83  &	  120  &	  91  &	  299  &	  535  &	  889  &	  1354\\
60  &	  15  &	  34  &	  54  &	  93  &	  127  &	  122  &	  304  &	  541  &	  969  &	  1537\\
70  &	  16  &	  34  &	  64  &	  97  &	  141  &	  158  &	  323  &	  556  &	  1066  &	  1659\\
80  &	  18  &	  41  &	  67  &	  100  &	  145  &	  173  &	  407  &	  570  &	  1135  &	  1659\\
90  &	  19  &	  42  &	  70  &	  108  &	  156  &	  180  &	  412  &	  571  &	  1211  &	  1703\\
100  &	  26  &	  44  &	  85  &	  116  &	  156  &	  193  &	  468  &	  634  &	  1219  &	  1794\\
}\latencyWCVweightonmakespan

\usepgfplotslibrary{groupplots}

\begin{figure}[t]
\begin{tikzpicture}
\pgfplotsset{
    width=0.25\textwidth,
    height=40mm,
    xtick={50,60,70,80,90,100},
    xticklabels={50,60,70,80,90,100},
    x tick label style={rotate=90,anchor=east,font=\tiny},
    xmin=50,xmax=100,
    ymin=5,ymax=200,
    ylabel style={align=center},
    every axis/.append style={font=\scriptsize},
    every axis title/.style={below right,at={(0.37,-0.11)}},
}     
    \begin{groupplot}[ 
        group style={
            group size=4 by 1,
            horizontal sep=15mm,
        },
    ]
    \nextgroupplot[
            title={(a)},            
            ylabel={Avg. latency (in ms)},            
    ]
        \addplot[smooth,mark=o,black] table[x=noofTasks,y=fog_p5]{\latencyWCVweightuniform}; 
        \addplot[smooth,mark=x,red] table[x=noofTasks,y=fog_p4]{\latencyWCVweightuniform}; 
        \addplot[smooth,mark=triangle*,blue] table[x=noofTasks,y=fog_p3]{\latencyWCVweightuniform}; 
        \addplot[smooth,mark=square*,green] table[x=noofTasks,y=fog_p2]{\latencyWCVweightuniform}; 
        \addplot[smooth,mark=diamond*,magenta] table[x=noofTasks,y=fog_p1]{\latencyWCVweightuniform}; 
    \nextgroupplot[
            title={(b)},
    ]
        
        \addplot[smooth,mark=o,black] table[x=noofTasks,y=fog_p5]{\latencyWCVweightonresource}; 
        \addplot[smooth,mark=x,red] table[x=noofTasks,y=fog_p4]{\latencyWCVweightonresource}; 
        \addplot[smooth,mark=triangle*,blue] table[x=noofTasks,y=fog_p3]{\latencyWCVweightonresource}; 
        \addplot[smooth,mark=square*,green] table[x=noofTasks,y=fog_p2]{\latencyWCVweightonresource}; 
        \addplot[smooth,mark=diamond*,magenta] table[x=noofTasks,y=fog_p1]{\latencyWCVweightonresource};
    \coordinate (c2) at (rel axis cs:0,0);
    \nextgroupplot[
            title={(c)},
            legend style={legend columns=5,fill=none,draw=black,anchor=center,align=center, font=\scriptsize},
            legend to name=fred
    ]
        
        \addplot[smooth,mark=o,black] table[x=noofTasks,y=fog_p5]{\latencyWCVweightonpriority}; 
        \addplot[smooth,mark=x,red] table[x=noofTasks,y=fog_p4]{\latencyWCVweightonpriority}; 
        \addplot[smooth,mark=triangle*,blue] table[x=noofTasks,y=fog_p3]{\latencyWCVweightonpriority}; 
        \addplot[smooth,mark=square*,green] table[x=noofTasks,y=fog_p2]{\latencyWCVweightonpriority}; 
        \addplot[smooth,mark=diamond*,magenta] table[x=noofTasks,y=fog_p1]{\latencyWCVweightonpriority};  
    \addlegendentry{p5};    
    \addlegendentry{p4};    
    \addlegendentry{p3};    
    \addlegendentry{p2};    
    \addlegendentry{p1};
    \coordinate (c3) at (rel axis cs:0,0);
    \nextgroupplot[
            title={(d)},
    ]             
        \addplot[smooth,mark=o,black] table[x=noofTasks,y=fog_p5]{\latencyWCVweightonmakespan}; 
        \addplot[smooth,mark=x,red] table[x=noofTasks,y=fog_p4]{\latencyWCVweightonmakespan}; 
        \addplot[smooth,mark=triangle*,blue] table[x=noofTasks,y=fog_p3]{\latencyWCVweightonmakespan}; 
        \addplot[smooth,mark=square*,green] table[x=noofTasks,y=fog_p2]{\latencyWCVweightonmakespan}; 
        \addplot[smooth,mark=diamond*,magenta] table[x=noofTasks,y=fog_p1]{\latencyWCVweightonmakespan};  
    \end{groupplot}
    \coordinate (c6) at ($(c2)!1.0!(c3)$);    
    \node[below] at (c6 |- current bounding box.south)
      {\footnotesize Number of tasks (in thousands)};    

    \coordinate (c6) at ($(c2)!1.0!(c3)$);    
    \node[above] at (c6 |- current bounding box.north)
      {\pgfplotslegendfromname{fred}};    
\end{tikzpicture}
\caption{Latency of different priority tasks in Fog environment under different weighted critical values\\ ($WV$) \scriptsize (a) $w_1=w_2=w_3$, (b) $w_1=0.2, w_2=0.3, w_3=0.5$, (c) $w_1=0.2, w_2=0.5, w_3=0.3$, (d) $w_1=0.5, w_2=0.3, w_3=0.2$} \vspace{-4mm}
\label{fig:sim:Ltncy_F}
\end{figure}
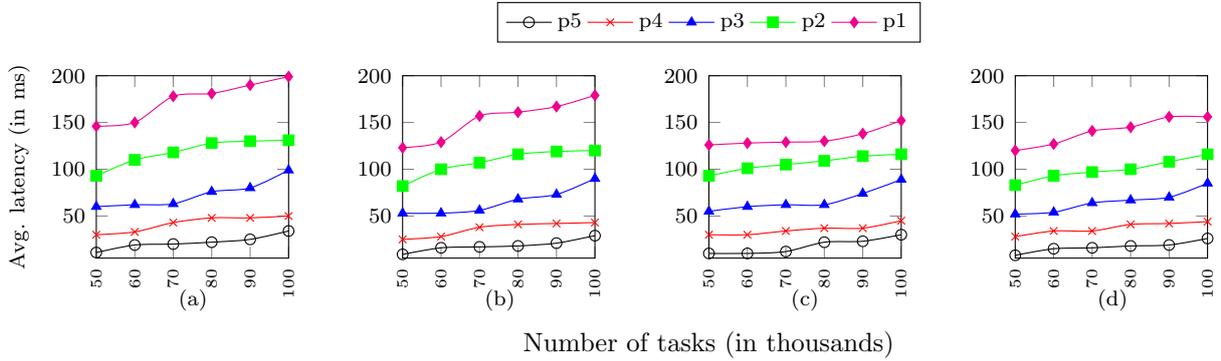

\begin{figure}[h]
\begin{tikzpicture}
\pgfplotsset{
    width=0.25\textwidth,
    height=40mm,
    xtick={50,60,70,80,90,100},
    xticklabels={50,60,70,80,90,100},
    x tick label style={rotate=90,anchor=east,font=\tiny},
    xmin=50,xmax=100,
    ymin=50,ymax=2000,
    ylabel style={align=center},
    every axis/.append style={font=\scriptsize},
    every axis title/.style={below right,at={(0.37,-0.11)}},
}     
    \begin{groupplot}[ 
        group style={
            group size=4 by 1,
            horizontal sep=15mm,
        },
    ]
    \nextgroupplot[
            title={(a)},            
            ylabel={Avg. latency (in ms)},            
    ]
        \addplot[smooth,mark=o,black] table[x=noofTasks,y=cloud_p5]{\latencyWCVweightuniform}; 
        \addplot[smooth,mark=x,red] table[x=noofTasks,y=cloud_p4]{\latencyWCVweightuniform}; 
        \addplot[smooth,mark=triangle*,blue] table[x=noofTasks,y=cloud_p3]{\latencyWCVweightuniform}; 
        \addplot[smooth,mark=square*,green] table[x=noofTasks,y=cloud_p2]{\latencyWCVweightuniform}; 
        \addplot[smooth,mark=diamond*,magenta] table[x=noofTasks,y=cloud_p1]{\latencyWCVweightuniform};
    \nextgroupplot[
            title={(b)},
    ]
        
        \addplot[smooth,mark=o,black] table[x=noofTasks,y=cloud_p5]{\latencyWCVweightonresource}; 
        \addplot[smooth,mark=x,red] table[x=noofTasks,y=cloud_p4]{\latencyWCVweightonresource}; 
        \addplot[smooth,mark=triangle*,blue] table[x=noofTasks,y=cloud_p3]{\latencyWCVweightonresource}; 
        \addplot[smooth,mark=square*,green] table[x=noofTasks,y=cloud_p2]{\latencyWCVweightonresource}; 
        \addplot[smooth,mark=diamond*,magenta] table[x=noofTasks,y=cloud_p1]{\latencyWCVweightonresource};
    \coordinate (c2) at (rel axis cs:0,0);
    \nextgroupplot[
            title={(c)},
            legend style={legend columns=5,fill=none,draw=black,anchor=center,align=center, font=\scriptsize},
            legend to name=fred
    ]
        
        \addplot[smooth,mark=o,black] table[x=noofTasks,y=cloud_p5]{\latencyWCVweightonpriority}; 
        \addplot[smooth,mark=x,red] table[x=noofTasks,y=cloud_p4]{\latencyWCVweightonpriority}; 
        \addplot[smooth,mark=triangle*,blue] table[x=noofTasks,y=cloud_p3]{\latencyWCVweightonpriority}; 
        \addplot[smooth,mark=square*,green] table[x=noofTasks,y=cloud_p2]{\latencyWCVweightonpriority}; 
        \addplot[smooth,mark=diamond*,magenta] table[x=noofTasks,y=cloud_p1]{\latencyWCVweightonpriority};  
    \addlegendentry{p5};    
    \addlegendentry{p4};    
    \addlegendentry{p3};    
    \addlegendentry{p2};    
    \addlegendentry{p1};
    \coordinate (c3) at (rel axis cs:0,0);
    \nextgroupplot[
            title={(d)},
    ]             
        \addplot[smooth,mark=o,black] table[x=noofTasks,y=cloud_p5]{\latencyWCVweightonmakespan}; 
        \addplot[smooth,mark=x,red] table[x=noofTasks,y=cloud_p4]{\latencyWCVweightonmakespan}; 
        \addplot[smooth,mark=triangle*,blue] table[x=noofTasks,y=cloud_p3]{\latencyWCVweightonmakespan}; 
        \addplot[smooth,mark=square*,green] table[x=noofTasks,y=cloud_p2]{\latencyWCVweightonmakespan}; 
        \addplot[smooth,mark=diamond*,magenta] table[x=noofTasks,y=cloud_p1]{\latencyWCVweightonmakespan};  
    \end{groupplot}
    \coordinate (c6) at ($(c2)!1.0!(c3)$);    
    \node[below] at (c6 |- current bounding box.south)
      {\footnotesize Number of tasks (in thousands)};  

    \coordinate (c6) at ($(c2)!1.0!(c3)$);    
    \node[above] at (c6 |- current bounding box.north)
      {\pgfplotslegendfromname{fred}};    
\end{tikzpicture}
\caption{Latency of different priority tasks in Cloud environment under different weighted critical values \\($WV$) \scriptsize (a) $w_1=w_2=w_3$, (b) $w_1=0.2, w_2=0.3, w_3=0.5$, (c) $w_1=0.2, w_2=0.5, w_3=0.3$, (d) $w_1=0.5, w_2=0.3, w_3=0.2$} \vspace{-4mm}
\label{fig:sim:Ltncy_C}
\end{figure}
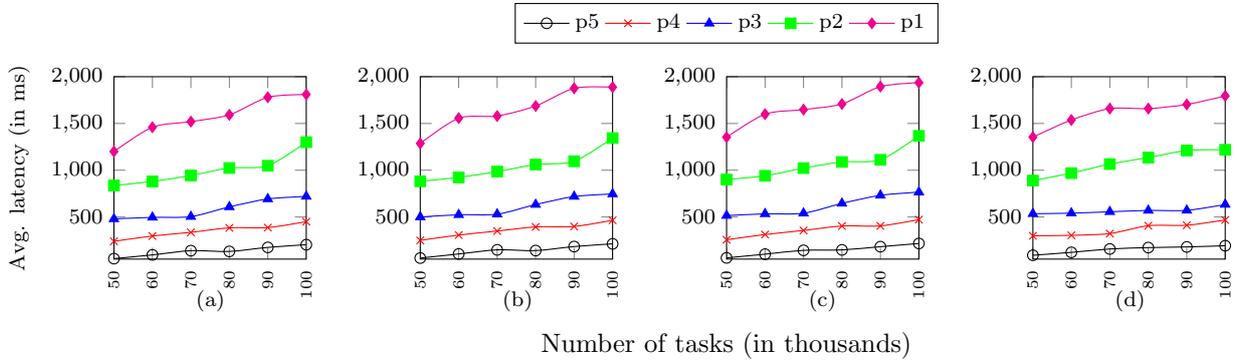

\revison{Under the first configuration, the average latency of tasks with $P_i=5$ and $P_i=1$ in only fog environment is approximately $11 ms$ and $93 ms$ when the number of tasks is $50$ thousand. The latency of those tasks increases to approx. $34 ms$ and $130ms$ when the number of tasks increases to $10000$, as shown in Figure }\ref{fig:sim:Ltncy_F}(a)\revison{. A similar pattern is observed under the forth configuration, where a higher weightage is given to the makespan parameter ($w_1=0.5, w_2=0.3, w_3=0.2$). The average latency of the tasks with priority ($P_i=5$) is approx. $8ms$, whereas the latency of $P_=1$ tasks is approx. $120ms$ when the number of tasks is $50$ thousand. From all the observations (Figure }\ref{fig:sim:Ltncy_F}(a) - \ref{fig:sim:Ltncy_F}(d)\revison{), we can conclude that under different configurations, the proposed HeRAFC algorithm gives a higher preference to the fog environment while fulfilling the resource demand and other objectives.}

\revison{Unlike Figure }\ref{fig:sim:Ltncy_F}\revison{, Figure }\ref{fig:sim:Ltncy_C}\revison{ shows the effect of weighted critical value ($WV$) and the priority of tasks upon the average latency in the cloud environment. Figure }\ref{fig:sim:Ltncy_C}(a)\revison{ represents the average latency when all the parameters are given equal weightage, i.e., $w_1=w_2=w_3$. Figure }\ref{fig:sim:Ltncy_C}(b)\revison{, }\ref{fig:sim:Ltncy_C}(c)\revison{, and }\ref{fig:sim:Ltncy_C}(d)\revison{ represent the average latency when a higher weightage is given to resource demand, priority, and makespan parameters, respectively. Similar to the results in Figure }\ref{fig:sim:Ltncy_F}\revison{, it is observed that the tasks with higher priority have less latency. For instance, as shown in Figure }\ref{fig:sim:Ltncy_C}(a)\revison{, the tasks with $P_i=5$ and that are sent to the cloud environment have an average latency of $55ms$ when the number of tasks is $50$ thousand. The latency increases to $204ms$ when the total number of tasks increases to $100$ thousand. On the contrary, the  have the average latency increases from $1200ms$ and $1810ms$ when the total number of tasks  with $P_i=1$ increases from $50$ thousand to $100$ thousand. The similar patterns can be seen in Figure }\ref{fig:sim:Ltncy_C}(b)\revison{ and }\ref{fig:sim:Ltncy_C}(c)\revison{, where the average latency of $P_i=5$ tasks hosted in cloud environment is approx. $61ms$ and $64ms$ when a higher weightage is given to resource demand and priority parameters, respectively. Moreover, when the $P_i=5$ tasks are increased to $100$ thousand, their average latency increases to approx. $213ms$ and $220ms$ under the same configurations, respectively. Tasks with $P_i=1$ have the maximum latency irrespective of the weightage value of the parameters.}



\pgfplotstableread[row sep=\\,col sep=&]{
noofApps &	  fog_p5 &	  fog_p4 &	  fog_p3 &	  fog_p2 &	  fog_p1 &	  cloud_p5 &	  cloud_p4 &	  cloud_p3 &	  cloud_p2 &	  cloud_p1\\
5 &	  80 &	  70 &	  60 &	  50 &	  25 &	  20 &	  30 &	  40 &	  50 &	  75\\
6 &	  79 &	  68 &	  55 &	  45 &	  20 &	  21 &	  32 &	  45 &	  55 &	  80\\
7 &	  76 &	  59 &	  50 &	  44 &	  19 &	  24 &	  41 &	  50 &	  56 &	  81\\
8 &	  72 &	  59 &	  48 &	  37 &	  14 &	  28 &	  41 &	  52 &	  63 &	  86\\
9 &	  71 &	  52 &	  45 &	  36 &	  14 &	  29 &	  48 &	  55 &	  64 &	  86\\
10 &	  70 &	  50 &	  42 &	  21 &	  13 &	  30 &	  50 &	  58 &	  79 &	  87\\
}\appTasksInFogAndCloud

\usepgfplotslibrary{groupplots}

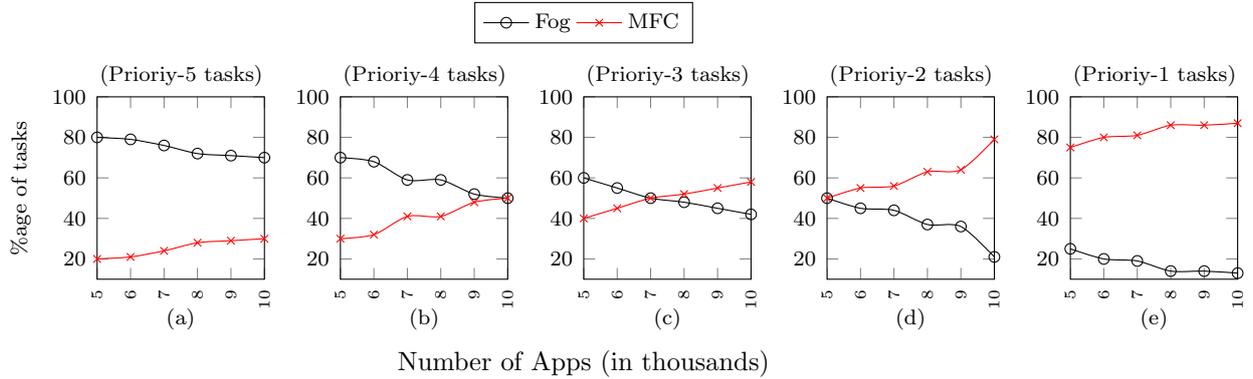
\begin{figure}[t]
\begin{tikzpicture}
\pgfplotsset{
    width=0.23\textwidth,
    height=40mm,
    xtick={5,6,7,8,9,10},
    xticklabels={5,6,7,8,9,10},
    x tick label style={rotate=90,anchor=east,font=\tiny},   
    xmin=5,xmax=10,
    ymin=10,ymax=100,
    ylabel style={align=center},
    xlabel style={below=-3mm},
    every axis/.append style={font=\scriptsize},
    every axis title/.style={above right,at={(-0.05,1)}},
}     
    \begin{groupplot}[ 
        group style={
            group size=5 by 1,     
        },
    ]
    \nextgroupplot[
            title={(Prioriy-5 tasks)},            
            xlabel={(a)},,
            ylabel={\%age of tasks},            
    ]
        \addplot[smooth,mark=o,black] table[x=noofApps,y=fog_p5]{\appTasksInFogAndCloud}; 
        \addplot[smooth,mark=x,red] table[x=noofApps,y=cloud_p5]{\appTasksInFogAndCloud}; 
    \nextgroupplot[
            title={(Prioriy-4 tasks)},
            xlabel={(b)},
    ]
        
        \addplot[smooth,mark=o,black] table[x=noofApps,y=fog_p4]{\appTasksInFogAndCloud}; 
        \addplot[smooth,mark=x,red] table[x=noofApps,y=cloud_p4]{\appTasksInFogAndCloud}; 
    \coordinate (c2) at (rel axis cs:0,0);
    \nextgroupplot[
            title={(Prioriy-3 tasks)},
            xlabel={(c)},
            legend style={legend columns=5,fill=none,draw=black,anchor=center,align=center, font=\scriptsize},
            legend to name=fred
    ]        
        \addplot[smooth,mark=o,black] table[x=noofApps,y=fog_p3]{\appTasksInFogAndCloud}; 
        \addplot[smooth,mark=x,red] table[x=noofApps,y=cloud_p3]{\appTasksInFogAndCloud}; 
    \addlegendentry{Fog};    
    \addlegendentry{MFC};    
    \coordinate (c3) at (rel axis cs:0,0);
    \nextgroupplot[
            xlabel={(d)},
            title={(Prioriy-2 tasks)},
    ]             
        \addplot[smooth,mark=o,black] table[x=noofApps,y=fog_p2]{\appTasksInFogAndCloud}; 
        \addplot[smooth,mark=x,red] table[x=noofApps,y=cloud_p2]{\appTasksInFogAndCloud}; 
    
    \nextgroupplot[
            xlabel={(e)},
            title={(Prioriy-1 tasks)},
    ]             
        \addplot[smooth,mark=o,black] table[x=noofApps,y=fog_p1]{\appTasksInFogAndCloud}; 
        \addplot[smooth,mark=x,red] table[x=noofApps,y=cloud_p1]{\appTasksInFogAndCloud}; 
    \end{groupplot}
    
    \coordinate (c6) at ($(c2)!1.0!(c3)$);    
    \node[below] at (c6 |- current bounding box.south)
      {\footnotesize Number of Apps (in thousands)};    

    \coordinate (c6) at ($(c2)!1.0!(c3)$);    
    \node[above] at (c6 |- current bounding box.north)
      {\pgfplotslegendfromname{fred}};    
\end{tikzpicture}
\caption{Percentage of tasks (of different priorities) allocated to only FNs and to MFC environments.}\vspace{-4mm}
\label{fig:sim:PerofTskInFogCld}
\end{figure}
\revison{It is essential to investigate the percentage of tasks allocated to fog and cloud environments; the same simulation results are presented in Figure }\ref{fig:sim:PerofTskInFogCld}\revison{. The X-axis represents the number of applications ranging from $5000$ to $10000$. The applications consist of different priority tasks. Y-axis represents the percentage of tasks (with specific priority values) assigned to fog or cloud environments. Figure }\ref{fig:sim:PerofTskInFogCld}\revison{ consists of five sub-figures, each for specific priority values. Figure }\ref{fig:sim:PerofTskInFogCld}(a)\revison{ shows the percentage of $P_i=5$ tasks processed by the FNs that might be in the nearby location or at a multi-hop distance and by the cloud environment. Similarly, Figure }\ref{fig:sim:PerofTskInFogCld}(b)\revison{ and }\ref{fig:sim:PerofTskInFogCld}(c)\revison{ represent the percentage of $P_i=4$ and $P_i=3$ tasks, respectively, that are processed by the FNs and by the cloud environment. The simulation result in all those sub-figures shows the tasks with higher priority are given higher preference to be processed by the FN. For instance, out of all the $P_i=5$ tasks (of $5000$ applications), approximately $80\%$ of the tasks are processed by FNs and rest $20\%$ are processed by the cloud. However, when the number of applications increases to $10000$, a decreasing percentage of the $P_i=5$ tasks are handled by FNs, i.e., $70\%$ of the tasks are processed by fog environment and the rest are handled by cloud, as shown in Figure }\ref{fig:sim:PerofTskInFogCld}(a). 

\revison{Similarly, in the case of $P_i=4$ tasks, it is observed that approx. $70\%$ of the tasks are assigned to the fog environment and the rest $30\%$ are assigned to the cloud when the total number of applications was $5000$. When the total number of applications increases to $10000$ (i.e. number of tasks is $100$ thousand), an approximately equal percentage of $P_i=4$ tasks are processed by FNs and by cloud environment, as shown in Figure }\ref{fig:sim:PerofTskInFogCld}(b)\revison{. Figure }\ref{fig:sim:PerofTskInFogCld}(e)\revison{ gives a result that is opposite to the result in }\ref{fig:sim:PerofTskInFogCld}(a)\revison{. In Figure }\ref{fig:sim:PerofTskInFogCld}(e)\revison{, the maximum number of tasks with $P_i=1$ are handled by the cloud. When the number of applications was $5000$, approx. $25\%$ of the $P_i=1$ tasks are processed by FNs and rest $75\%$ of the $P_i=1$ tasks are processed by the cloud. The percentage of $P_i=1$ tasks that are sent to cloud further increases to more than $85\%$ when the total number of applications increases to $10$ thousand. This is due to the fact that, when the number of applications increases, the number of tasks increases, and their resource demand increases. However, the resource availability in the fog environment is fixed to fulfill those increasing resource demands of all the tasks. As a result, there is a higher chance of tasks being sent to the cloud environment.} 

\pgfplotstableread[row sep=\\,col sep=&]{
noOfApps &	  WMDOrder &	  PriorityBasedOrder &	  RandomOrder\\
5 &	  30 &	  23 &	  21\\
5.5 &	  34 &	  26 &	  22\\
6 &	  38 &	  28 &	  25\\
6.5 &	  44 &	  35 &	  26\\
7 &	  57 &	  53 &	  29\\
7.5 &	  60 &	  56 &	  29\\
8 &	  72 &	  61 &	  32\\
8.5 &	  75 &	  62 &	  43\\
9 &	  81 &	  65 &	  45\\
9.5 &	  85 &	  71 &	  49\\
10 &	  86 &	  72 &	  50\\
}\taskOrderCompResrcUtil

\pgfplotstableread[row sep=\\,col sep=&]{
noOfApps &	  WMDOrder &	  PriorityBasedOrder &	  RandomOrder\\
5 &	  31 &	  26 &	  20\\
5.5 &	  36 &	  28 &	  26\\
6 &	  39 &	  32 &	  30\\
6.5 &	  42 &	  35 &	  30\\
7 &	  45 &	  36 &	  30\\
7.5 &	  48 &	  37 &	  31\\
8 &	  50 &	  41 &	  34\\
8.5 &	  55 &	  44 &	  36\\
9 &	  57 &	  44 &	  38\\
9.5 &	  58 &	  45 &	  40\\
10 &	  59 &	  45 &	  40\\
}\taskOrderNWResrcUtil
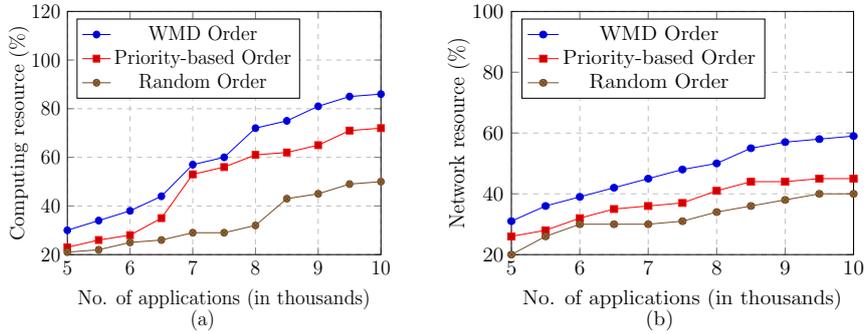
\begin{figure}[t]
	\centering
 \pgfplotsset{    
        every axis title/.style={below right,at={(0.37,-0.2)}},
    } 
\subfloat{\resizebox{0.35\textwidth}{45mm}{
    
    \begin{tikzpicture}
        \begin{axis}[
                width=0.5\textwidth,
                height=.4\textwidth,
                legend pos=north west,
                symbolic x coords={5, 5.5, 6, 6.5, 7, 7.5, 8, 8.5, 9, 9.5,  10},
                nodes near coords align={vertical},
                ymin=20,ymax=120,xmin=5,xmax=10,
                ylabel={Computing resource (\%)},
                y label style={at={(-0.1,0.5)}},
                title = {(a)},
                xlabel={No. of applications (in thousands)},ymajorgrids=true, 
                xmajorgrids=true, grid style=dashed,
            ]
            \addplot table[x=noOfApps,y=WMDOrder]{\taskOrderCompResrcUtil};
            \addplot table[x=noOfApps,y=PriorityBasedOrder]{\taskOrderCompResrcUtil};
            \addplot table[x=noOfApps,y=RandomOrder]{\taskOrderCompResrcUtil};
            \legend{WMD Order, Priority-based Order, Random Order}
        \end{axis}
    \end{tikzpicture} \hspace{5mm} \label{fig:sim:taskOrderCompResrcUtil}}}
\subfloat{\resizebox{0.35\textwidth}{45mm}{
    \begin{tikzpicture}
        \begin{axis}[
                width=0.5\textwidth,
                height=.4\textwidth,
                legend pos=north west,
                symbolic x coords={5, 5.5, 6, 6.5, 7, 7.5, 8, 8.5, 9, 9.5,  10},
                nodes near coords align={vertical},
                ymin=20,ymax=100,xmin=5,xmax=10,
                ylabel={Network resource (\%)},
                y label style={at={(-0.1,0.5)}},
                title = {(b)},
                xlabel={No. of applications (in thousands)},ymajorgrids=true, xmajorgrids=true, grid style=dashed
            ]
            \addplot table[x=noOfApps,y=WMDOrder]{\taskOrderNWResrcUtil};
            \addplot table[x=noOfApps,y=PriorityBasedOrder]{\taskOrderNWResrcUtil};
            \addplot table[x=noOfApps,y=RandomOrder]{\taskOrderNWResrcUtil};
            \legend{WMD Order, Priority-based Order, Random Order}
        \end{axis}
    \end{tikzpicture}\label{fig:sim:taskOrderNWResrcUtil}} }
	\caption{Impact of tasks' order on resource utilization: \protect \subref{fig:sim:taskOrderCompResrcUtil} average computing (CPU and memory) resource utilization and \protect \subref{fig:sim:taskOrderNWResrcUtil} average network resource utilization.}
	\label{fig:sim:TaskOrderRsrcUtil} 
    \vspace{-4mm}
\end{figure}

\revison{Algorithm }\ref{algo:HeRAFC:tskOrdr}\revison{ decides the order of tasks to be followed while allocating the resources. This paper also investigated the effect of tasks' order on resource utilization, as shown in Figure }\ref{fig:sim:TaskOrderRsrcUtil}\revison{. Along with the WMD (Weighted Multi-Dimensional)  approach, this paper investigated resource utilization when the order of tasks is decided randomly and solely based on the priority value. Figure }\ref{fig:sim:TaskOrderRsrcUtil}\subref{fig:sim:taskOrderCompResrcUtil}\revison{ and }\ref{fig:sim:TaskOrderRsrcUtil}\subref{fig:sim:taskOrderNWResrcUtil}\revison{ shows the computing and network resource utilization, respectively. The X-axis represents the number of applications (ranging from $5000$ to $10$ thousand) and the Y-axis represents the resource utilization in percentage. In the case of random order, the computing resource utilization ranges from approximately $20\%$ to $50\%$ with the number of applications increasing from $5000$ to $10000$ (}Figure \ref{fig:sim:TaskOrderRsrcUtil}\subref{fig:sim:taskOrderCompResrcUtil}\revison{). Similarly, the network resource utilization increases from $20\%$ to $40\%$, when the number of tasks increases from $5000$ to $10000$ (}\ref{fig:sim:TaskOrderRsrcUtil}\subref{fig:sim:taskOrderNWResrcUtil}\revison{). A better result can be observed when the tasks are ordered solely based on their priority values. The computing and network resource utilization increases from $23\%$ to $72\%$ and $26\%$ to $45\%$, when the number of tasks increases from $5000$ to $10000$. However, with the proposed WMD order, resource utilization can further be improved. The computing and network resource utilization increases from $30\%$ to $86\%$ and from $31\%$ to $59\%$, respectively, which is better than the other two approaches. The user decides the tasks' priority values, which do not include tasks' resource demand and makespan values. This affects the final resource utilization, as discussed before. Similarly, in the random order of tasks, none of the parameters are given any priority/importance, resulting in a degraded resource utilization compared to the WMD-based ordering tasks strategy.}   

\subsection{Uncertainty Analysis}
\revison{The impact of fluctuating resource availability at the FNs on the resource utilization is analyzed and the results are presented in }\ref{apndx:uncertainAna}.

\subsection{Simulation time analysis}
\revison{The time taken to handle different number of applications by both the algorithms is observed and the results are analyzed and presented in }\ref{apndx:simTimeAna}.

\section{Conclusions and future works}\label{sec:concls}
In this paper, we investigated the problem of resource allocation to users' applications consisting of multiple dependent tasks in MultiFog-Cloud environment. The proposed HeRAFC algorithm allows one FN to communicate with other FNs through Fog Cloud Interface. As a result, the service providers are allowed to choose a FN that is at multi-hop distance from the nearby FN to host and execute part of the application. Addressing the resource allocation problem in such environment, the aforementioned problem is first formulated as Integer Linear Programming problem and a novel heuristic HeRAFC algorithm is presented that facilitates the users to decide the importance of one task over other through priority values. For each task, the proposed algorithm explores the best hosting environment that is close to the environment where the adjacent tasks are deployed considering the makespan, latency and resource demand as the major parameters. To prove the efficiency, the simulation results of HeRAFC algorithm and other related algorithms are compared and presented. The performance of the proposed algorithm is evaluated in terms of different parameters such as average latency in fog and cloud, resource utilization in fog and cloud etc. 

In a similar manner to our previous work in \cite{MAHMUD2018_2018_j7}, a small experimental testbed will be used to implement the HeRAFC algorithm to make it more applicable to practical scenarios while minimizing energy consumption and improving QoS. We will consider applying the proposed HeRAFC to our proposed Indie Fog \cite{IndieFog} architecture. Additionally, HeRAFC should incorporate both FN mobility and end user applications. The development of a simulated testbed for fog/edge computing is currently underway using an opportunistic network emulator (STEP-ONE \cite{MASS2019100051}), which allows us to use different mobility models. As part of our future work, we will incorporate this dynamic behavior into HeRAFC.

\section*{Acknowledgment}

This research is supported by SERB, India, through grant CRG/2021/003888. We also thank financial support to UoH-IoE by MHRD, India (F11/9/2019-U3(A)).

\bibliographystyle{elsarticle-num}
\bibliography{references}

\begin{wrapfigure}{i}{1in}
\includegraphics[width=1in,height=1.25in,clip,keepaspectratio]{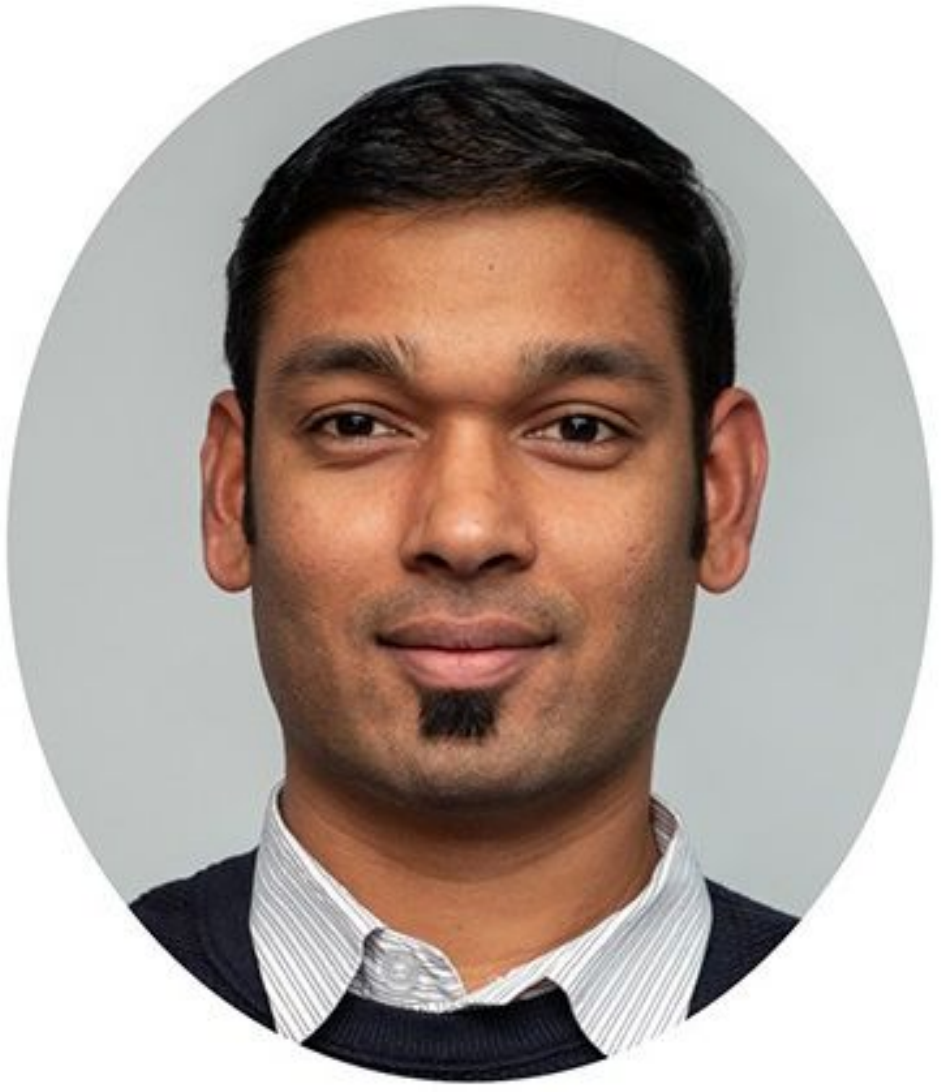}
\end{wrapfigure}
\textbf{Chinmaya Kumar Dehury} received a Bachelor’s degree from Sambalpur University, India, in June 2009 and a Master’s degree from Biju Pattnaik University of Technology, India, in June 2013. He received a Ph.D. Degree in the Department of Computer Science and Information Engineering, Chang Gung University, Taiwan. He is currently a Lecturer of Distributed System, member of Mobile \& Cloud Lab in the Institute of Computer Science, University of Tartu, Estonia. His research interests include scheduling, resource management and fault tolerance problems of Cloud and fog Computing, the application of artificial intelligence in cloud management, edge intelligence, Internet of Things, and data management frameworks. His research results are published by top-tier journals and transactions such as IEEE TCC, JSAC, TPDS, FGCS, etc. He is a member of IEEE and ACM India. He is also serving as a PC member of several conferences and reviewer to several journals and conferences, such as IEEE TPDS, IEEE JSAC, IEEE TCC, IEEE TNNLS, Wiley Software: Practice and Experience, etc. 

\begin{wrapfigure}{i}{1in}
\includegraphics[width=1in,height=1.25in,clip,keepaspectratio]{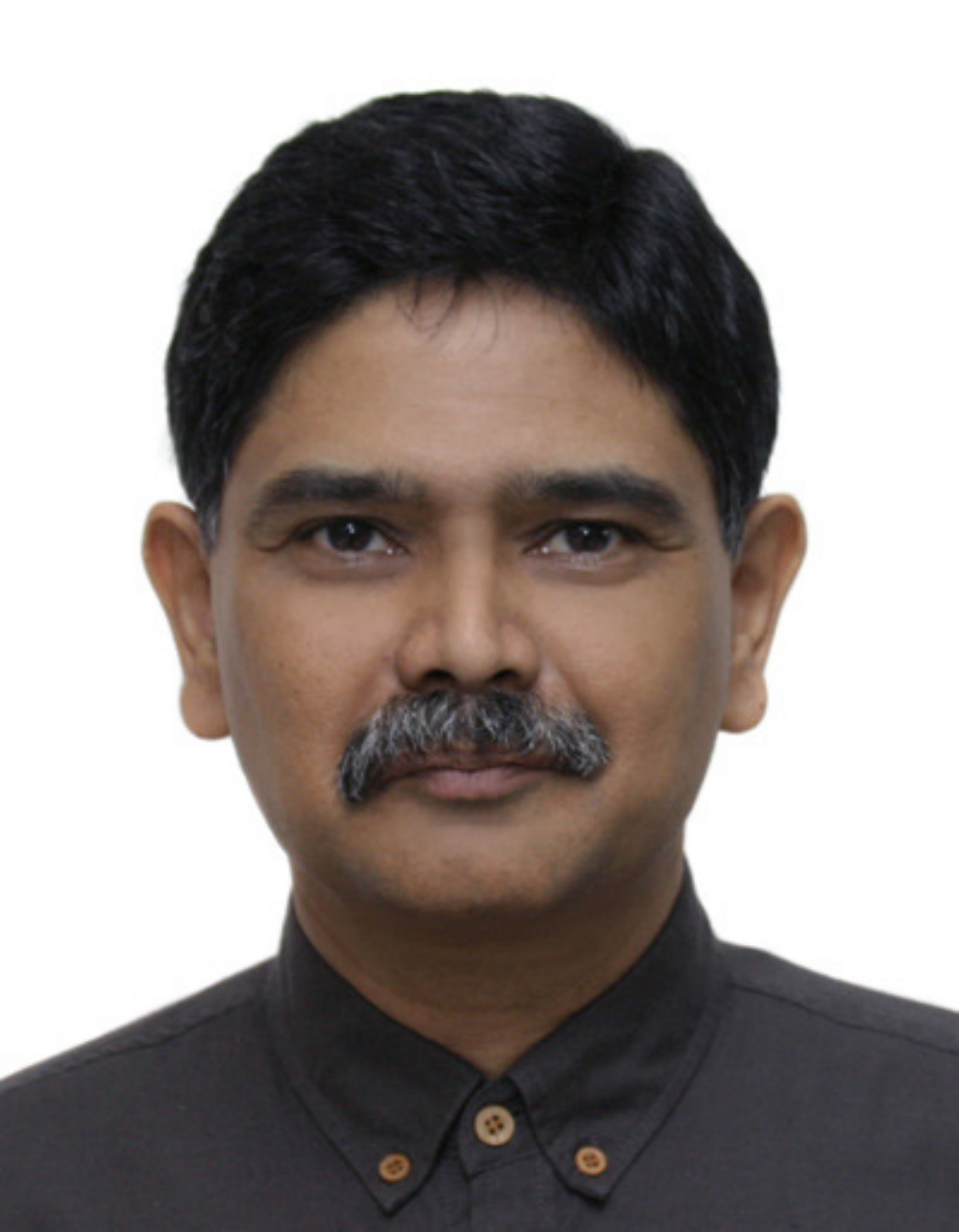}
\end{wrapfigure}
\textbf{Bharadwaj Veeravalli} received his BSc degree in Physics, from Madurai-Kamaraj University, India, in 1987, the Master's degree in Electrical Communication Engineering from the Indian Institute of Science, Bangalore, India in 1991, and the PhD degree from the Department of Aerospace Engineering, Indian Institute of Science, Bangalore, India, in 1994. He received gold medals for his bachelor degree overall performance and for an outstanding PhD thesis (IISc, Bangalore India) in the years 1987 and 1994, respectively. He is currently with the Department of Electrical and Computer Engineering, Communications and Information Engineering (CIE) division, at The National University of Singapore, Singapore, as a tenured Associate Professor. His main stream research interests include cloud/grid/cluster computing (big data processing, analytics and resource allocation), scheduling in parallel and distributed systems, Cybersecurity, and multimedia computing. He is one of the earliest researchers in the field of Divisible Load Theory (DLT). He is currently serving the editorial board of IEEE Transactions on Parallel and Distributed Systems as an associate editor. He is a senior member of the IEEE and the IEEE-CS.
	
\begin{wrapfigure}{i}{1in}
\includegraphics[width=1in,height=1.25in,clip,keepaspectratio]{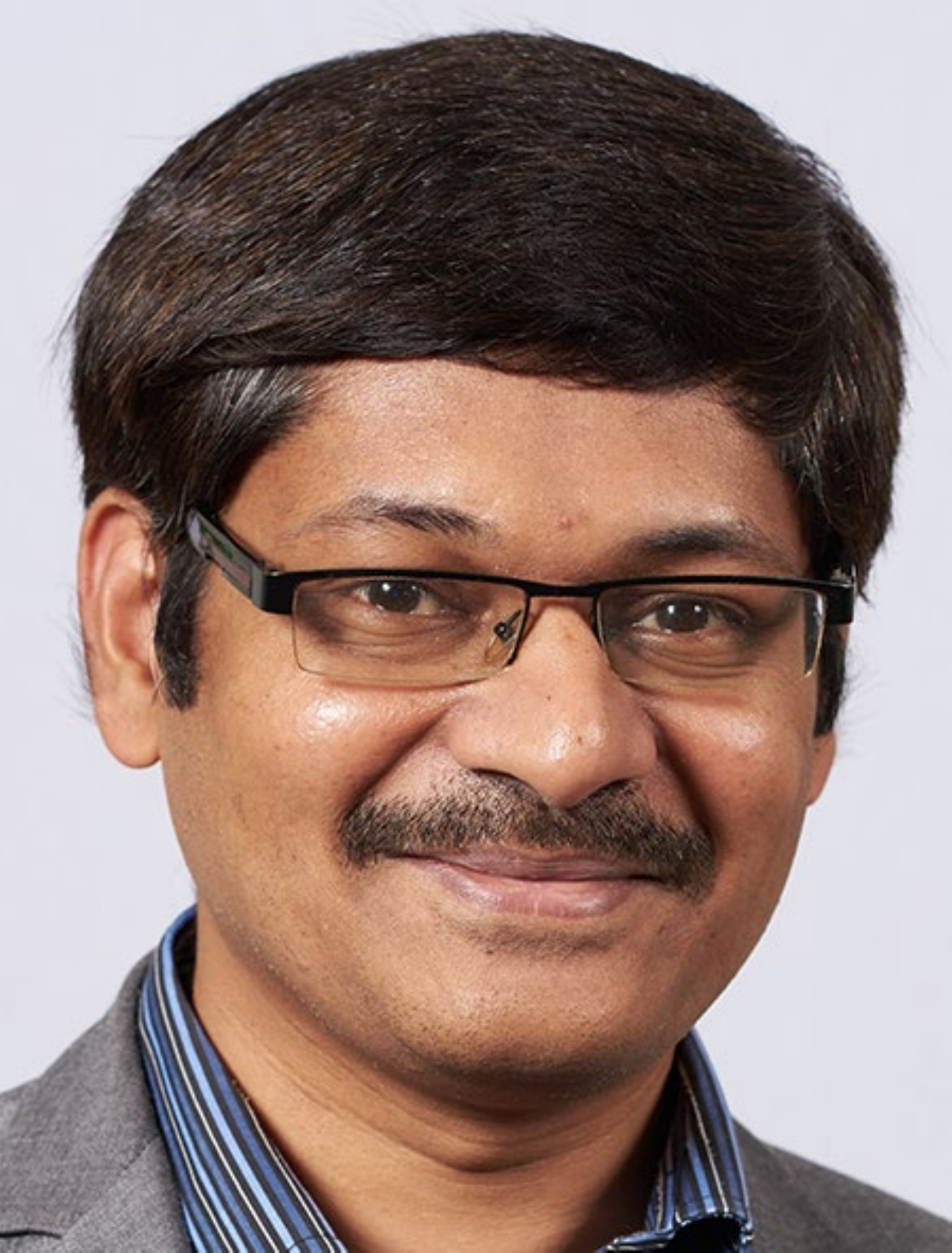}
\end{wrapfigure}
\textbf{Satish Narayana Srirama}
	is an Associate Professor at the School of Computer and Information Sciences, University of Hyderabad, India. He is also a Visiting Professor and the honorary head of the Mobile \& Cloud Lab at the Institute of Computer Science, University of Tartu, Estonia, which he led as a Research Professor until June 2020. He received his PhD in computer science from RWTH Aachen University, Germany in 2008. His current research focuses on cloud computing, mobile web services, mobile cloud, Internet of Things, fog computing, migrating scientific computing and enterprise applications to the cloud and large-scale data analytics on the cloud. He is IEEE Senior Member, an Editor of Wiley Software: Practice and Experience, a 52 year old Journal, was an Associate Editor of IEEE Transactions in Cloud Computing and a program committee member of several international conferences and workshops. Dr. Srirama has co-authored over 170 refereed scientific publications in international conferences and journals. 


%
\newpage
\setcounter{page}{1}
\appendix
\section{Time complexity of HeRAFC}
\subsection{Running time of Algorithm \ref{algo:HeRAFC:tskOrdr}}\label{appen:TimeCompAlgo1}
\begin{theorem}
	The time complexity of the Algorithm \ref{algo:HeRAFC:tskOrdr} is $ \mathcal{O}(|V'| \log{}|V'|)$ , where $V'$ be the set of tasks.
\end{theorem}
\begin{proof}
	Algorithm \ref{algo:HeRAFC:tskOrdr} focuses on ordering of the tasks for deployment purpose. $ST$ contains $|V'|$ number of tasks. Sorting of $|V'|$ tasks based on the number of out-edges depends on the sorting algorithm. This sorting time can be reduced to $\mathcal{O}(|V'| \log{}|V'|)$ with QUICKSORT algorithm, as in Line \ref{algo:HeRAFC:tskOrdr:sortTask}. Following this, extraction of the tasks in Line \ref{algo:HeRAFC:tskOrdr:extrctTaskLst} would need a constant time as the list of the tasks is already in the previous step. The looping block from Line \ref{algo:HeRAFC:tskOrdr:MCVstart}-\ref{algo:HeRAFC:tskOrdr:MCVend} would depend on the size of $TL$ set. Furthermore, the sorting process in Line \ref{algo:HeRAFC:tskOrdr:sortExtTskLst} would consume the time-complexity of $\mathcal{O}(|TL| \log{}|TL|)$. The extraction process in Line \ref{algo:HeRAFC:tskOrdr:extractPrntTask} would consume a running time of $\mathcal{O}(|TL|)$. The outer \emph{while} loop from Line \ref{algo:HeRAFC:tskOrdr:LoopStart}-\ref{algo:HeRAFC:tskOrdr:LoopEnd} would depend on the maximum depth of the graph and hence would iterate for $|V'|$ (which is nothing but the size of set $ST$) number of times. This situation would occur if all the tasks run in a sequential manner, resulting only one task in set $TL$ in each iteration. The running time of Algorithm \ref{algo:HeRAFC:tskOrdr} can be calculated as $\mathcal{O}(|V'| \log{}|V'|) + \mathcal{O}(|V'|))$, which can be written as $\mathcal{O}(|V'| \log{}|V'|)$, where $|V'|$ is the number of tasks.
\end{proof}

\subsection{Running time of Algorithm \ref{algo:HeRAFC}}\label{appen:TimeCompAlgo2}
\begin{theorem}
	The time complexity of the proposed HeRAFC algorithm (Algorithm \ref{algo:HeRAFC}) is $\mathcal{O}(|E'|\left[\log{}|E'| + |E|log(|V|)\right]$ , where $E$ and $E'$ are the set of physical edges and task edges, respectively and $V$ is the number of cloud and FNs.
\end{theorem}
\begin{proof}
	Algorithm \ref{algo:HeRAFC} depends on the output from the Algorithm \ref{algo:HeRAFC:tskOrdr} (Line \ref{algo:HeRAFC:callAlgo1}), which is having a running time of $\mathcal{O}(|V'| \log{}|V'|)$, where $|V'|$ is the number of tasks. Calculation of availability matrix takes a constant running time. The \emph{for} loop at Line \ref{algo:HeRAFC:nodeMap:innerLoop:start} iterate for each task in the set $ETL$. For each task, finding the FNs (Line \ref{algo:HeRAFC:nodeMap:innerLoop:start}-\ref{algo:HeRAFC:nodeMap:innerLoop:end}) that are in the communication range of user, at 1-hop distance, and at 2-hop distance would take the running time of $\mathcal{O}(|V|)$, where $|V|$ is the set of cloud and FNs as defined in Section \ref{sec:MFCmodel}. This process will be repeated for $|V'|$ number of times, the total running time can be calculated as $\mathcal{O}(|V'| * |V|)$.
	
	Sorting of edges of the task graph, as in Line \ref{algo:HeRAFC:EdgeSorting}, would take a running time of $\mathcal{O}(|E'| \log{}|E'|)$, where $|E'|$ is the number of task edges. For every task edge, an existing shortest path algorithm can be implemented. In Line \ref{algo:HeRAFC:shrtstPath}, assuming that Dijkstra shortest path algorithm is implemented to find the shortest path from FN $sh$ to $th$, the worst case running time for every task edge would be $\mathcal{O}((|V|+|E|)log(|V|))$, where $|E|$ is the number of physical edges among cloud and FNs. The looping block in Line \ref{algo:HeRAFC:EdgeMapping:start} would iterate for $|E'|$ number of times and hence the total running time can be calculated as $\mathcal{O}(|E'|*(|V|+|E|)log(|V|))$. The total running time of Algorithm \ref{algo:HeRAFC} can be calculated as follows.
	
	\begin{equation}
	\begin{aligned}
	\mathcal{O}(|V'| \log{}|V'|) + \mathcal{O}(|V'|*|V|) + \mathcal{O}(|E'| \log{}|E'|) & \\ \nonumber 
	+ \mathcal{O}(|E'|*(|V|+|E|)log(|V|)) & \\
	= \mathcal{O}(|V'| \log{}|V'|) + \mathcal{O}(|V'| * |V|) + \mathcal{O}(|E'| \log{}|E'|) & \\ 
	+ \mathcal{O}(|E'|*|V|log(|V|)) + \mathcal{O}(|E'|*|E|log(|V|)) \nonumber
	\end{aligned}
	\end{equation}
	In case of a larger task graph, the number of task edges, $|E'|$, increases significantly as compared to the number of tasks, $|V'|$. Similarly, when the number of cloud and FNs increases, relatively the number of physical edges, $|E|$, increases significantly. Considering this scenario, the total running time can be calculated as follows.
	\begin{equation}
	\begin{aligned}
	\mathcal{O}(|E'| \log{}|E'|) + \mathcal{O}(|E'|*|E|log(|V|)) & \\ 
	=  \mathcal{O}(|E'|\left[\log{}|E'| + |E|log(|V|)\right])  & \nonumber
	\end{aligned}
	\end{equation}
\end{proof}

\section{Simulated MFC environment and Applications}\label{appen:simEnv}
\revison{This section provides an in-depth understanding of programmatically generated application and MFC environment that following several distribution functions. The entire simulation environment basically consists of two primary components: Applications and the fog-Cloud environment. The resource configuration of the FCIs, fog, cloud and the users' applications are presented in JSON file. Separate JSON files are created for each user's app, including the tasks list, edges list, resource demand of tasks and edges, and other related information.} 

\subsection{Fog-Cloud Environment}
\pgfplotstableread[row sep=\\,col sep=&]{
noofFNs &	  avgCPUcapa &	  avgMEMcapa\\
50 &	  74.04280945 &	  147.5285839\\
100 &	  71.50642447 &	  145.2479015\\
150 &	  74.627629 &	  148.3136586\\
200 &	  70.85510558 &	  147.544122\\
250 &	  72.46060638 &	  150.1420846\\
300 &	  71.25753588 &	  146.6953972\\
350 &	  74.23895694 &	  147.1088198\\
400 &	  73.06415461 &	  146.7105056\\
450 &	  71.53028302 &	  145.2307475\\
500 &	  73.5263397 &	  147.7836599\\
}\envAnalysisFNresCapacity

\pgfplotstableread[row sep=\\,col sep=&]{
noofFNs &	  avgBWwithFCI &	  avgLatencywithFCI\\
50 &	  348.394 &	  74.807\\
100 &	  349.6455 &	  75.4174\\
150 &	  351.2589333 &	  75.89713333\\
200 &	  352.2171 &	  75.82075\\
250 &	  351.27576 &	  76.2922\\
300 &	  351.4957 &	  76.25306667\\
350 &	  351.4740571 &	  76.28311429\\
400 &	  351.36105 &	  76.049725\\
450 &	  351.3246667 &	  75.60968889\\
500 &	  350.77276 &	  75.61882\\
}\envAnalysisFNNWCapacity

\pgfplotstableread[row sep=\\,col sep=&]{
noofFCIs &	  totalFNs & avgFNs   \\
20 &	  57 &	  2.45\\
40 &	  102 &	  2.55\\
60 &	  150 &	  2.5\\
80 &	  205 &	  2.56\\
100 &	  259 &	  2.59\\
120 &	  306 &	  2.55\\
140 &	  350 &	  2.5\\
160 &	  396 &	  2.47\\
180 &	  452 &	  2.51\\
200 &	  498 &	  2.49\\
}\envAnalysisFCIwithFNs
\pgfplotstableread[row sep=\\,col sep=&]{
noofFCIs &	  avgBWwithCloud &	  avgLatencyWithCloud\\
20 &	  739.794 &	  147.0825\\
40 &	  737.909 &	  145.0325\\
60 &	  724.7745 &	  147.8718333\\
80 &	  725.080125 &	  145.947125\\
100 &	  732.9384 &	  148.0655\\
120 &	  724.6948333 &	  148.9795833\\
140 &	  721.0546429 &	  148.1541429\\
160 &	  730.1338125 &	  149.2205625\\
180 &	  723.9213889 &	  148.7714444\\
200 &	  725.701 &	  148.52455\\
}\envAnalysisFCIwithCloudNWres
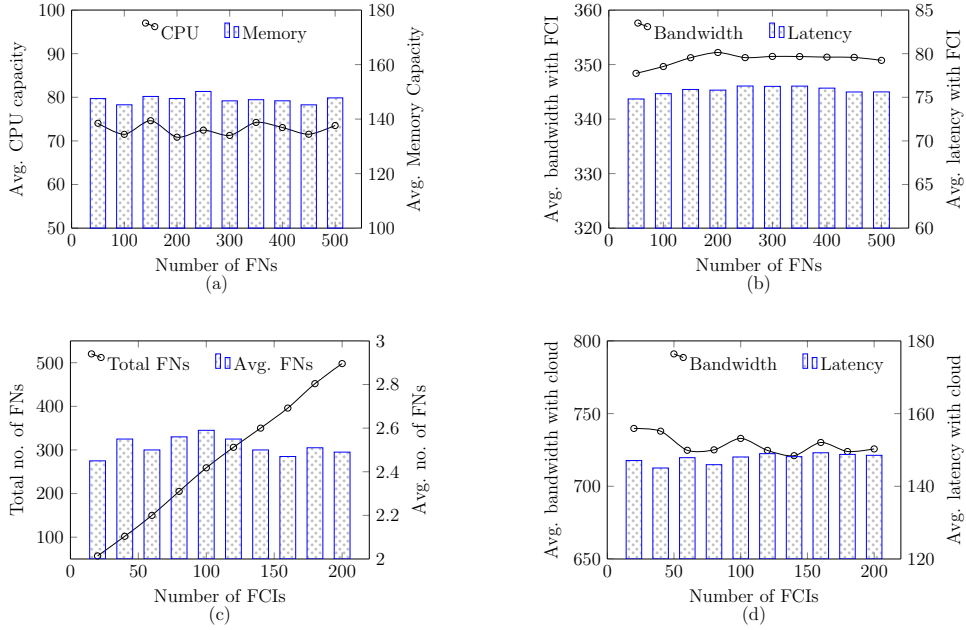
\begin{figure}[t]
	\centering
\subfloat{\resizebox{0.35\textwidth}{40mm}{
    \begin{tikzpicture}
        \begin{axis}[
                width=0.5\textwidth,
                height=.4\textwidth,                
                ymin=50,ymax=100,xmin=0,xmax=550,
                axis y line*=left,
                ylabel={Avg. CPU capacity},                
                xlabel={Number of FNs},
                title style={at={(0.5,-0.23)},anchor=north},
                title = (a)
            ]
            \addplot[smooth,mark=o,black] table[x=noofFNs,y=avgCPUcapa]{\envAnalysisFNresCapacity}; 
            \label{CPU} 
        \end{axis}
        \begin{axis}[
              ybar,            
              yticklabel pos=right,              
              axis y line*=right,
              width=0.5\textwidth,
              height=.4\textwidth,
              axis x line=none,
              ymin=100, ymax=180,
              xmin=0,xmax=550,
              ylabel={Avg. Memory Capacity},
              legend columns=2,
              legend style={at={(0.84,0.97)}, draw=none,  /tikz/every even column/.append style={column sep=5mm}},
        ]
        \addlegendimage{/pgfplots/refstyle=CPU}\addlegendentry{CPU}
        \addplot[fill opacity=0.5, color=blue, pattern=crosshatch dots] table[x=noofFNs,y=avgMEMcapa]{\envAnalysisFNresCapacity};  
        \addlegendentry{Memory}
        \end{axis}
    \end{tikzpicture} \label{fig:sim:envAnalysis:FNresCapacity}}}\hspace{1cm}
\subfloat{
    \resizebox{0.35\textwidth}{40mm}{
    \begin{tikzpicture}
        \begin{axis}[              
                width=0.5\textwidth,
                height=.4\textwidth,                
                ymin=320,ymax=360,xmin=0,xmax=550,
                axis y line*=left,
                ylabel={Avg. bandwidth with FCI},                
                xlabel={Number of FNs},
                title style={at={(0.5,-0.23)},anchor=north},
                title = (b)
            ]
            \addplot[smooth,mark=o,black] table[x=noofFNs,y=avgBWwithFCI]{\envAnalysisFNNWCapacity}; 
            \label{BW} 
        \end{axis}
        \begin{axis}[
              ybar,            
              yticklabel pos=right,              
              axis y line*=right,
              width=0.5\textwidth,
              height=.4\textwidth,
              axis x line=none,
              ymin=60, ymax=85,
              xmin=0,xmax=550,
              ylabel={Avg. latency with FCI},
              legend columns=2,
              legend style={at={(0.84,0.97)}, draw=none, /tikz/every even column/.append style={column sep=5mm}},
        ]
        \addlegendimage{/pgfplots/refstyle=BW}\addlegendentry{Bandwidth}
        \addplot[fill opacity=0.5, color=blue, pattern=crosshatch dots] table[x=noofFNs,y=avgLatencywithFCI]{\envAnalysisFNNWCapacity};  
        \addlegendentry{Latency}
        \end{axis}        
    \end{tikzpicture}\label{fig:sim:envAnalysis:FNNWCapacity}} } \\
\subfloat{\resizebox{0.35\textwidth}{40mm}{
    \begin{tikzpicture}
        \begin{axis}[              
                width=0.5\textwidth,
                height=.4\textwidth,                
                ymin=50,ymax=550,xmin=0,xmax=220,
                axis y line*=left,
                ylabel={Total no. of FNs},                
                xlabel={Number of FCIs},
                title style={at={(0.5,-0.23)},anchor=north},
                title = (c)
            ]
            \addplot[smooth,mark=o,black] table[x=noofFCIs,y=totalFNs]{\envAnalysisFCIwithFNs}; 
            \label{FNs} 
        \end{axis}
        \begin{axis}[
              ybar,            
              yticklabel pos=right,              
              axis y line*=right,
              width=0.5\textwidth,
              height=.4\textwidth,
              axis x line=none,
              ymin=2, ymax=3,xmin=0,xmax=220,
              ylabel={Avg. no. of FNs},
              legend columns=2,
              legend style={at={(0.84,0.97)}, draw=none, /tikz/every even column/.append style={column sep=5mm}},
        ]
        \addlegendimage{/pgfplots/refstyle=FNs}\addlegendentry{Total FNs}
        \addplot[fill opacity=0.5, color=blue, pattern=crosshatch dots] table[x=noofFCIs,y=avgFNs]{\envAnalysisFCIwithFNs};  
        \addlegendentry{Avg. FNs}
        \end{axis} 
    \end{tikzpicture}\label{fig:sim:envAnalysis:FCIwithFNs}} } \hspace{1cm}
\subfloat{\resizebox{0.35\textwidth}{40mm}{
    \begin{tikzpicture}
        \begin{axis}[              
                width=0.5\textwidth,
                height=.4\textwidth,                
                ymin=650,ymax=800,xmin=0,xmax=220,
                axis y line*=left,
                ylabel={Avg. bandwidth with cloud},                
                xlabel={Number of FCIs},
                title style={at={(0.5,-0.23)},anchor=north},
                title = (d)
            ]
            \addplot[smooth,mark=o,black] table[x=noofFCIs,y=avgBWwithCloud]{\envAnalysisFCIwithCloudNWres}; 
            \label{BW} 
        \end{axis}
        \begin{axis}[
              ybar,            
              yticklabel pos=right,              
              axis y line*=right,
              width=0.5\textwidth,
              height=.4\textwidth,
              axis x line=none,
              ymin=120, ymax=180,xmin=0,xmax=220,
              ylabel={Avg. latency with cloud},
              legend pos=north east,
              legend columns=2,
              legend style={draw=none, /tikz/every even column/.append style={column sep=5mm}}
        ]
        \addlegendimage{/pgfplots/refstyle=BW}\addlegendentry{Bandwidth}
        \addplot[fill opacity=0.5, color=blue, pattern=crosshatch dots] table[x=noofFCIs,y=avgLatencyWithCloud]{\envAnalysisFCIwithCloudNWres};  
        \addlegendentry{Latency}
        \end{axis} 
    \end{tikzpicture}\label{fig:sim:envAnalysis:FCIwithCloudNWres}} }
    \caption{Analysis of the MFC environment (i.e. FNs, FCIs, and cloud) and their resource availability and connectivity.}
	\label{fig:sim:envAnalysis} 
\end{figure}

\revison{The MFC environment consists of $500$ FNs that are connected to $200$ FCIs. Each FCI is connected to at least one FN, whereas a single FN is connected to precisely one FCI. For each FN, the total amount of CPU capacity ranges from 50 through 100 and the values are assigned by following a random distribution. The total memory (RAM) for each FN is assigned randomly ranging from $200 GB$ to $400 GB$, where the unit $GB$ represents GigaByte. The approximate computation power of a FN is calculated as MIPS (Million instructions per  second), which ranges from $3000$ to $5000$.} 
\revison{As discussed before, each FN is connected to exactly one FCI. The maximum network bandwidth available between a FN and the corresponding FCI ranges from $300 Mbps$ through $400 Mbps$. In this paper, it is considered that a FCI may directly communicate with another FCI. In such a situation, there exist a network bandwidth capacity among FCIs, which ranges from $400 Mbps$ through $1000 Mbps$. Similarly, the bandwidth between FCIs and the cloud range from $400 Mbps$ through $1000 Mbps$. The network bandwidth unit is in Mbps (Mega bits Per Second). It is assumed that latency value between any FN and the corresponding FCI is small than the latency among FCIs and between FCIs and cloud. Latency values are assigned to the physical links by following a uniform distribution. The latency (in millisecond) between FN and FCIs ranges from $50 ms$ through $100 ms$. Similarly, the latency values for rest of the physical links range from $100 ms$ through $200 ms$. Technically, a task can be hosted on 1-hop or at a multi-hop distance. However, in the current implementation maximum number of hops is set to $2$. This means there are a maximum of two FCIs and no cloud are present in the physical path between any two FNs that are used to host user's tasks.}

\revison{Figure }\ref{fig:sim:envAnalysis}(a)\revison{ and }\ref{fig:sim:envAnalysis}(b)\revison{ present the computing resource and network resource distribution, respectively, among the FNs. Figure }\ref{fig:sim:envAnalysis}(a)\revison{ shows the average CPU and memory resource availability among the FNs. Though the minimum and maximum CPU resource assigned to a FN is $50$ and $100$, respectively, the average CPU resource available is approximately $75$ and $73$ when there are $50$ and $500$ FNs, respectively. Similarly, the average memory resource available ranges between $145GB$ and $150GB$ with the same number of FNs, as shown in Figure }\ref{fig:sim:envAnalysis}(a)\revison{. Figure }\ref{fig:sim:envAnalysis}(b)\revison{ represents the relationship between the number of FNs and the average bandwidth and latency with the connected nearby FCI. The x-axis of the figure represents the number of FNs, while the y-axis represents the average bandwidth and average latency with FCI as two separate series. The figure shows that as the number of FNs increases, the average bandwidth with FCI remains relatively constant, fluctuating around the value between $348Mbps$ and $353Mbps$. However, the average latency with FCI also shows a slightly increasing trend, from $75.4 ms$ for $50$ FNs to $75.6 ms$ for $500$ FNs, with a maximum of $76.2$ for $350$ FNs. This means that as the number of FNs increases, the average latency with FCI slightly increases, but it remains relatively low overall. Figure }\ref{fig:sim:envAnalysis}(b)\revison{ indicates that increasing the number of FNs has little impact on the average bandwidth with FCI, while it slightly increases the average latency with FCI.}

\revison{Figure }\ref{fig:sim:envAnalysis}(c)\revison{ and }\ref{fig:sim:envAnalysis}(d)\revison{ show the relationship of the FCIs with the number of FNs and the network resource availability with the cloud. In Figure }\ref{fig:sim:envAnalysis}(c)\revison{, the x-axis denotes the number of FCIs, while the left y-axis represents the total number of FNs and the right-side y-axis represents the average number of FNs that are connected to the FCI. This shows how the FCIs are distributed and connected with all the $500$ FNs in the simulated environment. The total number of FNs connected to the FCIs is constantly increasing and as expected, $498$ number of FNs are connected to at least one FCI. In the simulated environment, only two FNs are directly connected with cloud. In other words, two FNs have no nearby FCI. The average number of FNs that are connected to the FCIs ranges between $2.4$ and $2.6$, as shown in Figure }\ref{fig:sim:envAnalysis}(c)\revison{. Each connection between an FCI and the cloud is associated with a certain amount of network bandwidth and latency, as shown in Figure }\ref{fig:sim:envAnalysis}(d)\revison{. The minimum and maximum network bandwidth available with any pair of FCI and cloud connection is $400 Mbps$ and $1000 Mbps$, respectively. However, the average bandwidth available among FCIs and the cloud ranges between $720Mbps$ and $740Mbps$. Similarly, in the simulated environment, the average latency between FCI and cloud ranges between $145ms$ and $150ms$, as shown in Figure }\ref{fig:sim:envAnalysis}(d). 

\subsection{Users Application}
\revison{The maximum number of applications is set to $10000$. For each application, random number of tasks are generated ranging from $4$ through $12$. However, the total maximum number of tasks for the entire simulation is set to $100,000$. For each task, the CPU and memory demand range from $1-4 CPUs$ and $100 MB - 1000 MB$, respectively. The average makespan of each task is $350 ms$ with minimum and maximum makespan of $10 ms$ and $1000 ms$, respectively. The network bandwidth demand of edges is ranging from $100 Mbps$ and $200 Mbps$. Similarly, the average latency demand among tasks ranges between $10 ms$ and $50 ms$. The unit of latency is millisecond (ms). The minimum and maximum priority values of a task are $1$ and $5$, respectively. The bandwidth and latency demand of a task ranges between $100-200 Mbps$ and $10-300 ms$, respectively. The number of edges within an application is decided by a link probability of $0.6$. This indicates the connectivity probability among two task. }

\pgfplotstableread[row sep=\\,col sep=&]{
priority &	  totalNoOfTasks &	  avgTasksPerApp\\
1 &	  19154 &	  1.599398776\\
2 &	  20463 &	  2.214058334\\
3 &	  21184 &	  2.647084372\\
4 &	  17123 &	  1.506067589\\
5 &	  22076 &	  3.755280649\\
}\AppAnalysisPriorityTasks

\pgfplotstableread[row sep=\\,col sep=&]{
priority &	  avgCPUReq &	  avgMemReq &	  AvgMakespan\\
1 &	  1.891995851 &	  546.0494996 &	  503.8287706\\
2 &	  2.695035332 &	  545.5481142 &	  505.8583308\\
3 &	  3.085399336 &	  550.0472737 &	  503.4208021\\
4 &	  2.296676673 &	  549.1707227 &	  505.9124278\\
5 &	  2.36234532 &	  547.1823416 &	  503.0895224\\
}\AppAnalysisPriorityResourceReq

\pgfplotstableread[row sep=\\,col sep=&]{
NoOfApps &	  avgPriority &	  avgNoOfTasks\\
5 &	  3.434 &	  7.513\\
5.5 &	  3.127 &	  7.625\\
6 &	  2.933 &	  6.801\\
6.5 &	  3.081 &	  6.983\\
7 &	  2.739 &	  7.822\\
7.5 &	  2.948 &	  6.828\\
8 &	  3.192 &	  7\\
8.5 &	  2.803 &	  7.265\\
9 &	  3.253 &	  7.512\\
9.5 &	  2.858 &	  7.041\\
10 &	  3.009 &	  6.909\\
}\AppAnalysisAppTasksPriority
\pgfplotstableread[row sep=\\,col sep=&]{
NoOfApps &	  avgCPUReq &	  avgMemReq &	  avgMakespan\\
5 &	  2.824 &	  541.052 &	  502.227\\
5.5 &	  2.9 &	  546.304 &	  500.563\\
6 &	  1.739 &	  540.847 &	  500.918\\
6.5 &	  3.024 &	  544.11 &	  502.741\\
7 &	  2.132 &	  547.031 &	  501.629\\
7.5 &	  2.428 &	  547.946 &	  502.551\\
8 &	  2.731 &	  541.993 &	  502.793\\
8.5 &	  2.238 &	  541.878 &	  502.023\\
9 &	  2.825 &	  549.786 &	  502.622\\
9.5 &	  3.063 &	  547.29 &	  500.542\\
10 &	  2.098 &	  542.094 &	  500.674\\
}\AppAnalysisAppResourceReq
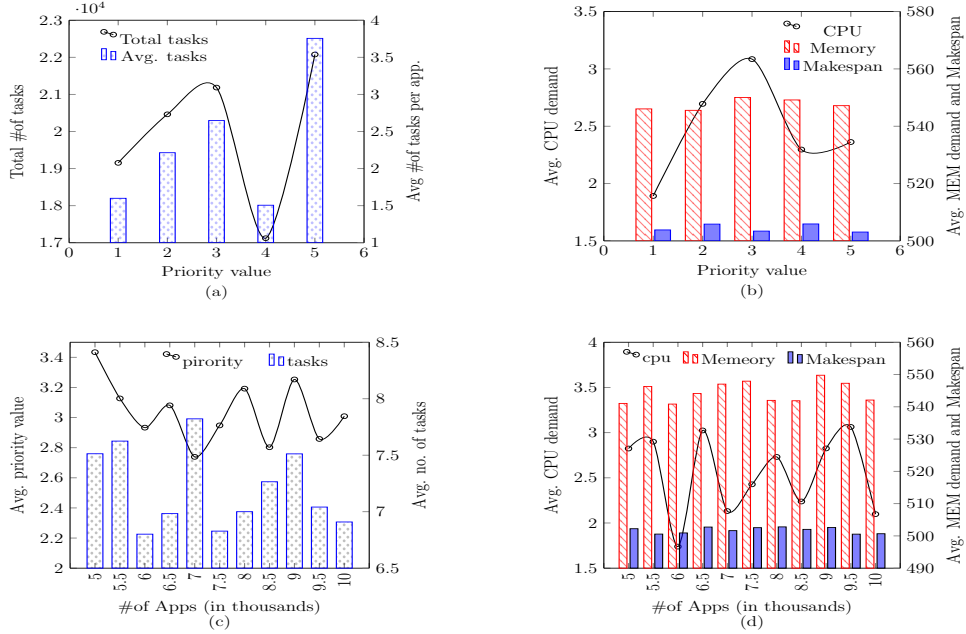
\begin{figure}[t]
	\centering
\subfloat{\resizebox{0.35\textwidth}{40mm}{
    \begin{tikzpicture}
        \begin{axis}[
                width=0.5\textwidth,
                height=0.5\textwidth,                
                ymin=17000,ymax=23000,xmin=0,xmax=6,
                axis y line*=left,
                ylabel={Total \#of tasks},                
                xlabel={Priority value},
                title style={at={(0.5,-0.2)},anchor=north},
                title = (a),
            ]
            \addplot[smooth,mark=o,black] table[x=priority,y=totalNoOfTasks]{\AppAnalysisPriorityTasks}; 
            \label{totalNoOfTasks} 
        \end{axis}
        \begin{axis}[
              ybar,            
              yticklabel pos=right,              
              axis y line*=right,
              width=0.5\textwidth,
              height=0.5\textwidth,
              axis x line=none,
              ymin=1, ymax=4,
              xmin=0,xmax=6,
              ylabel={Avg \#of tasks per app.},
              legend style={at={(0.5,0.97)}, draw=none, /tikz/every even column/.append style={column sep=5mm}},
        ]
        \addlegendimage{/pgfplots/refstyle=totalNoOfTasks}\addlegendentry{Total tasks}
        \addplot[fill opacity=0.5, color=blue, pattern=crosshatch dots,pattern color=blue] table[x=priority,y=avgTasksPerApp]{\AppAnalysisPriorityTasks};  
        \addlegendentry{Avg. tasks}
        \end{axis}
    \end{tikzpicture} \label{fig:sim:appAnalysis:PriorityTasks}}}\hspace{1cm}
\subfloat{
    \resizebox{0.35\textwidth}{40mm}{
    \begin{tikzpicture}
        \begin{axis}[              
                width=0.5\textwidth,
                height=.5\textwidth,                
                ymin=1.5,ymax=3.5,xmin=0,xmax=6,
                axis y line*=left,
                ylabel={Avg. CPU demand},                
                xlabel={Priority value},
                title style={at={(0.5,-0.2)},anchor=north},
                title = (b),
            ]
            \addplot[smooth,mark=o,black] table[x=priority,y=avgCPUReq]{\AppAnalysisPriorityResourceReq}; 
            \label{cpu} 
        \end{axis}        
        \begin{axis}[
              ybar,            
              yticklabel pos=right,              
              axis y line*=right,
              width=0.5\textwidth,
              height=.5\textwidth,
              axis x line=none,
              ymin=500,ymax=580,
              xmin=0,xmax=6,
              ylabel={Avg. MEM demand and Makespan},
              legend style={at={(0.97,0.97)}, draw=none, /tikz/every even column/.append style={column sep=5mm}},
        ]
        \addlegendimage{/pgfplots/refstyle=CPU}\addlegendentry{CPU}
        \addplot[fill opacity=0.5, color=red, pattern=north west lines,pattern color=red] table[x=priority,y=avgMemReq]{\AppAnalysisPriorityResourceReq};
        \addlegendentry{Memory}
        \addplot[fill opacity=0.5, color=blue, fill=blue] table[x=priority,y=AvgMakespan]{\AppAnalysisPriorityResourceReq};  
        \addlegendentry{Makespan}
        \end{axis}        
    \end{tikzpicture}\label{fig:sim:TasksPriority:PriorityResourceReq}} } \\
\subfloat{\resizebox{0.35\textwidth}{40mm}{
    \begin{tikzpicture}
        \begin{axis}[              
                width=0.5\textwidth,
                height=.5\textwidth,                
                ymin=2.0,ymax=3.5,
                axis y line*=left,
                ylabel={Avg. priority value},                
                xlabel={\#of Apps (in thousands)},
                xtick={4.5, 5, 5.5, 6, 6.5, 7, 7.5, 8, 8.5, 9, 9.5, 10, 10.5},
                xticklabels={, 5, 5.5, 6, 6.5, 7, 7.5, 8, 8.5, 9, 9.5, 10, },
                x tick label style={rotate=90,anchor=east},
                title style={at={(0.5,-0.22)},anchor=north},
                title = (c),
            ]
            \addplot[smooth,mark=o,black] table[x=NoOfApps,y=avgPriority]{\AppAnalysisAppTasksPriority}; 
            \label{pirority} 
        \end{axis}
        \begin{axis}[
              ybar,            
              yticklabel pos=right,              
              axis y line*=right,
              width=0.5\textwidth,
              height=.5\textwidth,
              axis x line=none,
              ymin=6.5, ymax=8.5, 
              xtick={4.5, 5, 5.5, 6, 6.5, 7, 7.5, 8, 8.5, 9, 9.5, 10, 10.5},
              xticklabels={, 5, 5.5, 6, 6.5, 7, 7.5, 8, 8.5, 9, 9.5, 10, },
              x tick label style={rotate=90,anchor=east},
              ylabel={Avg. no. of tasks},
              legend columns=2,
              legend style={at={(0.89,0.97)}, draw=none, /tikz/every even column/.append style={column sep=5mm}},
        ]
        \addlegendimage{/pgfplots/refstyle=pirority}\addlegendentry{pirority}
        \addplot[fill opacity=0.5, color=blue, pattern=crosshatch dots] table[x=NoOfApps,y=avgNoOfTasks]{\AppAnalysisAppTasksPriority};  
        \addlegendentry{tasks}
        \end{axis} 
    \end{tikzpicture}\label{fig:sim:appAnalysis:TasksPriority}} } \hspace{1cm}
\subfloat{\resizebox{0.35\textwidth}{40mm}{
    \begin{tikzpicture}
        \begin{axis}[              
                width=0.5\textwidth,
                height=.5\textwidth,                
                ymin=1.5,ymax=4,                
                xtick={4.5, 5, 5.5, 6, 6.5, 7, 7.5, 8, 8.5, 9, 9.5, 10, 10.5},
                xticklabels={, 5, 5.5, 6, 6.5, 7, 7.5, 8, 8.5, 9, 9.5, 10, },
                x tick label style={rotate=90,anchor=east},
                axis y line*=left,
                ylabel={Avg. CPU demand},                
                xlabel={\#of Apps (in thousands)},
                title style={at={(0.5,-0.22)},anchor=north},
                title = (d),
            ]
            \addplot[smooth,mark=o,black] table[x=NoOfApps,y=avgCPUReq]{\AppAnalysisAppResourceReq}; 
            \label{cpu} 
        \end{axis}
        \begin{axis}[
              ybar,            
              yticklabel pos=right,              
              axis y line*=right,
              width=0.5\textwidth,
              height=.5\textwidth,
              axis x line=none,
              ymin=490, ymax=560,            
                xtick={4.5, 5, 5.5, 6, 6.5, 7, 7.5, 8, 8.5, 9, 9.5, 10, 10.5},
                xticklabels={, 5, 5.5, 6, 6.5, 7, 7.5, 8, 8.5, 9, 9.5, 10, },
                x tick label style={rotate=90,anchor=east},
              ylabel={Avg. MEM demand and Makespan},
              legend columns=3,              
              legend style={draw=none, /tikz/every even column/.append style={column sep=2mm}}
        ]
        \addlegendimage{/pgfplots/refstyle=cpu}\addlegendentry{cpu}
        \addplot[bar width=5, fill opacity=0.5, color=red,pattern=north west lines,pattern color=red] table[x=NoOfApps,y=avgMemReq]{\AppAnalysisAppResourceReq};  \addlegendentry{Memeory}
        \addplot[bar width=5, fill opacity=0.5, fill=blue] table[x=NoOfApps,y=avgMakespan]{\AppAnalysisAppResourceReq};  
        \addlegendentry{Makespan}
        \end{axis} 
    \end{tikzpicture}\label{fig:sim:appAnalysis:resourceReq}} }
    \caption{Analysis of the applications and their tasks including the priorities and resource demand. }
	\label{fig:sim:appAnalysis} 
\end{figure}

\revison{Figure }\ref{fig:sim:appAnalysis}\revison{ gives an insight into the applications and their tasks, including the priorities and the resource demand. Figure B.11a shows the total and average number of tasks with different priority values. The priority values are assigned to the tasks in a programmatic manner. From Figure }\ref{fig:sim:appAnalysis}(a)\revison{, it can be observed that there are a total of $19154$, $20463$, and $21184$ number of tasks having priority $1$, $2$, and $3$. With a priority value $4$ there are a total of $17123$ number of tasks. In the entire simulation, the highest number of tasks have a priority value $5$, i.e. $22076$ tasks. The average number of tasks per application with a specific priority value is also analyzed and presented in the same figure. An average of $1.5$ number of tasks per application are with priority $4$, which is the lowest. On the other hand, an average of $3.75$ number of tasks per application have the priority $5$, which is the maximum, as shown in Figure }\ref{fig:sim:appAnalysis}(a)\revison{. Figure }\ref{fig:sim:appAnalysis}(b)\revison{ gives the tasks' resource configuration and makespan information with their priority. It can be seen that the tasks with priority $3$ are having highest CPU and memory demand, which is $3$ and $550MB$, respectively. Moreover, such tasks have comparatively lower makespan requirement, which is $503 ms$. The tasks with priority $2$ and $4$ have the higher makespan requirement of $505.8 ms$ and $505.9 ms$, respectively.}

\revison{The applications with average priority and the average number of tasks are also analyzed in Figure }\ref{fig:sim:appAnalysis}(c)\revison{. It can be observed that the average priority of the applications ranges between $2.5$ and $3.5$. The average number of tasks of $5000$ applications is $7.5$, whereas this number is reduced to $6.9$ when all the $10$ thousand applications are taken into account. The applications are further analyzed in terms of their resource (CPU, memory) demand and makespan requirement in Figure }\ref{fig:sim:appAnalysis}(d)\revison{. The average CPU demand with a different number of applications ranges between $1.7$ and $3.0$, whereas the memory demand ranges between $540 MB$ and $550 MB$. The average makespan requirement of all $10$  thousand applications is $500.6 ms$. The pattern shows that the distribution of CPU demand, memory demand, and makespan requirements are independent of each other.}

\section{Uncertainty Analysis}\label{apndx:uncertainAna}
\subsection{Uncertain resource availability impact on utilization}

\revison{Figure }\ref{fig:sim:uncertaintyAnalysis}\revison{ represents the relationship of fluctuation of resource availability, including network bandwidth and latency among FNs, FCIs, and cloud with their utilization. In other simulation instances, the HeRAFC algorithm takes the resource availability values of the entire environment, which remains unchanged until the algorithm handles the last application. For example, suppose the CPU availability of a FN is $75\%$ of its maximum capacity before HeRAFC starts processing the first application. In that case, CPU availability remain unchanged until the last application, provided no task is assigned to that FN.} 

\revison{It is implemented in a way that when the proposed algorithm starts, the resource availability at FNs and at the cloud is not $100\%$ to represent a real-life scenario. In this simulation, the distribution of resource availability follows a random distribution. As shown in Figure }\ref{fig:sim:simulationTime:for_both_algos}\revison{, the total simulation time required to assign $10000$ applications is approximately $230 sec$. However, in this simulation, we randomly modify the resource availability of FNs and cloud at an interval of $20 sec$, as shown in Figure }\ref{fig:sim:uncertaintyAnalysis}\revison{. The X-axis of all subfigures represents the time interval at which the resource availability of the entire environment is modified. }

\revison{Figure }\ref{fig:sim:uncertaintyAnalysis}(a)\revison{ shows the modified CPU and memory resource availability at different time intervals. It can be seen that the initial CPU and memory resource availability at FNs was $45\%$ and $74\%$. After $20 sec$, the CPU and memory configurations of the FNs were changed to $67\%$ and $79\%$, respectively. A similar pattern can be seen in Figure }\ref{fig:sim:uncertaintyAnalysis}(b)\revison{, where bandwidth among FNs and FCIs (Blue solid line) and bandwidth among FCI and cloud (Red solid line) resource availability is modified. Bandwidth availability among FNs and FCI ranges between $35\%$ and $89\%$ across the entire simulation. On the other hand, bandwidth among FCIs and cloud ranges between $45\%$ and $70\%$. As shown in Figure }\ref{fig:sim:uncertaintyAnalysis}(c)\revison{, the latency value, calculated in milliseconds, among FNs and FCI is modified at every $20 sec$, and the value varies between $70 ms$ and $80 ms$. Similarly, the latency among FCIs and cloud, at every $20 sec$, ranges between $140 ms$ and $150 ms$.}  

\pgfplotstableread[row sep=\\,col sep=&]{
timeInstance &	cpuAvailableAtFog &	memAvailableAtFog\\
0	&	44.6	&	74.83	\\
20	&	66.99	&	79.21	\\
40	&	27.87	&	58.38	\\
60	&	43.52	&	67.83	\\
80	&	52.11	&	82.86	\\
100	&	42.63	&	32.19	\\
120	&	49.26	&	55.32	\\
140	&	73.92	&	57.81	\\
160	&	41.25	&	56.62	\\
180	&	55.73	&	62.91	\\
200	&	54.85	&	67.59	\\
}\uncertaintyAnalysisCpuMem

\pgfplotstableread[row sep=\\,col sep=&]{
timeInstance &	BWAmongFNsFCI &	BWAmongFCICloud & LatencyAmongFNsFCI &	LatencyAmongFCICloud\\
0	&	63.37	&	59.67	&	76.26	&	143.49	\\
20	&	35.53	&	65.79	&	74.82	&	144.74	\\
40	&	67.28	&	46.36	&	73.49	&	149.05	\\
60	&	68.27	&	55.77	&	76.31	&	145.56	\\
80	&	42.96	&	63.64	&	70.88	&	145.53	\\
100	&	84.76	&	67.61	&	79.41	&	146.57	\\
120	&	88.81	&	58.41	&	76.7	&	144.08	\\
140	&	79.78	&	58.98	&	70.34	&	144.16	\\
160	&	61.95	&	71.4	&	72.59	&	144.97	\\
180	&	66.3	&	45.17	&	71.23	&	144.58	\\
200	&	74.29	&	50.43	&	77.58	&	150.37	\\
}\uncertaintyAnalysisNetwork

\pgfplotstableread[row sep=\\,col sep=&]{
timeInstance &	cpuUtilAtFog &	memUtilAtFog\\
0	&	30.45	&	23.14	\\
20	&	33.44	&	24.64	\\
40	&	37.38	&	30.6	\\
60	&	41.85	&	38.72	\\
80	&	47.32	&	44.5	\\
100	&	48.16	&	45.76	\\
120	&	54.6	&	51.04	\\
140	&	55.68	&	54.4	\\
160	&	56.12	&	55.18	\\
180	&	57	&	58.62	\\
200	&	57.74	&	59.74	\\
}\uncertaintyCpuMemUtil

\pgfplotstableread[row sep=\\,col sep=&]{
timeInstance &	BWUtilAmongFNsFCI &	BWUtilAmongFCICloud & LatencyUtilAmongFNsFCI &	LatencyUtilAmongFCICloud\\
0	&	42	&	23	&	155	&	1568	\\
20	&	45	&	23	&	162	&	1588	\\
40	&	46	&	23	&	184	&	1594	\\
60	&	46	&	25	&	189	&	2022	\\
80	&	49	&	25	&	198	&	2144	\\
100	&	49	&	26	&	199	&	2305	\\
120	&	57	&	27	&	215	&	2349	\\
140	&	61	&	27	&	217	&	2349	\\
160	&	63	&	28	&	277	&	2476	\\
180	&	64	&	29	&	292	&	2558	\\
200	&	67	&	30	&	299	&	2576	\\
}\uncertaintyAnalysisNetworkUtil
\begin{figure}[t]
	\centering
\pgfplotsset{    
        every axis title/.style={below right,at={(0.37,-0.2)}},
        symbolic x coords={0, 20, 40, 60, 80, 100, 120, 140, 160, 180, 200},
        ymin=0,ymax=100,  xmin=0,xmax=200,
        title style={at={(0.5,-0.22)},anchor=north},
        width=0.5\textwidth,
        height=.4\textwidth,
        nodes near coords align={vertical},
        y label style={at={(-0.1,0.5)}},
        legend pos=south west,
        legend columns=2,
        xlabel={Time (in sec)},
    }
\subfloat{\resizebox{0.32\textwidth}{40mm}{
    \begin{tikzpicture}
        \begin{axis}[
                ylabel={CPU-mem availability (in \%)},
                title = (a),
            ]
            \addplot table[x=timeInstance,y=cpuAvailableAtFog]{\uncertaintyAnalysisCpuMem};
            \addplot table[x=timeInstance,y=memAvailableAtFog] {\uncertaintyAnalysisCpuMem};
            \legend{CPU, Memory}
        \end{axis}
    \end{tikzpicture} \label{fig:sim:uncertaintyAnalysisCpuMem}}}
\subfloat{
    \resizebox{0.32\textwidth}{40mm}{
    \begin{tikzpicture}
        \begin{axis}[
                ylabel={Bandwidth availability (\%)},
                title = (b),
            ]
            \addplot table[x=timeInstance,y=BWAmongFNsFCI]{\uncertaintyAnalysisNetwork};
            \addplot table[x=timeInstance,y=BWAmongFCICloud] {\uncertaintyAnalysisNetwork};
            \legend{FNs-FCI, FCI-Cloud}
        \end{axis}
    \end{tikzpicture}\label{fig:sim:uncertaintyAnalysisBW}} }
\subfloat{\resizebox{0.32\textwidth}{40mm}{
    \begin{tikzpicture}
        \begin{axis}[
                ylabel={Latency (in ms)},
                ymin=0,ymax=200,
                title = (c),
            ]
             \addplot table[x=timeInstance,y=LatencyAmongFNsFCI]{\uncertaintyAnalysisNetwork};
            \addplot table[x=timeInstance,y=LatencyAmongFCICloud] {\uncertaintyAnalysisNetwork};
            \legend{FNs-FCI, FCI-Cloud}
        \end{axis}
    \end{tikzpicture}\label{fig:sim:uncertaintyAnalysisNetworkLat}} }
    \\ 
\pgfplotsset{    
        every axis title/.style={below right,at={(0.37,-0.2)}},
        symbolic x coords={0, 20, 40, 60, 80, 100, 120, 140, 160, 180, 200},
        ymin=0,ymax=100,  xmin=0,xmax=200,
        title style={at={(0.5,-0.22)},anchor=north},
        width=0.5\textwidth,
        height=.4\textwidth,
        nodes near coords align={vertical},
        y label style={at={(-0.1,0.5)}},
        legend pos=south west,
        legend columns=2,
        xlabel={Time (in sec)},
    }
\subfloat{\resizebox{0.32\textwidth}{40mm}{
    \begin{tikzpicture}
        \begin{axis}[
                ylabel={CPU-mem. utilization (in \%)},
                title = (d),
            ]
            \addplot table[x=timeInstance,y=cpuUtilAtFog]{\uncertaintyCpuMemUtil};
            \addplot table[x=timeInstance,y=memUtilAtFog] {\uncertaintyCpuMemUtil};
            \legend{CPU, Memory}
        \end{axis}
    \end{tikzpicture} \label{fig:sim:uncertaintyCpuMemUtil}}}
\subfloat{
    \resizebox{0.32\textwidth}{40mm}{
    \begin{tikzpicture}
        \begin{axis}[
                ylabel={Bandwidth Utilization (\%)},
                title = (e),
            ]
            \addplot table[x=timeInstance,y=BWUtilAmongFNsFCI]{\uncertaintyAnalysisNetworkUtil};
            \addplot table[x=timeInstance,y=BWUtilAmongFCICloud] {\uncertaintyAnalysisNetworkUtil};
            \legend{FNs-FCI, FCI-Cloud}
        \end{axis}
    \end{tikzpicture}\label{fig:sim:uncertaintyAnalysisBWUtil}} }
\subfloat{\resizebox{0.32\textwidth}{40mm}{
    \begin{tikzpicture}
        \begin{axis}[
                ylabel={Average latency (in ms)},
                ymin=100,ymax=5400,
                title = (f),
                scaled y ticks=base 10:-2,
                legend style={at={(0.1,0.25)}}
            ]
            \addplot table[x=timeInstance,y=LatencyUtilAmongFNsFCI]{\uncertaintyAnalysisNetworkUtil};
            \addplot table[x=timeInstance,y=LatencyUtilAmongFCICloud] {\uncertaintyAnalysisNetworkUtil};
            \legend{FNs-FCI, FCI-Cloud}
        \end{axis}
    \end{tikzpicture}\label{fig:sim:uncertaintyAnalysisLatUtil}} }
    \caption{Comparison of only CLOUD resource utilization (a) CPU and memory availability, (b) Bandwidth availability, (c)  Latency (in ms), (d) CPU-memory utilization (in \%), (e) Bandwidth Utilization (\%), (f) Average latency of tasks (in ms).}	\label{fig:sim:uncertaintyAnalysis} 
\end{figure}
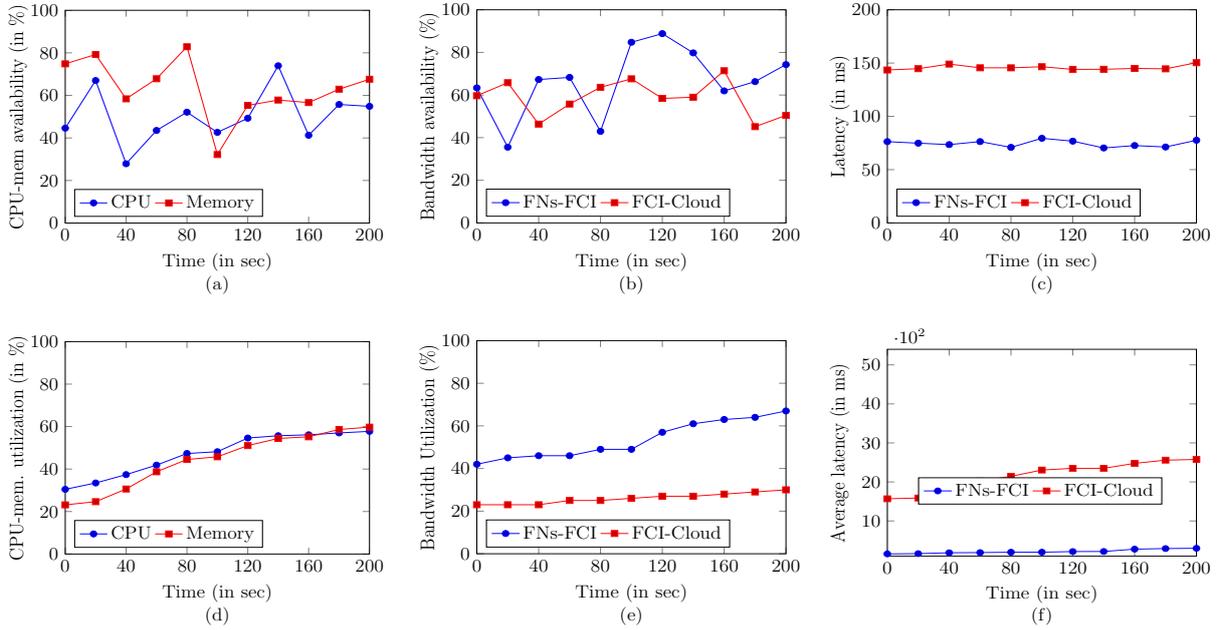

\revison{Under the above uncertain resource availability, Figures }\ref{fig:sim:uncertaintyAnalysis}(d)\revison{ - }\ref{fig:sim:uncertaintyAnalysis}(f)\revison{ represent the impact on resource utilization. It is observed that such uncertain conditions and fluctuated resource availability have a higher negative impact on utilization. Figure }\ref{fig:sim:uncertaintyAnalysis}(d)\revison{ shows the CPU and memory utilization of the FNs. In comparison with the results in Figure }\ref{fig:sim:compCPUUtil_F}\revison{ and }\ref{fig:sim:compMemUtil_F}\revison{, it is observed that the CPU and memory resources are being underutilized by approximately $10\%$. This can be observed in the figure where the maximum CPU and memory utilization are less than $60\%$ each. At the same time, this utilization is observed to be more than $65\%$, as in Figure }\ref{fig:sim:compCPUUtil_F}\revison{ and }\ref{fig:sim:compMemUtil_F}.

\revison{A similar pattern is observed in the case of network bandwidth utilization. The maximum bandwidth utilization among FNs and among FCI and cloud is observed as $78\%$ and $40\%$, in Figure }\ref{fig:sim:compBWUtil_F}\revison{ and }\ref{fig:sim:compBWUtil_C}\revison{, respectively. However, this utilization is reduced by more than $10\%$ when the bandwidth resource availability fluctuates, as shown in Figure }\ref{fig:sim:uncertaintyAnalysis}(e)\revison{. The bandwidth utilization is less than $67\%$ among FNs and FCI, and the utilization among FCI and cloud is less than $30\%$. The uncertain latency among FNs, FCI, and cloud is also analyzed and presented the results in Figure }\ref{fig:sim:uncertaintyAnalysis}(f)\revison{. The latency of the network communications among FNs and FCIs increases from $155 ms$ to $300 ms$ while processing all the applications by HeRAFC algorithm, which is more than $30\%$ increase from the previous result, shown in Figure }\ref{fig:sim:Ltncy_F}\revison{. Similarly, the minimum latency among FCIs and cloud is $1570\%$, which is more than approximately $30\%$ of the previous result, as shown in Figure }\ref{fig:sim:Ltncy_C}\revison{. Overall this observation indicates that latency is the worst-performing parameter among others when the network communication becomes very unstable and fluctuates over time.}

\section{Simulation time analysis}\label{apndx:simTimeAna}
\revison{Figure }\ref{fig:sim:simulationTime}\revison{ helps in visualizing the performance of Algorithm }\ref{algo:HeRAFC:tskOrdr}\revison{ and }\ref{algo:HeRAFC}\revison{ in terms of simulation time taken for different numbers of applications, ranging from $5000$ to $10000$. The number of tasks ranges from $4$ to $12$ per application. However, the maximum number of tasks is set to $100$ thousand. The x-axis represents the number of applications, while the Y-axis in all the sub-figures shows the simulation time. The results suggest that the total simulation time for both algorithms increases as the number of applications increases. Figure }\ref{fig:sim:simulationTime}(a)\revison{ represents the time taken by Algorithm }\ref{algo:HeRAFC:tskOrdr}\revison{ and }\ref{algo:HeRAFC}\revison{ to process a certain number of applications. For instance, as in Figure }\ref{fig:sim:simulationTime:totaltime}\revison{, when there were $5000$ applications, Algorithm }\ref{algo:HeRAFC:tskOrdr}\revison{ takes a total of approximately $5$ seconds, while Algorithm }\ref{algo:HeRAFC}\revison{ takes more than $100$ seconds. Similarly, when the number of applications increases to  $10000$, Algorithm }\ref{algo:HeRAFC:tskOrdr}\revison{ takes approx. $10$ seconds, while Algorithm }\ref{algo:HeRAFC:tskOrdr}\revison{ takes more than $220$ seconds. It is observed (also expected) the total simulation time taken by Algorithm }\ref{algo:HeRAFC}\revison{ is consistently higher than that of Algorithm }\ref{algo:HeRAFC:tskOrdr}\revison{, across all tested scenarios.}   

\pgfplotstableread[row sep=\\,col sep=&]{
noofApp &	  total_algo1 &	  total_algo2 &	  avg_algo1 &	  avg_algo2 &	  total_algo1_algo2 &	  avg_algo1_algo2\\
5 &	  5.036248836 &	  109.744 &	  1.52 &	  20.29 &	  114.7802488 &	  21.81\\
5.5 &	  5.523605199 &	  120.782 &	  0.72 &	  22.4 &	  126.3056052 &	  23.12\\
6 &	  6.00564462 &	  131.454 &	  0.83 &	  20.26 &	  137.4596446 &	  21.09\\
6.5 &	  6.522686131 &	  142.388 &	  1.61 &	  22.83 &	  148.9106861 &	  24.44\\
7 &	  7.014720569 &	  153.29 &	  0.76 &	  20.65 &	  160.3047206 &	  21.41\\
7.5 &	  7.521068268 &	  164.203 &	  0.84 &	  20.4 &	  171.7240683 &	  21.24\\
8 &	  8.012624086 &	  175.164 &	  0.62 &	  21.07 &	  183.1766241 &	  21.69\\
8.5 &	  8.511775842 &	  186.118 &	  1.08 &	  21.85 &	  194.6297758 &	  23.83\\
9 &	  8.998254172 &	  196.953 &	  1.02 &	  21.9 &	  205.9512542 &	  23.66\\
9.5 &	  9.504412117 &	  208.162 &	  1.18 &	  22 &	  217.6664121 &	  23.18\\
10 &	  10.01442243 &	  218.831 &	  1.25 &	  21.9 &	  228.8454224 &	  23.15\\
}\simulationTime
\begin{figure}[t]
	\centering
\subfloat{\resizebox{0.32\textwidth}{40mm}{
    \begin{tikzpicture}
    \pgfplotsset{every axis/.append style={font=\footnotesize}}
        \begin{axis}[                  
                width=0.5\textwidth,
                height=.4\textwidth,                
                 ymin=0,ymax=250,
                xtick={4.5, 5, 5.5, 6, 6.5, 7, 7.5, 8, 8.5, 9, 9.5, 10, 10.5},
                xticklabels={, 5, 5.5, 6, 6.5, 7, 7.5, 8, 8.5, 9, 9.5, 10, },
                x tick label style={rotate=90,anchor=east},
                axis y line*=left,
                ylabel={By Algo 1 (in sec)},                
                xlabel={No. of apps. (in thousands)},
                title style={at={(0.5,-0.3)},anchor=north,yshift=-0.1},
                title = (a) Total simulation time by Algo 1 and 2,
            ]
            \addplot[smooth,mark=o,red] table[x=noofApp,y=total_algo1]{\simulationTime}; 
            \label{simulationTime_total_algo1} 
        \end{axis}
        \begin{axis}[
              ybar=0.3, 
              yticklabel pos=right,              
              axis y line*=right,
              width=0.5\textwidth,
              height=.4\textwidth,
              axis x line=none,
              ymin=0, ymax=250,
              xtick={4.5, 5, 5.5, 6, 6.5, 7, 7.5, 8, 8.5, 9, 9.5, 10, 10.5},
              xticklabels={, 5, 5.5, 6, 6.5, 7, 7.5, 8, 8.5, 9, 9.5, 10, },
              x tick label style={rotate=90,anchor=east},
              ylabel={By Algo 2 (in sec)},
              legend columns=2,
              legend style={at={(0.6,0.97)}, draw=none,  /tikz/every even column/.append style={column sep=5mm}},
        ]
        \addlegendimage{/pgfplots/refstyle=simulationTime_total_algo1}\addlegendentry{Algo 1}
        \addplot[fill opacity=0.5, color=blue, pattern=crosshatch dots, pattern color=blue] table[x=noofApp,y=total_algo2]{\simulationTime};  
        \addlegendentry{Algo 2}
        \end{axis}
    \end{tikzpicture} \label{fig:sim:simulationTime:totaltime}}}
\subfloat{
    \resizebox{0.32\textwidth}{40mm}{
    \begin{tikzpicture}
    \pgfplotsset{every axis/.append style={font=\footnotesize}}
        \begin{axis}[              
                width=0.5\textwidth,
                height=.4\textwidth,                
                ymin=0,ymax=30,
                xtick={4.5, 5, 5.5, 6, 6.5, 7, 7.5, 8, 8.5, 9, 9.5, 10, 10.5},
                xticklabels={, 5, 5.5, 6, 6.5, 7, 7.5, 8, 8.5, 9, 9.5, 10, },
                x tick label style={rotate=90,anchor=east},
                axis y line*=left,
                ylabel={By Algo 1 (in ms)},                
                xlabel={No. of apps. (in thousands)},
                title style={at={(0.5,-0.3)},anchor=north,yshift=-0.1},
                title = (b) Average simulation time by Algo 1 and 2,                
            ]
            \addplot[smooth,mark=o,red] table[x=noofApp,y=avg_algo1]{\simulationTime}; 
            \label{simulationTime_avg_algo1} 
        \end{axis}
        \begin{axis}[
              ybar=0.3,            
              yticklabel pos=right,              
              axis y line*=right,
              width=0.5\textwidth,
              height=.4\textwidth,
              axis x line=none,
              ymin=0.00, ymax=30.00,
              xtick={4.5, 5, 5.5, 6, 6.5, 7, 7.5, 8, 8.5, 9, 9.5, 10, 10.5},
                xticklabels={, 5, 5.5, 6, 6.5, 7, 7.5, 8, 8.5, 9, 9.5, 10, },
                x tick label style={rotate=90,anchor=east},
              ylabel={By Algo 2 (in ms)},                
              legend columns=2,
              legend style={at={(0.84,0.97)}, draw=none, /tikz/every even column/.append style={column sep=5mm}},
        ]
        \addlegendimage{/pgfplots/refstyle=simulationTime_avg_algo1}\addlegendentry{Algo 1}
        \addplot[fill opacity=0.5, color=blue, pattern=crosshatch dots, pattern color=blue] table[x=noofApp,y=avg_algo2]{\simulationTime};  
        \addlegendentry{Algo 2}
        \end{axis}        
    \end{tikzpicture} \label{fig:sim:simulationTime:avgtime}} }
\subfloat{\resizebox{0.32\textwidth}{40mm}{
    \begin{tikzpicture}
    \pgfplotsset{every axis/.append style={font=\footnotesize}}
        \begin{axis}[              
                width=0.5\textwidth,
                height=.4\textwidth,                
                ymin=0, ymax=250,
                xtick={4.5, 5, 5.5, 6, 6.5, 7, 7.5, 8, 8.5, 9, 9.5, 10, 10.5},
                xticklabels={, 5, 5.5, 6, 6.5, 7, 7.5, 8, 8.5, 9, 9.5, 10, },
                x tick label style={rotate=90,anchor=east},
                axis y line*=left,
                ylabel={Total time (in sec)},                
                xlabel={No. of apps. (in thousands)},
                title style={at={(0.5,-0.3)},anchor=north,yshift=-0.1},
                title = (c) Simulation time by both algorithms,
            ]
            \addplot[smooth,mark=o,red] table[x=noofApp,y=avg_algo1_algo2]{\simulationTime}; 
            \label{simulationTime_Average} 
        \end{axis}
        \begin{axis}[
              ybar=0.3,            
              yticklabel pos=right,              
              axis y line*=right,
              width=0.5\textwidth,
              height=.4\textwidth,
              axis x line=none,
              ymin=0, ymax=250,
              xtick={4.5, 5, 5.5, 6, 6.5, 7, 7.5, 8, 8.5, 9, 9.5, 10, 10.5},
              xticklabels={, 5, 5.5, 6, 6.5, 7, 7.5, 8, 8.5, 9, 9.5, 10, },
              x tick label style={rotate=90,anchor=east},
              ylabel={Avg. time (in ms)},
              legend columns=2,
              legend style={at={(0.7,0.97)}, draw=none, /tikz/every even column/.append style={column sep=5mm}},
        ]
        \addlegendimage{/pgfplots/refstyle=simulationTime_Average}\addlegendentry{Average}
        \addplot[fill opacity=0.5, color=blue, pattern=crosshatch dots, pattern color=blue] table[x=noofApp,y=total_algo1_algo2]{\simulationTime};  
        \addlegendentry{Total}
        \end{axis} 
    \end{tikzpicture}\label{fig:sim:simulationTime:for_both_algos}} } 
    \caption{Simulation time by Algorithm 1, Algorithm 2 and by both combined.}
	\label{fig:sim:simulationTime} 
\end{figure}
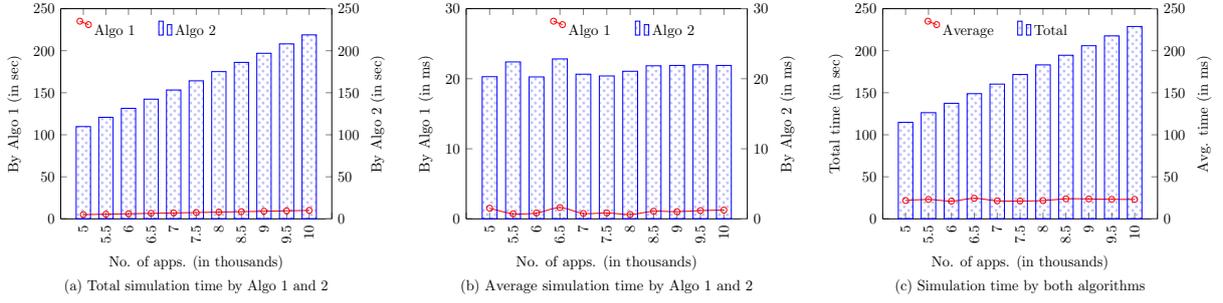

\revison{Similar to the above results, the average time to handle an application by Algorithm }\ref{algo:HeRAFC:tskOrdr}\revison{ and }\ref{algo:HeRAFC}\revison{ is also observed and presented in Figure }\ref{fig:sim:simulationTime}(b)\revison{. Algorithm }\ref{algo:HeRAFC:tskOrdr}\revison{ takes an average of $1.5ms$ when the number of applications is $5000$, whereas Algorithm }\ref{algo:HeRAFC}\revison{ takes an average of more than $20ms$ for the same number of applications, as shown in Figure }\ref{fig:sim:simulationTime}(b)\revison{. Algorithm }\ref{algo:HeRAFC:tskOrdr}\revison{ and }\ref{algo:HeRAFC}\revison{ took an average of $1.25ms$ and $21.9ms$ per application, respectively, when the number of applications increases to $10000$. 
Unlike Figure }\ref{fig:sim:simulationTime}(a)\revison{, presenting the time taken by individuals algorithms, Figure }\ref{fig:sim:simulationTime}(c)\revison{ presents the average and total time taken by both the Algorithms }\ref{algo:HeRAFC:tskOrdr}\revison{ and }\ref{algo:HeRAFC}\revison{. Figure }\ref{fig:sim:simulationTime}(c)\revison{ shows that approximately $114 sec$ is required to process $5000$ applications, with an average of $21ms$ per application. When the number of applications increases to $100000$, a total of $228sec$ (with an average of $23ms$ per application) is required to assign those applications to fog and cloud environments, as shown in Figure }\ref{fig:sim:simulationTime}(c)\revison{. Overall, the Figure }\ref{fig:sim:simulationTime}\revison{ provides a comparison of the simulation time performance of two algorithms for different application scenarios.}

\end{document}